\documentclass[10pt, letterpaper]{article}

\usepackage[hidelinks]{hyperref}
\usepackage[affil-it]{authblk}

\usepackage{booktabs} 
\usepackage{multirow}

\usepackage{tikz}
\usetikzlibrary{arrows,automata,positioning,decorations.pathreplacing} 
\usepackage{color}
\usepackage[english]{babel}
\usepackage[utf8]{inputenc}
\usepackage{amsmath}
\usepackage{amssymb,pifont}
\usepackage{amsthm}
\usepackage{bbm}
\usepackage{graphicx}
\usepackage{float}
\usepackage{wrapfig}
\usepackage{subcaption,paralist,pdflscape,afterpage}
\usepackage{enumitem}
\usepackage{dblfloatfix}
\usepackage{balance}

\usepackage{siunitx}
\newtheorem{theorem}{Theorem}

\newtheorem{lma}{Lemma}

\newtheorem{definition}{Definition}

\newtheorem{assumption}{Assumption}

\newcommand{\Ex}[1]{\mathop{\mathbb{E\/}}{\left[#1\right]}}
\newcommand{\Pb}[1]{\mathop{\mathbb{P\/}}{\left[#1\right]}}

\usepackage{color,xspace}

\definecolor{darkgreen}{RGB}{32, 192, 32}
\definecolor{darkyellow}{RGB}{192, 140, 37}
\definecolor{darkred}{RGB}{192, 0, 0}

\newcommand{\cmark}{\textcolor{darkgreen}{\ding{51}}\xspace}
\newcommand{\xmark}{\textcolor{darkred}{\ding{55}}\xspace}

\makeatletter
\newcommand\footnoteref[1]{\protected@xdef\@thefnmark{\ref{#1}}\@footnotemark}
\makeatother

\tikzset{
  every overlay node/.style={
    draw=none,fill=white,rounded corners,anchor=south west,
  },
}
\def\tikzoverlay{%
  \tikz[baseline,overlay]\node[every overlay node]
}%

\newcommand{\figFooMcf}{%
  \begin{table}[h]
    \centering
     \begin{minipage}{0.2\textwidth}
     \resizebox{\textwidth}{!}{%
     \begin{tikzpicture}[->,>=stealth',shorten >=1pt,auto,node distance=1.8cm,semithick,every node/.style={inner sep=0,outer sep=0.01mm,minimum size=0.3mm}]
       \node[state]                               (i) {$i$};
       \node[below=0mm of i,yshift=-1mm]                      (id) {\large$\beta_i$};
       \node[state,right=of i]                    (j) {$j$};
       \node[below=0mm of j,yshift=-1mm]                      (id) {\large$\beta_j$};
       \path
       (i)   edge node {\large$(cap,cost)$}   (j);
     \end{tikzpicture}}
 \end{minipage}
 \hspace{4mm}
 \small
     \begin{tabular}{cl}
       \toprule
       $cap$ & Capacity of edge ($i,j$), i.e., maximum flow along this edge. \\
       $cost$ & Cost per unit flow on edge ($i,j$). \\
       $\beta_i$ & Flow surplus at node $i$, if $\beta_i>0$, flow demand if $\beta_i<0$. \\
       \bottomrule
     \end{tabular}
     \vspace{3mm}
     \caption{Notation for FOO's min-cost flow. Edges are labeled with their $(capacity, cost)$, and nodes are labeled with flow $-demand$ or $+surplus$.}
     \vspace{-4mm}
     \label{tbl:mcf:notation}
   \end{table}
   \begin{figure}[h]
     \vspace{-4mm}
     \center
    \resizebox{.99\textwidth}{!}{%
      \begin{tikzpicture}[->,>=stealth',shorten >=1pt,auto,node distance=.8cm,semithick,every node/.style={inner sep=0,outer sep=0.01mm,minimum size=0.3mm}]
        \node[state,fill=darkyellow]               (a) {\LARGE\bf\textcolor{white}{a}};
        \node[below=1mm of a]                      (ad) {\textcolor{darkyellow}{$+3$}};
        \node[state,right=of a,fill=darkred]                    (b) {\LARGE\bf\textcolor{white}{b}};
        \node[below=1mm of b]                      (bd) {\textcolor{darkred}{$+1$}};
        \node[state,right=of b,fill=blue]                    (c) {\LARGE\bf\textcolor{white}{c}};
        \node[below=1mm of c]                      (cd) {\textcolor{blue}{$+1$}};
        \node[state,right=of c,fill=darkred]                    (b2) {\LARGE\bf\textcolor{white}{b}};
        \node[state,right=of b2,fill=darkgreen]                   (d) {\LARGE\bf\textcolor{white}{d}};
        \node[below=2mm of d]                      (dd) {\textcolor{darkgreen}{$+2$}};
        \node[state,right=of d,fill=darkyellow]                    (a2) {\LARGE\bf\textcolor{white}{a}};
        \node[state,right=of a2,fill=blue]                    (c2) {\LARGE\bf\textcolor{white}{c}};
        \node[below=.5mm of c2]                      (cd) {\textcolor{blue}{$-1$}};
        \node[state,right=of c2,fill=darkgreen]                    (d2) {\LARGE\bf\textcolor{white}{d}};
        \node[below=1mm of d2]                      (dd) {\textcolor{darkgreen}{$-2$}};
        \node[state,right=of d2,fill=darkyellow]                    (a3) {\LARGE\bf\textcolor{white}{a}};
        \node[state,right=of a3,fill=darkred]                    (b3) {\LARGE\bf\textcolor{white}{b}};
        \node[state,right=of b3,fill=darkred]                    (b4) {\LARGE\bf\textcolor{white}{b}};
        \node[below=1mm of b4]                      (bd) {\textcolor{darkred}{$-1$}};
        \node[state,right=of b4,fill=darkyellow]                    (a4) {\LARGE\bf\textcolor{white}{a}};
        \node[below=1mm of a4]                      (ad) {\textcolor{darkyellow}{$-3$}};
        \path[line width=0.5mm, color=black]
        (a)   edge node[yshift=.3mm] {$(3,0)$}   (b)
        (b)   edge node[yshift=.3mm] {$(3,0)$}   (c)
        (c)   edge node[yshift=.3mm] {$(3,0)$}   (b2)
        (b2)   edge node[yshift=.3mm] {$(3,0)$}   (d)
        (d)   edge node[yshift=.3mm] {$(3,0)$}   (a2)
        (a2)   edge node[yshift=.3mm] {$(3,0)$}   (c2)
        (c2)   edge node[yshift=.3mm] {$(3,0)$}   (d2)
        (d2)   edge node[yshift=.3mm] {$(3,0)$}   (a3)
        (a3)   edge node[yshift=.3mm] {$(3,0)$}   (b3)
        (b3)   edge node[yshift=.3mm] {$(3,0)$}   (b4)
        (b4)   edge node[yshift=.3mm] {$(3,0)$}   (a4);
        \path[line width=0.75mm, color=darkyellow]
        (a.north)   edge[bend left=10] node[yshift=.3mm] {$(3,1/3)$}   (a2.north)
        (a2.north)   edge[bend left=12] node[yshift=.3mm] {$(3,1/3)$}   (a3.north)
        (a3.north)   edge[bend left=12] node[yshift=.3mm,xshift=0mm] {$(3,1/3)$}   (a4.north);
        \path[->,line width=0.25mm, color=darkred]
        (b)   edge[bend right=45] node[yshift=-4mm] {$(1,1)$}   (b2)
        (b3)   edge[bend right=40] node[yshift=-4mm,xshift=0mm] {$(1,1)$}   (b4);
        \path[->,line width=0.25mm, color=darkred]
        (b2)   edge[bend right=18] node[yshift=-4mm] {$(1,1)$}   (b3);
        \path[line width=0.25mm, color=blue]
        (c.north west)   edge[bend left=16] node[xshift=6mm,yshift=.3mm] {$(1,1)$}   (c2.north west);
        \path[line width=0.5mm, color=darkgreen]
        (d)   edge[bend right=32] node[xshift=-1mm,yshift=0.3mm] {$(2,1/2)$}   (d2);
      \end{tikzpicture}}
      \caption{FOO's min-cost flow problem for the short trace in \autoref{fig:foo:trace}, using the notation from \autoref{tbl:mcf:notation}.
        Nodes represent requests, and cost measures cache misses.
        Requests are connected by central edges with capacity equal to the cache capacity
        and cost zero---%
        flow routed along this path represents cached objects (hits).
        Outer edges connect requests to the same object,
        with capacity equal to the object's size---%
        flow routed along this path represents misses.
        The first request for each object is a source of flow equal to the object's size,
        and the last request is a sink of flow of the same amount.
        Outer edges' costs are inversely proportional to object size
        so that they cost 1 miss when an entire object is not cached.
        The min-cost flow achieves the fewest misses.}
      \label{fig:foo:mcf}
    \end{figure}
  }

\newcommand{\figFooMcfsmall}{%
  \begin{minipage}{1.01\textwidth}
    \resizebox{\textwidth}{!}{%
      \begin{tikzpicture}[->,>=stealth',shorten >=1pt,auto,node distance=.8cm,semithick,every node/.style={inner sep=0,outer sep=0.01mm,minimum size=0.3mm}]
        \node[state,fill=darkyellow]               (a) {\LARGE\bf\textcolor{white}{a}};
        \node[below=1.0mm of a]                      (s0) {\textcolor{darkyellow}{}};
        \node[state,right=of a,fill=darkred]                    (b) {\LARGE\bf\textcolor{white}{b}};
        \node[state,right=of b,fill=blue]                    (c) {\LARGE\bf\textcolor{white}{c}};
        \node[below=8mm of c]                      (s1) {\textcolor{darkyellow}{}};
        \node[state,right=of c,fill=darkred]                    (b2) {\LARGE\bf\textcolor{white}{b}};
        \node[state,right=of b2,fill=darkgreen]                   (d) {\LARGE\bf\textcolor{white}{d}};
        \node[below=1.0mm of d]                      (s2) {\textcolor{darkyellow}{}};
        \node[state,right=of d,fill=darkyellow]                    (a2) {\LARGE\bf\textcolor{white}{a}};
        \node[state,right=of a2,fill=blue]                    (c2) {\LARGE\bf\textcolor{white}{c}};
        \node[below=8mm of c2]                      (s3) {\textcolor{darkyellow}{}};
        \node[state,right=of c2,fill=darkgreen]                    (d2) {\LARGE\bf\textcolor{white}{d}};
        \node[state,right=of d2,fill=darkyellow]                    (a3) {\LARGE\bf\textcolor{white}{a}};
        \node[below=1.0mm of a3]                      (s4) {\textcolor{darkyellow}{}};
        \node[state,right=of a3,fill=darkred]                    (b3) {\LARGE\bf\textcolor{white}{b}};
        \node[state,right=of b3,fill=darkred]                    (b4) {\LARGE\bf\textcolor{white}{b}};
        \node[below=8mm of b4]                      (s5) {\textcolor{darkyellow}{}};
        \node[state,right=of b4,fill=darkyellow]                    (a4) {\LARGE\bf\textcolor{white}{a}};
        \node[below=1.0mm of a4]                      (s6) {\textcolor{darkyellow}{}};
        \path[line width=0.5mm, color=black]
        (a)   edge node[yshift=.3mm] {}   (b)
        (b)   edge node[yshift=.3mm] {}   (c)
        (c)   edge node[yshift=.3mm] {}   (b2)
        (b2)   edge node[yshift=.3mm] {}   (d)
        (d)   edge node[yshift=.3mm] {}   (a2)
        (a2)   edge node[yshift=.3mm] {}   (c2)
        (c2)   edge node[yshift=.3mm] {}   (d2)
        (d2)   edge node[yshift=.3mm] {}   (a3)
        (a3)   edge node[yshift=.3mm] {}   (b3)
        (b3)   edge node[yshift=.3mm] {}   (b4)
        (b4)   edge node[yshift=.3mm] {}   (a4);
        \path[line width=0.5mm, color=darkyellow]
        (a.north)   edge[bend left=10] node[yshift=.3mm] {}   (a2.north)
        (a2.north)   edge[bend left=12] node[yshift=.3mm] {}   (a3.north)
        (a3.north)   edge[bend left=12] node[yshift=.3mm,xshift=0mm] {}   (a4.north);
        \path[->,line width=0.5mm, color=darkred]
        (b)   edge[bend right=25] node[yshift=-4mm] {}   (b2)
        (b3)   edge[bend right=40] node[yshift=-4mm,xshift=0mm] {}   (b4);
        \path[->,line width=0.5mm, color=darkred]
        (b2)   edge[bend right=12] node[yshift=-4mm] {}   (b3);
        \path[line width=0.5mm, color=blue]
        (c.north west)   edge[bend left=16] node[xshift=6mm,yshift=.3mm] {}   (c2.north west);
        \path[line width=0.5mm, color=darkgreen]
        (d)   edge[bend right=19] node[xshift=-1mm,yshift=0.3mm] {}   (d2);
        \draw[-, line width=0.5mm, decorate, decoration={brace,amplitude=3mm}]
        (s2) -- node[below=2.8mm] {\fontsize{13}{13} \textbf{Segment 1}}  (s0);
        \draw[-, line width=0.5mm,  decorate, decoration={brace,amplitude=3mm}]
        (s3) -- node[below=2.8mm] {\fontsize{13}{13} \textbf{Segment 2}}  (s1);
        \draw[-, line width=0.5mm,  decorate, decoration={brace,amplitude=3mm}]
        (s4) -- node[below=2.8mm] {\fontsize{13}{13} \textbf{Segment 3}}  (s2);
        \draw[-, line width=0.5mm,  decorate, decoration={brace,amplitude=3mm}]
        (s5) -- node[below=2.8mm] {\fontsize{13}{13} \textbf{Segment 4}}  (s3);
        \draw[-, line width=0.5mm,  decorate, decoration={brace,amplitude=3mm}]
        (s6) -- node[below=2.8mm] {\fontsize{13}{13} \textbf{Segment 5}}  (s4);
      \end{tikzpicture}
    }
  \end{minipage}
}

\newcommand{\pfooufigure}{%
  \begin{figure}[ht]
    \centering
\begin{subfigure}[b]{\linewidth}

  \figFooMcfsmall
  \caption{PFOO-U breaks the trace into small, overlapping segments \dots}
\end{subfigure}

\vspace{4mm}
\begin{subfigure}[b]{\linewidth}
  \scriptsize
  \resizebox{.96\linewidth}{!}{%
    \begin{minipage}{1.1\textwidth}
      
\hspace{1mm}\begin{tikzpicture}[->,>=stealth',shorten >=1pt,auto,node distance=.8cm,semithick,every node/.style={inner sep=0,outer sep=0.01mm,minimum size=0.3mm}]
    \node[state,fill=darkyellow]               (a) {\Large\bf\textcolor{white}{a}};
    \node[state,right=of a,fill=darkred]                    (b) {\Large\bf\textcolor{white}{b}};
    \node[state,right=of b,fill=blue]                    (c) {\Large\bf\textcolor{white}{c}};
    \node[state,right=of c,fill=darkred]                    (b2) {\Large\bf\textcolor{white}{b}};
    \node[state,right=of b2,xshift=3mm]                    (f) {};
    \path[line width=0.75mm]
    (a)   edge node[yshift=.3mm] {\large$(3,0)$}   (b)
    (b)   edge node[yshift=.3mm] {\large$(3,0)$}   (c)
    (c)   edge node[yshift=.3mm] {\large$(3,0)$}   (b2)
    (b2)   edge node[yshift=.3mm] {\large$(3,0)$}   (f)
    (a.north)   edge[bend left=25] node[yshift=.3mm] {\large$(3,1/3)$}   (f.north);
    \path[line width=0.25mm]
    (b)   edge[bend right=35] node[yshift=-4mm] {\large$(1,1)$}   (b2)
    (c.north)   edge[bend left=45] node[xshift=0mm,yshift=.3mm] {\large$(1,1)$}   (f.north)
    (b2)   edge[bend right=45] node[yshift=-4mm] {\large$(1,1)$}   (f);
\end{tikzpicture}
\hfill\hfill
\makebox[0pt]{\raisebox{0.35in}{\parbox{1.8in}{\normalsize \textbf{Segment 1.} Cache both \textcolor{darkred}{\bf b}s and \textcolor{blue}{\bf c}; forget \textcolor{blue}{\bf c} and second \textcolor{darkred}{\bf b}.}}}
\hfill\mbox{}

\hspace{1.13in}
\begin{tikzpicture}[->,>=stealth',shorten >=1pt,auto,node distance=.8cm,semithick,every node/.style={inner sep=0,outer sep=0.01mm,minimum size=0.3mm}]
    \node[state,right=of b,fill=blue]                    (c) {\Large\bf\textcolor{white}{c}};
    \node[state,right=of c,fill=darkred]                    (b2) {\Large\bf\textcolor{white}{b}};
    \node[state,right=of b2,fill=darkgreen]                   (d) {\Large\bf\textcolor{white}{d}};
    \node[state,right=of d,fill=darkyellow]                    (a2) {\Large\bf\textcolor{white}{a}};
    \node[state,right=of a2,xshift=3mm]                    (f) {};
    \path[line width=0.5mm]
    (c)   edge node[yshift=.3mm] {\large$(\textcolor{darkred}{2},0)$}   (b2)
    (d)   edge[bend right=32.5] node[yshift=-3.5mm] {\large$(2,1/2)$}   (f);
    \path[line width=0.75mm]
    (b2)   edge node[yshift=.3mm] {\large$(3,0)$}   (d)
    (d)   edge node[yshift=.3mm] {\large$(3,0)$}   (a2)
    (a2)   edge node[yshift=.3mm] {\large$(3,0)$}   (f)
    (a2)   edge[bend left=50] node[yshift=.3mm] {\large$(3,1/3)$}   (f.north);
    \path[line width=0.25mm]
    (b2)   edge[bend right=20] node[xshift=-5mm,yshift=-4mm] {\large$(1,1)$}   (f)
    (c.north)   edge[bend left=15] node[xshift=0mm,yshift=.3mm] {\large$(1,1)$}   (f.north);
\end{tikzpicture}
\hfill\hfill
\hspace{18mm}\makebox[0pt]{\raisebox{0.35in}{\parbox{2.8in}{\normalsize \textbf{Segment 2.} Cache \textcolor{blue}{\bf c} and \textcolor{darkred}{\bf b}.}}} 
\hfill\mbox{}

\hspace{1.29in}
\makebox[0pt]{\raisebox{0.35in}{\parbox{2.2in}{\normalsize \textbf{\mbox{Segment 3.}}\\ Cache \textcolor{blue}{\bf c}; forget \textcolor{blue}{\bf c}.}}}
\hspace{.9in}
\mbox{}
\begin{tikzpicture}[->,>=stealth',shorten >=1pt,auto,node distance=.8cm,semithick,every node/.style={inner sep=0,outer sep=0.01mm,minimum size=0.3mm}]
    \node[state,right=of b2,fill=darkgreen]                   (d) {\Large\bf\textcolor{white}{d}};
    \node[state,right=of d,fill=darkyellow]                    (a2) {\Large\bf\textcolor{white}{a}};
    \node[state,right=of a2,fill=blue]                    (c2) {\Large\bf\textcolor{white}{c}};
    \node[state,right=of c2,fill=darkgreen]                    (d2) {\Large\bf\textcolor{white}{d}};
    \node[state,right=of d2,xshift=3mm]                    (f) {};
    \path[line width=0.75mm]
    (a2.north)   edge[bend left=35] node[yshift=.3mm] {\large$(3,1/3)$}   (f.north);
    \path[line width=0.25mm]
    (d)   edge node[yshift=.3mm] {\large$(\textcolor{darkred}{1},0)$}   (a2)
    (a2)   edge node[yshift=.3mm] {\large$(\textcolor{darkred}{1},0)$}   (c2)
    (c2.north)   edge[bend left=18] node[xshift=-5mm,yshift=.3mm] {\large$(1,1)$}   (f.north);
    \path[line width=0.5mm]
    (d)   edge[bend right=25] node[yshift=-3.5mm] {\large$(2,1/2)$}   (d2)
    (c2)   edge node[yshift=.3mm] {\large$(\textcolor{darkred}{2},0)$}   (d2)
    (d2)   edge node[yshift=.3mm] {\large$(\textcolor{darkred}{2},0)$}   (f)
    (d2)   edge[bend right=25] node[yshift=-4mm] {\large$(2,1/2)$}   (f);
  \end{tikzpicture}
  
\vspace{-5mm}\hfill\raisebox{0.35in}{\Huge \textbf{$\ddots$}}\hspace{.5in}

    \end{minipage}
  }

\caption{\dots and solves min-cost flow for each segment.}
\label{fig:pfoo-u:lns}
\end{subfigure}
\caption{Starting from FOO's full formulation, PFOO-U breaks the min-cost flow problem into overlapping segments.
  Going left-to-right through the trace, PFOO-U optimally solves MCF on each segment, and updates link capacities in subsequent segments to maintain feasibility for all cached objects.
  The segments overlap to capture interactions across segment boundaries.
}
\label{fig:pfoo-u}
\end{figure}
}

\newcommand{\figTraceChars}{ %
\begin{figure*}[h!]
  \begin{subfigure}{0.24\linewidth}
    \centering
    \includegraphics[trim={2mm 0 4.1mm 0},clip,width=\linewidth]{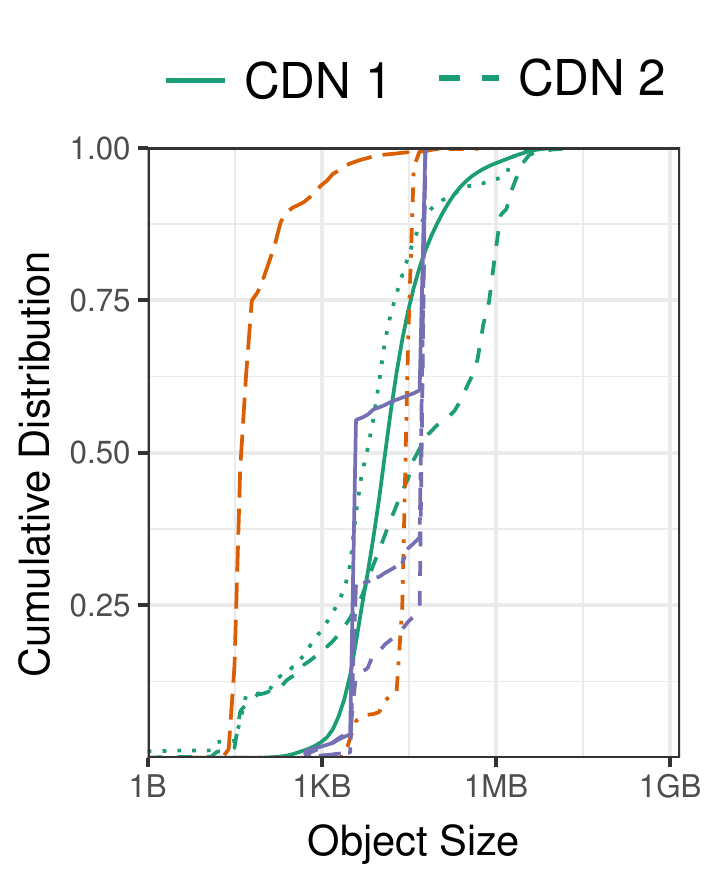}
    \caption{Object sizes.}
    \label{fig:methodology:traces:sizes}
    \end{subfigure}
  \hfill
  \begin{subfigure}{0.24\linewidth}
    \includegraphics[trim={1mm 0 4.1mm 0},clip,width=\linewidth]{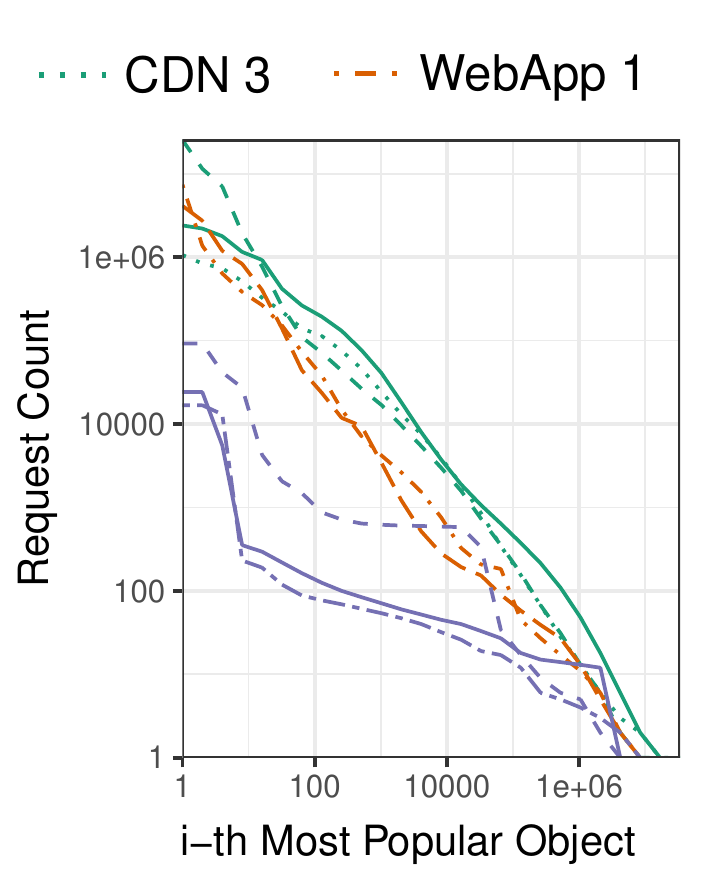}
    \caption{Popularities.}
    \label{fig:methodology:traces:popularity}
  \end{subfigure}
  \hfill
  \begin{subfigure}{0.24\linewidth}
    \includegraphics[trim={1mm 0 4.1mm 0},clip,width=\linewidth]{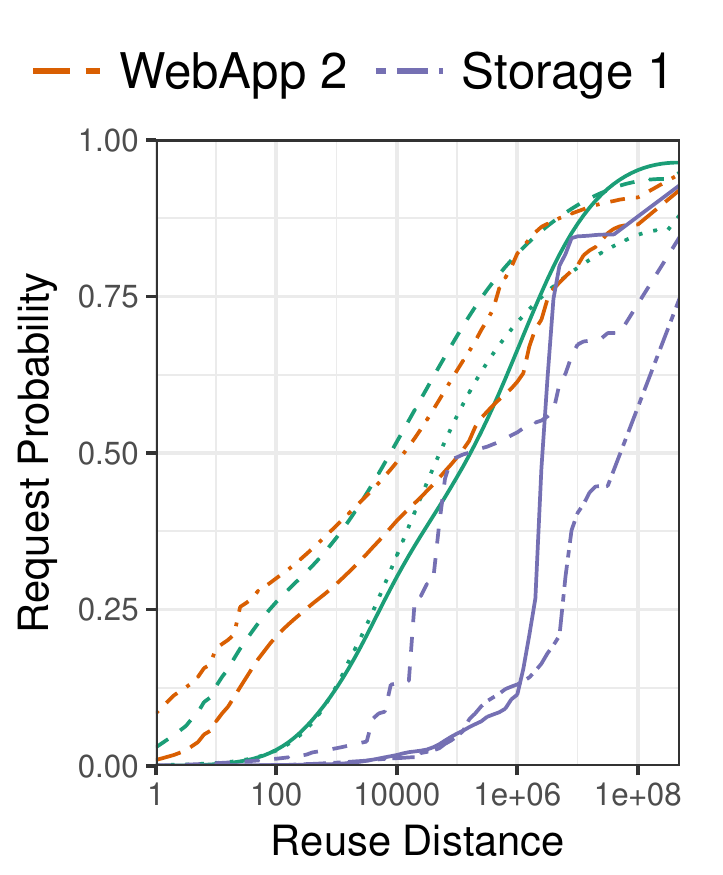}
    \caption{Reuse distances.}
    \label{fig:methodology:traces:reuse-distance}
  \end{subfigure}
  \hfill
  \begin{subfigure}{0.24\linewidth}
    \includegraphics[trim={2mm 0 4.1mm 0},clip,width=\linewidth]{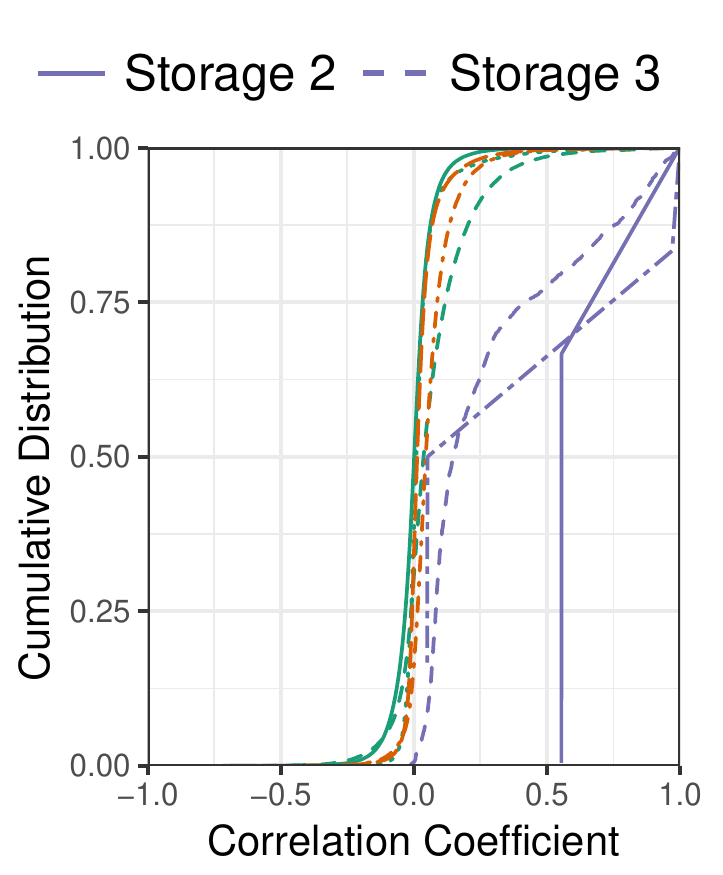}
    \caption{Correlations.}
    \label{fig:methodology:traces:correlation}
  \end{subfigure}
  \hfill
  \mbox{}
   \caption{The production traces used in our evaluation comes from three different domains (CDNs, WebApps, and storage) and thus exhibit starkly different request patterns.
     (a) The object size distribution in CDN traces space more than nine orders of mangitude, where object sizes in WebApp and Storage traces are much smaller.
     (b) Object popularities in CDNs and WebApps follow approximately Zipf popularities, whereas the popularity distribution for storage traces is more complex.
     (c) The reuse distance distribution of CDNs and WebApps is smooth, whereas in the Storage there are jumps -- indicating correlated request sequences such as scans.
     (d) The distribution of Pearson correlation coffecients for the top 10k-most popular objects shows that CDN and WebApp traces do not exhibit linear correlations, whereas the three storage traces show distince patterns with significant positive correlation. One of the storage traces (Storage 2) in fact shows a correlation coefficient greater than $0.5$ for all 10k-most popular objects.}
  \label{fig:methodology:traces}
\end{figure*}
}

\newcommand{\figprecproof}{%
  \begin{figure}[h]
  \centering
    \resizebox{.8\textwidth}{!}{%
      \begin{tikzpicture}[->,>=stealth',shorten >=1pt,auto,node distance=1.5cm,semithick,every node/.style={inner sep=0,outer sep=0.01mm,minimum size=0.3mm}]
        \node[]               (a) {\Huge\textbf{\dots}};
        \node[state,right=of a,fill=darkred,xshift=-13mm]                    (b) {\LARGE\textcolor{white}{$\boldsymbol{j}$}};
        \node[right=of b,font=\Huge,scale=1.7]                    (c) {\textbf{\dots}};
        \node[state,right=of c,fill=blue]                    (b2) {\LARGE\textcolor{white}{$\boldsymbol{i}$}};
        \node[right=of b2,font=\Huge,scale=1.7]                   (d) {\textbf{\dots}};
        \node[state,right=of d,fill=blue]                    (a2) {\LARGE\textcolor{white}{$\boldsymbol{\ell_i}$}};
        \node[right=of a2,font=\Huge,scale=1.7]                    (c2) {\textbf{\dots}};
        \node[state,right=of c2,fill=darkred]                    (d2) {\LARGE\textcolor{white}{$\boldsymbol{\ell_j}$}};
        \node[right=of d2,xshift=-11mm,font=\Huge]                    (a3) {\textbf{\dots}};
        \path[line width=0.5mm, color=black]
        (b)   edge[color=black!55!white,bend left=7] node[yshift=.3mm] {}   (c)
        (c)   edge[line width=0.7mm,bend left=7] node[yshift=.3mm] {}   (b)
        (c)   edge[color=black!55!white,bend left=7] node[yshift=.3mm] {}   (b2)
        (b2)   edge[line width=0.7mm,bend left=7] node[yshift=.3mm] {}   (c)
        (b2)   edge[color=black!55!white,bend left=7] node[yshift=.3mm] {}   (d)
        (d)   edge[line width=0.7mm,bend left=7] node[yshift=.3mm] {}   (b2)
        (d)   edge[color=black!55!white,bend left=7] node[yshift=.3mm] {}   (a2)
        (a2)   edge[line width=0.7mm,bend left=7] node[yshift=.3mm] {}   (d)
        (a2)   edge[color=black!55!white,bend left=7] node[yshift=.3mm] {}   (c2)
        (c2)   edge[line width=0.7mm,bend left=7] node[yshift=.3mm] {}   (a2)
        (c2)   edge[color=black!55!white,bend left=7] node[yshift=.3mm] {}   (d2)
        (d2)   edge[line width=0.7mm,bend left=7] node[yshift=.3mm] {}   (c2);
        \path[line width=0.5mm, color=blue]
        (b2)   edge[color=blue!55!white,bend right=37] node[yshift=.3mm] {}   (a2)
        (a2)   edge[line width=0.7mm,line width=0.7mm,bend left=50] node[yshift=.3mm] {}   (b2.south);
        \path[->,line width=0.5mm, color=darkred]
        (b)   edge[line width=0.7mm,bend left=20] node[yshift=-4mm] {}   (d2);
      \end{tikzpicture}
    }
    \caption{The precedence relation $i \prec j$ from \autoref{def:precrelation} forces integer decisions on interval $i$.
      In any min-cost flow solution, we can reroute flow such that if $x_j>0$ then $x_i=1$.}
    \label{fig:precproof}
  \end{figure}
}

\usepackage[font=small, labelfont=bf]{caption}
\usepackage{geometry}
\usepackage{microtype}
\usepackage{etoolbox}
\usepackage{mathpazo}

\makeatletter
\patchcmd{\@maketitle}{\LARGE \@title}{\vspace{-3em}\fontsize{12}{12}\selectfont\textbf{\@title}}{}{}
\makeatother

\usepackage{titlesec}
\titleformat*{\section}{\fontsize{13}{14}\bfseries}
\titleformat*{\subsection}{\fontsize{11}{11.5}\bfseries}

\begin{document}
\graphicspath{{figures/}}

\title{Practical Bounds on Optimal Caching with Variable Object Sizes}

\author{Daniel S. Berger\thanks{dsberger@cs.cmu.edu, Carnegie Mellon University, 5000 Forbes Ave, Pittsburgh, PA.}}
\author{Nathan Beckmann\thanks{beckmann@cs.cmu.edu}}
\author{Mor Harchol-Balter\thanks{harchol@cs.cmu.edu}}
\affil{Carnegie Mellon University}
\date{}

\maketitle
  \vspace{-3em}

\begin{abstract}
  Many recent caching systems aim to improve miss ratios, but there is no good sense among practitioners of how much further miss ratios can be improved.
In other words, should the systems community continue working on this problem?

Currently, there is no principled answer to this question.
In practice, object sizes often vary by several orders of magnitude, where computing the optimal miss ratio (OPT) is known to be NP-hard.
The few known results on caching with variable object sizes provide very weak bounds and are impractical to compute on traces of realistic length.

We propose a new method to compute upper and lower bounds on OPT.
Our key insight is to represent caching as a min-cost flow problem, hence we call our method the \emph{flow-based offline optimal} (FOO).
We prove that, under simple independence assumptions, FOO’s bounds become tight as the number of objects goes to infinity.
Indeed, FOO's error over 10M requests of production CDN and storage traces is negligible: at most 0.3\%.
FOO thus reveals, for the first time, the limits of caching with variable object sizes.

While FOO is very accurate, it is computationally impractical on traces with hundreds of millions of requests.
We therefore extend FOO to obtain more efficient bounds on OPT, which we call \emph{practical flow-based offline optimal} (PFOO).
We evaluate PFOO on several full production traces and use it to compare OPT to prior online policies.
This analysis shows that current caching systems are in fact still far from optimal, suffering 11--43\% more cache misses than OPT,
whereas the best prior offline bounds suggest that there is essentially no room for improvement.

\end{abstract}

{\footnotesize \noindent The full version of this paper appears in Proceedings of the ACM on Measurement and Analysis of Computing Systems as Article 32 in Volume 2, Issue 2, June 2018.
\url{https://doi.org/10.1145/3224427}.}



\section{Introduction}
\label{sec:intro}

Caches are pervasive in computer systems, and their miss ratio often determines end-to-end application performance. For example, content distribution networks (CDNs) depend on large, geo-distributed caching networks to serve user requests from nearby datacenters, and their response time degrades dramatically when the requested content is not cached nearby. 
Consequently, there has been a renewed focus on improving cache miss ratios, particularly for web services and content delivery~\cite{blankstein2017hyperbolic,berger2017adaptsize,neglia2017cache,cidon2016cliffhanger,maggs2015algorithmic,gast2015transient,li2015gd,einziger2014tinylfu,huang2013analysis}. These systems have demonstrated significant miss ratio improvements over least-recently used (LRU) caching, the de facto policy in most production systems. \emph{But how much further can miss ratios improve?} \emph{Should the systems community continue working on this problem, or have all achievable gains been exhausted?}

\subsection{The problem: Finding the offline optimal}
To answer these questions, one would like to know the best achievable miss ratio, free of constraints---i.e., the offline optimal (OPT).
Unfortunately, very little is known about OPT with variable object sizes.
For objects with equal sizes, computing OPT is simple (i.e., Belady~\cite{belady,mattson}), and it is widely used in the systems community to bound miss ratios.
But object sizes often vary widely in practice, from a few bytes (e.g., metadata~\cite{nishtala2013scaling}) to several gigabytes (e.g., videos~\cite{huang2013analysis,berger2017adaptsize}).
We need a way to compute OPT for variable object sizes,
but unfortunately this is known to be NP-hard~\cite{chorbak}.

There has been little work on bounding OPT with variable object sizes,
and all of this work gives very weak bounds.
On the theory side, prior work gives only three approximation algorithms~\cite{irani1997page,albers1999page,bar2001unified},
and the best approximation is only provably within a factor of 4 of OPT.
Hence, when this algorithm estimates a miss ratio of $0.4$, OPT may lie anywhere between $0.1$ and $0.4$.
This is a big range---in practice, a difference of $0.05$ in miss ratio is significant---, so bounds from prior theory are of limited practical value.
From a practical perspective, there is an even more serious problem with the theoretical bounds: they are simply too expensive to compute.
The best approximation takes 24 hours to process 500\,K requests and scales poorly (\autoref{sec:background}),
while production traces typically contain hundreds of millions of requests.

Since the theoretical bounds are incomputable,
practitioners have been forced to use conservative lower bounds or pessimistic upper bounds on OPT.
The only prior lower bound is an infinitely large cache~\cite{cidon2016cliffhanger,abrams1995,huang2013analysis},
which is very conservative and gives no sense of how OPT changes at different cache sizes.
Belady variants (e.g., Belady-Size in \autoref{sec:bg:practice}) are widely used as an upper bound~\cite{huang2013analysis,li2016popularity,hillmann2016simulation,shukla2016optimal},
despite offering no guarantees of optimality.
While these offline bounds are easy to compute,
we will show that they are in fact far from OPT.
They have thus given practitioners a false sense of complacency,
since existing online algorithms often achieve similar miss ratios to these weak offline upper bounds.

\subsection{Our approach: Flow-based offline optimal}

We propose a new approach to compute bounds on OPT with variable object sizes,
which we call the \emph{flow-based offline optimal (FOO)}.
The key insight behind FOO is to represent caching as a min-cost flow problem.
This formulation yields a lower bound on OPT by allowing non-integer decisions, i.e., letting the cache retain fractions of objects for a proportionally smaller reward.
It also yields an upper bound on OPT by ignoring all non-integer decisions.
Under simple independence assumptions, we prove that the non-integer decisions become negligible as the number of objects goes to infinity,
and thus the bounds are asymptotically tight.

Our proof is based on the observation that an optimal policy will strictly prefer some requests over others, forcing integer decisions.
We show such preferences apply to almost all requests by relating such preferences to the well-known coupon collector problem.

Indeed, FOO's error over 10M requests of five production CDN and web application traces is negligible: at most 0.15\%.
Even on storage traces, which have highly correlated requests that violate our proof's independence assumptions, FOO's error remains below 0.3\%.
FOO thus reveals, for the first time, the limits of caching with variable object sizes.

While FOO is very accurate, it is too computationally expensive to apply directly to production traces containing hundreds of millions of requests.
To extend our analysis to such traces, we develop more efficient upper and lower bounds on OPT, which we call \emph{practical flow-based offline optimal (PFOO)}.
PFOO enables the first analysis of optimal caching on traces with hundreds of millions of requests,
and reveals that there is still significant room to improve current caching systems.


\subsection{Summary of results}

PFOO yields nearly tight bounds, as shown in \autoref{fig:intro}.
This figure plots the miss ratio obtained on a production CDN trace with over 400 million requests for several techniques:
online algorithms (LRU, GDSF~\cite{cherkasova1998improving}, and AdaptSize~\cite{berger2017adaptsize}),
prior offline upper bounds (Belady and Belady-Size),
the only prior offline lower bound (Infinite-Cap),
and our new upper and lower bounds on OPT (PFOO).
The theoretical bounds, including both prior work and FOO, are too slow to run on traces this long.

\begin{figure}[h]
  \begin{minipage}{0.49\linewidth}
    \centering
    \includegraphics[width=\linewidth]{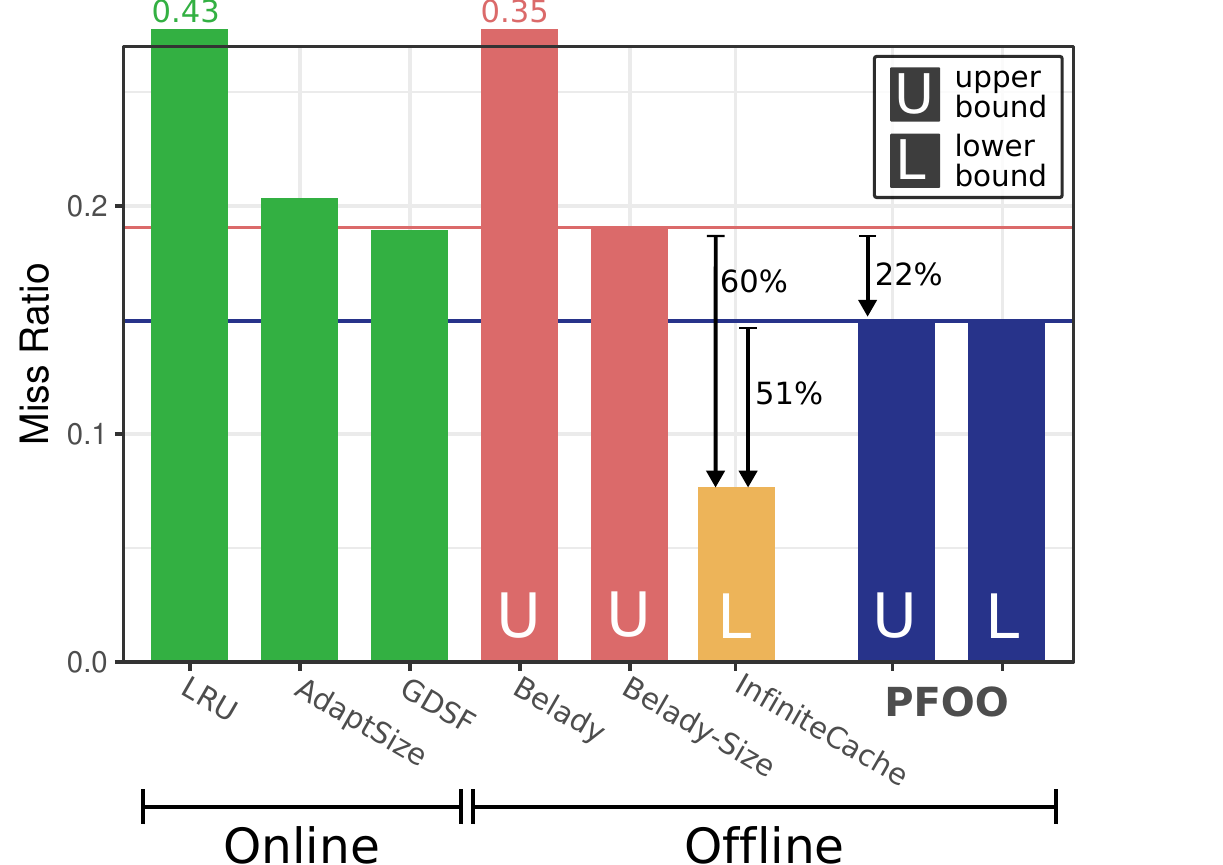}
  \end{minipage}
  \begin{minipage}{0.5\linewidth}
\caption{Miss ratios on a production CDN trace for a 4\,GB cache.
  Prior to our work, the best prior upper bound on OPT is within 1\% of online algorithms on this trace, leading to the false impression that there is no room for improvement.
  The only prior lower bound on OPT (Infinite-Cap) is 60\% lower than the best upper bound on OPT (Belady-Size).
  By contrast, PFOO provides nearly tight upper and lower bounds, which are 22\% below the online algorithms. PFOO thus shows that there is actually significant room for improving current caching algorithms.}
\label{fig:intro}
  \end{minipage}
\end{figure}

\autoref{fig:intro} illustrates the two problems with prior practical bounds and how PFOO improves upon them.
First, prior bounds on OPT are very weak.
The miss ratio of Infinite-Cap is 60\% lower than Belady-Size, whereas the gap between PFOO's upper and lower bounds is less than 1\%.
Second, prior bounds on OPT are misleading.
Comparing GDSF and Belady-Size, which are also within 1\%, one would conclude that online policies are nearly optimal.
In contrast, PFOO gives a much better upper bound and shows that there is in fact a 22\% gap between state-of-the-art online policies and OPT.

We evaluate FOO and PFOO extensively on eight production traces with hundreds of millions of requests across cache capacities from 16\,MB to 64\,GB.
On CDN traces from two large Internet companies, PFOO bounds OPT within 1.9\% relative error at most and 1.4\% on average.
On web application traces from two other large Internet companies, PFOO bounds OPT within 8.6\% relative error at most and 6.1\% on average.
On storage traces from Microsoft~\cite{snia}, where our proof assumptions do not hold, PFOO still bounds OPT within 11\% relative error at most and 5.7\% on average. 

\emph{PFOO thus gives nearly tight bounds on the offline optimal miss ratio for a wide range of real workloads.}
We find that PFOO achieves on average 19\% lower miss ratios than Belady, Belady-Size, and other obvious offline upper bounds, demonstrating the value of our min-cost flow formulation.

\subsection{Contributions}
In summary, this paper contributes the following:
\begin{itemize}[leftmargin=4mm]
\item
  We present \emph{flow-based offline optimal (FOO)}, a novel formulation
of offline caching as a min-cost flow problem (Sections~\ref{sec:foo} and~\ref{sec:overview}). FOO
exposes the underlying problem structure and enables our
remaining contributions.
\item
  We prove that FOO is asymptotically exact under simple independence assumptions,
  giving the first tight, polynomial-time bound on OPT with variable object sizes (\autoref{sec:proof}).
\item
  We present \emph{practical flow-based offline optimal (PFOO)},
  which relaxes FOO to give fast and nearly tight bounds on real traces with
  hundreds of millions of requests (\autoref{sec:practical}).
  PFOO gives the first reasonably tight lower bound on OPT on long traces.
\item
  We perform the first extensive evaluation of OPT with
  variable object sizes on eight production traces (Sections~\ref{sec:methodology} and~\ref{sec:evaluation}).
  In contrast to prior offline bounds,
  PFOO reveals that the gap between online policies and OPT is much larger than previously thought:
  27\% on average, and up to 43\% on web application traces.
\end{itemize}
Our implementations of FOO, PFOO, and of the previously unimplemented approximation algorithms are publicly available\footnote{\url{https://github.com/dasebe/optimalwebcaching}}.

\section{Background and Motivation}
\label{sec:background}

Little is known about how to efficiently compute OPT with variable object sizes.
On the theory side, the best known approximation algorithms give very weak bounds and are too expensive to compute.
On the practical side, system builders use offline heuristics that are much cheaper to compute, but give no guarantee that they are close to OPT.
This section surveys known theoretical results on OPT and the offline heuristics used in practice.

\subsection{Offline bounds are more robust than online bounds}\label{sec:bg:onlinevsopt}

\autoref{fig:intro} showed a 22\% gap between the best online policies and OPT.
One might wonder whether this gap arises from comparing an \emph{offline} policy with an \emph{online} policy,
rather than from any weakness in the online policies.
In other words, offline policies have an advantage because they know the future.
Perhaps this advantage explains the gap.

This raises the question of whether we could instead compute the \emph{online optimal} miss ratio.
Such a bound would be practically useful, since all systems necessarily use online policies.
Our answer is that we would like to, but unfortunately this is impossible.
To bound the online optimal miss ratio, one must make assumptions about what information an online policy can use to make caching decisions.
For example, traditional policies such as LRU and LFU make decisions using the history of past requests.
These policies are conceptually simple enough that prior work has proven online bounds on their performance in a variety of settings~\cite{aho1971principles,bahat2003optimal,o1999optimality,ferragut2016optimizing,ioannidis2016adaptive}.
However, real systems have quirks and odd behaviors that are hard to capture in tractable assumptions.
With enough effort, system designers can exploit these quirks to build online policies that \emph{outperform the online bounds}
and approach OPT's miss ratio for specific workloads.
For example, Hawkeye~\cite{jain2016back} recently demonstrated that a seemingly irrelevant request feature in computer architecture (the program counter) is in fact tightly correlated with OPT's decisions, allowing Hawkeye to mimic OPT on many programs.
Online bounds are thus inherently fragile
because real systems behave in strange ways, and
it is impossible to bound the ingenuity of system designers to discover and exploit these behaviors.




We instead focus on the offline optimal to get robust bounds on cache performance.
Offline bounds are widely used by practitioners,
especially when objects have the same size and OPT is simple to compute~\cite{belady,mattson}.
However, computing OPT with variable object sizes is strongly NP-complete~\cite{chorbak},
meaning that no fully polynomial-time approximation scheme (FPTAS) can exist.%
\footnote{That no FPTAS can exist follows from Corollary 8.6 in~\cite{vazirani2013approximation}, as OPT meets the assumptions of Theorem 8.5.}
We are therefore unlikely to find tight bounds on OPT for arbitrary traces.
Instead, our approach is to develop a technique that can estimate OPT on \emph{any} trace, which we call FOO,
and then show that FOO yields tight bounds on traces \emph{seen in practice}.
To do this, our approach is two-fold.
First, we prove that FOO gives tight bounds
on traces that obey the independent reference model (IRM),
the most common assumptions used in prior analysis of online policies%
~\cite{aho1971principles,king1971analysis,gelenbe1973unified,coffman1973operating,mccabe1965serial,burville1973model,hendricks1972stationary,dan1990,tsukada2012fluid,flajolet1992birthday,fill1996distribution,jelenkovic1999asymptotic,dobrow1995move,rodrigues1995performance,jelenkovic1999performance,jelenkovic2004least,panagakis2008approximate,psounis2004modeling,gallo2012performance,young2000line,o1999optimality,fricker2012versatile,martina13,berger2015maximizing,Berger20142,gast2015transient,starobinski2001probabilistic,jelenkovic2004optimizing}.
The IRM assumption is only needed for the proof; FOO is designed to work on any trace and does not itself rely on any properties of the IRM.
Second, we show empirically that FOO's error is less than 0.3\% on real-world traces, including several that grossly violate the IRM.
Together, we believe these results demonstrate that FOO is both theoretically well-founded and practically useful.

\subsection{OPT with variable object sizes is hard}

Since OPT is simple to compute for equal object sizes~\cite{belady,mattson},
it is surprising that OPT becomes NP-hard when object sizes vary.
Caching may seem similar to Bin-Packing or Knapsack,
which can be approximated well when there are only a few object sizes and costs.
But caching is quite different because the order of requests constrains OPT's choices in ways that are not captured by these problems or their variants.
In fact, the NP-hardness proof in~\cite{chorbak} is quite involved and reduces from Vertex Cover, not Knapsack.
Furthermore, OPT remains strongly NP-complete even with just three object sizes~\cite{folwarczny2017general},
and heuristics that work well on Knapsack perform badly in caching (see ``Freq/Size'' in \autoref{sec:bg:practice} and \autoref{sec:evaluation}).

\subsection{Prior bounds on OPT with variable object sizes are impractical}\label{sec:bg:offlinevarsize}

Prior work gives only three polynomial time bounds on OPT~\cite{albers1999page,bar2001unified,irani1997page}, which vary in time complexity and approximation guarantee.
\autoref{tbl:algos} summarizes these bounds by comparing their asymptotic run-time, how many requests can be calculated in practice (e.g., within 24\,hrs), and their approximation guarantee.

\begin{table}[h]
\centering
\begin{minipage}{.7\linewidth}
  \small
  \centering
  \begin{tabular}{cccc}
    \toprule
    \textbf{Technique} & \textbf{Time} & \textbf{Requests\,/\,24hrs} & \textbf{Approximation} \\
    \midrule
    OPT & NP-hard~\cite{chorbak} & <1\,K & 1\\
    LP rounding~\cite{albers1999page} & $\Omega(N^{5.6})$ & 50\,K & $O\!\left(\log\frac{\max_i\{ s_i\}}{\min_i\{ s_i\}}\right)$ \\
    LocalRatio~\cite{bar2001unified} & $O(N^3)$ & 500\,K & $4$ \\
    OFMA~\cite{irani1997page} & $O(N^2)$ & 28\,M & $O\!\left(\log C\right)$ \\
    \bf FOO\footnote{FOO's approximation guarantee holds under independence assumptions.} & ${O(N^{3/2})}$ & 28\,M & 1 \\
    \bf PFOO\footnote{PFOO does not have an approximation guarantee but its upper and lower bounds are within 6\% on average on production traces.} & ${O(N \log N)}$ & 250\,M & $\approx$1.06 \\
    \bottomrule
  \end{tabular}

\vspace{.5mm}
\footnotesize{Notation: $N$ is the trace length, $C$ is the cache capacity, and $s_i$ is the size of object $i$.}
\vspace{-1mm}

\end{minipage}
  \vspace{0.8em}
  \caption{Comparison of FOO and PFOO to prior bounds on OPT with variable object sizes.
    Computing OPT is NP-hard.
    Prior bounds~\cite{albers1999page,bar2001unified,irani1997page} provide only weak approximation guarantees, whereas FOO's bounds are tight.
    PFOO performs well empirically and can be calculated for hundreds of millions of requests.}
  \label{tbl:algos}
\end{table}

Albers et al.~\cite{albers1999page} propose an LP relaxation of OPT and a rounding scheme.
Unfortunately, the LP requires $N^2$ variables, which leads to a high $\Omega(N^{5.6})$-time complexity~\cite{koufogiannakis2014nearly}.
Not only is this running time high, but the approximation factor is logarithmic in the ratio of largest to smallest object (e.g., around 30 on production traces), making this approach impractical.

Bar et al.~\cite{bar2001unified} propose a general approximation framework (which we call \emph{LocalRatio}), which can be applied to the offline caching problem.
This algorithm gives the best known approximation guarantee, a factor of $4$.
Unfortunately, this is still a weak guarantee,
as we saw in \autoref{sec:intro}. 
Additionally, LocalRatio is a purely theoretical algorithm, with a high running time of $O(N^3)$, and which we believe had not been implemented prior to our work.
Our implementation of LocalRatio can calculate up to 500\,K requests in 24\,hrs, which is only a small fraction of the length of production traces.

Irani proposes the OFMA approximation algorithm~\cite{irani1997page}, which has $O(N^2)$ running time.
This running time is small enough for our implementation of OFMA to run on small traces (Section~\ref{sec:evaluation}).
Unfortunately, OFMA achieves a weak approximation guarantee, logarithmic in the cache capacity $C$,
and in fact OFMA does badly on our traces, giving much weaker bounds than simple Belady-inspired heuristics.

Hence, prior work that considers adversarial assumptions yields only weak approximation guarantees.
We therefore turn to stochastic assumptions to obtain tight bounds on the kinds of traces actually seen in practice.
Under independence assumptions, FOO achieves a tight approximation guarantee on OPT, unlike prior approximation algorithms,
and also has asymptotically better runtime, specifically $O\!\left(N^{3/2}\!\right)$.

We are not aware of any prior stochastic analysis of offline optimal caching.
%
In the context of \emph{online} caching policies,
there is an extensive body of work using stochastic assumptions similar to ours~\cite{king1971analysis,gelenbe1973unified,coffman1973operating,mccabe1965serial,burville1973model,hendricks1972stationary,dan1990,tsukada2012fluid,flajolet1992birthday,fill1996distribution,jelenkovic1999asymptotic,dobrow1995move,rodrigues1995performance,jelenkovic1999performance,jelenkovic2004least,panagakis2008approximate,psounis2004modeling,gallo2012performance,young2000line,o1999optimality,fricker2012versatile,martina13,berger2015maximizing,Berger20142,gast2015transient,starobinski2001probabilistic,jelenkovic2004optimizing,beckmann:hpca16:model},
of which five prove optimality results~\cite{aho1971principles,bahat2003optimal,o1999optimality,ferragut2016optimizing,ioannidis2016adaptive}.
Unfortunately, these results are only for objects with equal sizes,
and these policies perform poorly in our experiments 
because they do not account for object size.


\subsection{Heuristics used in practice to bound OPT give weak bounds}\label{sec:bg:practice}

Since the running times of prior approximation algorithms are too high for production traces,
practitioners have been forced to rely on heuristics that can be calculated more quickly.
However, these heuristics only give upper bounds on OPT, and there is no guarantee on how close to OPT they are.

The simplest offline upper bound is Belady's algorithm,
which evicts the object whose next use lies furthest in the future.%
\footnote{Belady's MIN algorithm actually operates somewhat differently.
  The algorithm commonly called ``Belady'' was invented (and proved to be optimal) by Mattson~\cite{mattson}.
  For a longer discussion, see~\cite{michaud:taco16:opt}.}
Belady is optimal in caching variants with equal object sizes~\cite{belady,mattson,gill2008multi,kallahalla2002pc}.
Even though it has no approximation guarantees for variable object sizes, it is still widely used in the systems community~\cite{huang2013analysis,li2016popularity,hillmann2016simulation,shukla2016optimal}.
However, as we saw in \autoref{fig:intro},
Belady performs very badly with variable object sizes and is easily outperformed by state-of-the-art online policies.

A straightforward size-aware extension of Belady is to evict the object with the highest cost = object size $\times$ next-use distance.
We call this variant \emph{Belady-Size}.
Among practitioners, Belady-Size is widely believed to perform near-optimally, but it has no guarantees.
It falls short on simple examples:
e.g., imagine that \textsf{A} is 4\,MB and is referenced 10 requests hence and never referenced again,
and \textsf{B} is 5\,MB and is referenced 9 and 12 requests hence.
With 5\,MB of cache space, the best choice between these objects is to keep \textsf{B},
getting two hits. 
But \textsf{A} has cost = 4 $\times$ 10 = 40,
and \textsf{B} has cost = 5 $\times$ 9 = 45,
so Belady-Size keeps \textsf{A} and gets only one hit.
In practice, Belady-Size often leads to poor decisions when a large
object is referenced twice in short succession: Belady-Size
will evict many other useful objects to make space for it, sacrificing
many future hits to gain just one.

Alternatively, one could use Knapsack heuristics as size-aware offline upper bounds, such as the density-ordered Knapsack heuristic, which is known to perform well on Knapsack in practice~\cite{du2013handbook}.
We call this heuristic \emph{Freq/Size}, as Freq/Size evicts the object with the lowest utility = frequency / size,
where frequency is the number of requests to the object.
Unfortunately, Freq/Size also falls short on simple examples:
e.g., imagine that \textsf{A} is 2\,MB and is referenced 10 requests hence,
and \textsf{B} is (as before) 5\,MB and is referenced 9 and 12 requests hence.
With 5\,MB of cache space, the best choice between these objects is to keep \textsf{B},
getting two hits.
But \textsf{A} has utility = 1 $\div$ 2 = 0.5,
and \textsf{B} has utility = 2 $\div$ 5 = 0.4,
so Freq/Size keeps \textsf{A} and gets only one hit.
In practice, Freq/Size often leads to poor decisions when an object is
referenced in bursts: Freq/Size will retain such objects long
after the burst has passed, wasting space that could have earned hits
from other objects.

Though these heuristics are easy to compute and intuitive, they give no approximation guarantees.
We will show that they are in fact far from OPT on real traces, and PFOO is a much better bound.

\section{Flow-based Offline Optimal}
\label{sec:foo}

This section gives a conceptual roadmap for our proof of FOO's optimality, which we present formally in Sections~\ref{sec:overview} and~\ref{sec:proof}.
Throughout this section we use a small request trace shown in \autoref{fig:foo:trace} as a running example.
This trace contains four objects,
\textcolor{darkyellow}{\textbf{a}},
\textcolor{darkred}{\textbf{b}},
\textcolor{blue}{\textbf{c}},
and \textcolor{darkgreen}{\textbf{d}}, with sizes 3, 1, 1, and 2, respectively.

\begin{figure}[h]
  \centering
  \begin{tabular}{c c}
    \toprule
    Object
    &
      \fontsize{11pt}{11pt}
      \bf
      \makebox[0.15in][l]{\textcolor{darkyellow}{a}}
      \makebox[0.15in][l]{\textcolor{darkred}{b}}
      \makebox[0.15in][l]{\textcolor{blue}{c}}
      \makebox[0.15in][l]{\textcolor{darkred}{b}}
      \makebox[0.15in][l]{\textcolor{darkgreen}{d}}
      \makebox[0.15in][l]{\textcolor{darkyellow}{a}}
      \makebox[0.15in][l]{\textcolor{blue}{c}}
      \makebox[0.15in][l]{\textcolor{darkgreen}{d}}
      \makebox[0.15in][l]{\textcolor{darkyellow}{a}}
      \makebox[0.15in][l]{\textcolor{darkred}{b}}
      \makebox[0.15in][l]{\textcolor{darkred}{b}}
      \makebox[0.05in][l]{\textcolor{darkyellow}{a}}
    \\
    Size
    &
      \makebox[0.155in][l]{3}
      \makebox[0.155in][l]{1}
      \makebox[0.155in][l]{1}
      \makebox[0.155in][l]{1}
      \makebox[0.155in][l]{2}
      \makebox[0.155in][l]{3}
      \makebox[0.155in][l]{1}
      \makebox[0.155in][l]{2}
      \makebox[0.155in][l]{3}
      \makebox[0.155in][l]{1}
      \makebox[0.155in][l]{1}
      \makebox[0.05in][l]{3}
    \\
    \bottomrule
  \end{tabular}
  \caption{Example trace of requests to objects \textcolor{darkyellow}{a}, \textcolor{darkred}{b}, \textcolor{blue}{c}, and \textcolor{darkgreen}{d}, of sizes 3, 1, 1, and 2, respectively.}
  \label{fig:foo:trace}
\end{figure}

First, we introduce a new integer linear program to represent OPT (\autoref{sec:foo:interval}).
After relaxing integrality constraints, we derive FOO's min-cost flow representation, which can be solved efficiently (\autoref{sec:foo:mcf}).
We then observe how FOO yields tight upper and lower bounds on OPT (\autoref{sec:foo:bounds}).
To prove that FOO's bounds are tight on real-world traces, we relate the gap between FOO's upper and lower bounds to the occurrence of a partial order on intervals,
and then reduce the partial order's occurrence to an instance of the generalized coupon collector problem (\autoref{sec:foo:proof}).

\subsection{Our new interval representation of OPT}\label{sec:foo:interval}

We start by introducing a novel representation of OPT.
Our integer linear program (ILP) minimizes the number of cache misses, while having full knowledge of the request trace.

We exploit a unique property of offline optimal caching: OPT never changes its decision to cache object $k$ in between two requests to $k$ (see \autoref{sec:overview}).
This naturally leads to an interval representation of OPT as shown in \autoref{fig:foo:ilp}.
While the classical representation of OPT uses decision variables to track the state of every object at every time step~\cite{albers1999page}, our ILP only keeps track of interval-level decisions.
Specifically, we use decision variables $x_i$ to indicate whether OPT caches the object requested at time $i$, or not.

\begin{figure}[h]
  \centering
  \resizebox{.65\linewidth}{!}{
  \begin{tabular}{p{8mm} c}
  \toprule
  Object
    &
      \fontsize{11pt}{11pt}{%
  \bf
  \makebox[0.22in][l]{\textcolor{darkyellow}{a}}
  \makebox[0.22in][l]{\textcolor{darkred}{b}}
  \makebox[0.22in][l]{\textcolor{blue}{c}}
  \makebox[0.22in][l]{\textcolor{darkred}{b}}
  \makebox[0.22in][l]{\textcolor{darkgreen}{d}}
  \makebox[0.22in][l]{\textcolor{darkyellow}{a}}
  \makebox[0.22in][l]{\textcolor{blue}{c}}
  \makebox[0.22in][l]{\textcolor{darkgreen}{d}}
  \makebox[0.22in][l]{\textcolor{darkyellow}{a}}
  \makebox[0.22in][l]{\textcolor{darkred}{b}}
  \makebox[0.22in][l]{\textcolor{darkred}{b}}
    \makebox[0.22in][l]{\textcolor{darkyellow}{a}}
    }
    \\
    \parbox[t]{8mm}{\multirow{6}{*}{\rotatebox[origin=c]{90}{
    \begin{minipage}[t]{21mm}
      Interval Decision Variables
    \end{minipage}
    \hspace{1mm}
    }}}
    &
    \\
    \\
    \\
    \\
    \\
    &  \tikzoverlay at (-38.5mm,-1.5mm){\includegraphics[width=0.51\textwidth]{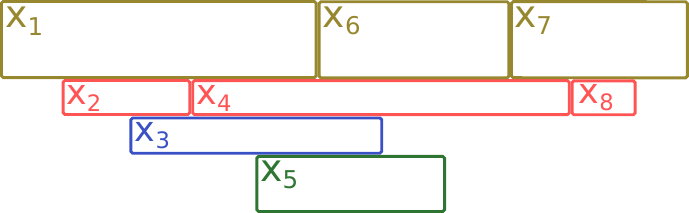}};
    \\
  \bottomrule
  \end{tabular}
  }
  
  \caption{Interval ILP representation of OPT.}
  \label{fig:foo:ilp}
  \vspace{-1em}
\end{figure}

\subsection{FOO's min-cost flow representation}\label{sec:foo:mcf}

This interval representation leads naturally to FOO's flow-based representation, shown in \autoref{fig:foo:mcf}.
We use the min cost flow notation shown in \autoref{tbl:mcf:notation}.
The key idea is to use flow to represent the interval decision variables.
Each request is represented by a node.
Each object's first request is a source of flow equal to the object's size,
and its last request is a sink of flow in the same amount.
This flow must be routed along intervening edges,
and hence min-cost flow must decide whether to cache the object throughout the trace.

\figFooMcf

For cached objects, there is a central path of black edges connecting all requests.
These edges have capacity equal to the cache capacity and cost zero (since cached objects lead to zero misses).
Min-cost flow will thus route as much flow as possible through this central path to avoid costly misses elsewhere~\cite{ahuja1993network}.

To represent cache misses, FOO adds outer edges between subsequent requests to the same object.
For example, there are three edges along the top of \autoref{fig:foo:mcf} connecting the requests to \textcolor{darkyellow}{\bf a}.
These edges have capacity equal to the object's size $s$
and cost inversely proportional to the object's size $1/s$.
Hence, if an object of size $s$ is not cached
(i.e., its flow $s$ is routed along this outer edge),
it will incur a cost of $s \times (1/s) = 1$ miss.

The routing of flow through this graph implies which objects are cached and when.
When no flow is routed along an outer edge,
this implies that the object is cached and the subsequent request is a hit.
All other requests, i.e., those with any flow routed along an outer edge, are misses.
The min-cost flow gives the decisions that minimize total misses.

\subsection{FOO yields upper and lower bounds on OPT}\label{sec:foo:bounds}

FOO can deviate from OPT as there is no guarantee that an object's flow will be entirely routed along its outer edge.
Thus, FOO allows the cache to keep fractions of objects, accounting for only a fractional miss on the next request to that object.
In a real system, each fractional miss would be a full miss.
This error is the price FOO pays for making the offline optimal computable.

To deal with fractional (non-integer) solutions, we consider two variants of FOO.
FOO-L keeps all non-integer solutions and is therefore a lower bound on OPT.
FOO-U considers all non-integer decisions as uncached, ``rounding up'' flow along outer edges, and is therefore an upper bound on OPT.
We will prove this in \autoref{sec:overview}.

\begin{figure}[h]
  \centering
  \begin{tabular}{c c}
  \toprule
  Object
  &
    \fontsize{11pt}{11pt}
  \bf
  \makebox[0.15in][l]{\textcolor{darkyellow}{a}}
  \makebox[0.15in][l]{\textcolor{darkred}{b}}
  \makebox[0.15in][l]{\textcolor{blue}{c}}
  \makebox[0.15in][l]{\textcolor{darkred}{b}}
  \makebox[0.15in][l]{\textcolor{darkgreen}{d}}
  \makebox[0.15in][l]{\textcolor{darkyellow}{a}}
  \makebox[0.15in][l]{\textcolor{blue}{c}}
  \makebox[0.15in][l]{\textcolor{darkgreen}{d}}
  \makebox[0.15in][l]{\textcolor{darkyellow}{a}}
  \makebox[0.15in][l]{\textcolor{darkred}{b}}
  \makebox[0.15in][l]{\textcolor{darkred}{b}}
  \makebox[0.05in][l]{\textcolor{darkyellow}{a}}
  \\
  OPT decision
  &
  \makebox[0.15in][l]{\xmark}
  \makebox[0.15in][l]{\cmark}
  \makebox[0.15in][l]{\cmark}
  \makebox[0.15in][l]{\cmark}
  \makebox[0.16in][l]{\xmark}
  \makebox[0.15in][l]{\xmark}
  \makebox[0.15in][l]{\cmark}
  \makebox[0.15in][l]{\xmark}
  \makebox[0.15in][l]{\xmark}
  \makebox[0.16in][l]{\cmark}
  \makebox[0.15in][l]{\xmark}
  \makebox[0.05in][l]{\xmark}
  \\
  FOO-L decision
  &
  \makebox[0.15in][l]{0}
  \makebox[0.15in][l]{1}
  \makebox[0.15in][l]{1}
  \makebox[0.15in][l]{1}
  \makebox[0.16in][l]{$\frac{1}{2}$}
  \makebox[0.15in][l]{0}
  \makebox[0.15in][l]{1}
  \makebox[0.15in][l]{$\frac{1}{2}$}
  \makebox[0.16in][l]{0}
  \makebox[0.15in][l]{1}
  \makebox[0.155in][l]{0}
  \makebox[0.05in][l]{0}
  \\
  FOO-U decision
  &
  \makebox[0.15in][l]{0}
  \makebox[0.15in][l]{1}
  \makebox[0.15in][l]{1}
  \makebox[0.15in][l]{1}
  \makebox[0.16in][l]{0}
  \makebox[0.15in][l]{0}
  \makebox[0.15in][l]{1}
  \makebox[0.15in][l]{0}
  \makebox[0.16in][l]{0}
  \makebox[0.15in][l]{1}
  \makebox[0.155in][l]{0}
  \makebox[0.05in][l]{0}
  \\
  \bottomrule
  \end{tabular}
  \caption{Caching decisions made by OPT, FOO-L, and FOO-U with a cache capacity of $C = 3$.}
  \label{fig:foo:decisions}
\end{figure}

\autoref{fig:foo:decisions} shows the caching decisions made by OPT, FOO-L, and FOO-U assuming a cache of size 3.
A ``\cmark'' indicates that OPT caches the object until its next request,
and a ``\xmark'' indicates it is not cached.
OPT suffers five misses on this trace by caching object \textcolor{darkred}{\bf b} and
either \textcolor{blue}{\bf c} or \textcolor{darkgreen}{\bf d}.
OPT caches \textcolor{darkred}{\bf b} because it is referenced thrice and is small.
This leaves space to cache the two references to either \textcolor{blue}{\bf c} or \textcolor{darkgreen}{\bf d}, but not both.
(OPT in \autoref{fig:foo:decisions} chooses to cache \textcolor{blue}{\bf c} since it requires less space.)
OPT does not cache \textcolor{darkyellow}{\bf a}
because it takes the full cache, forcing misses on all other requests.

%

The solutions found by FOO-L are very similar to OPT.
FOO-L decides to cache objects \textcolor{darkred}{\bf b} and \textcolor{blue}{\bf c},
matching OPT,
and also caches \emph{half} of \textcolor{darkgreen}{\bf d}.
FOO-L thus underestimates the misses by one, counting \textcolor{darkgreen}{\bf d}'s misses fractionally.
FOO-U gives an upper bound for OPT by counting \textcolor{darkgreen}{\bf d}'s misses fully.
In this example, FOO-U matches OPT exactly.

\subsection{Overview of our proof of FOO's optimality}\label{sec:foo:proof}
We show both theoretically and empirically that FOO-U and FOO-L yield tight bounds.
Specifically, we prove that FOO-L's solutions are almost always integer when there are many objects (as in production traces).
Thus, FOO-U and FOO-L coincide with OPT.

Our proof is based on a natural precedence relation between intervals such that an optimal policy strictly prefers some intervals over others.
For example, in \autoref{fig:foo:trace}, FOO will always prefer ${x_2}$ over $x_1$ and $x_8$ over $x_7$.
This can be seen in the figure, as interval $x_2$ fits entirely within $x_1$, and likewise $x_8$ fits within $x_7$.
In contrast, no such precedence relation exists between $x_6$ and $x_4$ because \textcolor{darkyellow}{\bf a} is larger than \textcolor{darkred}{\bf b}, and so $x_6$ does not fit within $x_4$.
Similarly, no precedence relation exists between $x_2$ and $x_5$ because, although $x_5$ is longer and larger, their intervals do not overlap, and so $x_2$ does not fit within $x_5$.

This precedence relation means that if FOO caches any part of $x_1$, then it must have cached all of $x_2$.
Likewise, if FOO caches any part of $x_7$, then it must have cached all of $x_8$.
The precedence relation thus forces integer solutions in FOO.
Although this relation is sparse in the small trace from \autoref{fig:foo:trace}, as one scales up the number of objects the precedence relation becomes dense.
Our challenge is to prove that this holds on traces seen in practice.

At the highest level, our proof distinguishes between ``typical'' and ``atypical'' objects.
Atypical objects are those that are exceptionally unpopular or exceptionally large; typical objects are everything else.
While the precedence relation may not hold for atypical objects, intervals from atypical objects are rare enough that they can be safely ignored.
We then show that for all the typical objects, the precedence relation is dense.
In fact, one only needs to consider precedence relations among \emph{cached} objects, as all other interval have zero decision variables.
The basic intuition behind our proof is that a popular cached object almost always takes precedence over another object.
Specifically, it will take precedence over one of the exceptionally large objects,
since the only way it could not is if \emph{all} of the exceptionally large objects were requested before it was requested again.
There are enough large objects to make this vanishingly unlikely.

This is an instance of the generalized \emph{coupon collector problem} (CCP).
In the CCP, one collects coupons (with replacement) from an urn with $k$ distinct types of coupons,
stopping once all $k$ types have been collected.
The classical CCP (where coupons are equally likely) is a well-studied problem~\cite{feller2008introduction}.
The generalized CCP, where coupons have non-uniform probabilities, is very challenging and the focus of recent work in probability theory~\cite{doumas2012coupon,moriarty2008generalized,zy2011coupon,anceaume2015new}.

Applying these recent results, we show that it is extremely unlikely that a popular object does not take precedence over any others.
Therefore, there are very few non-integer solutions among popular objects, which make up nearly all hits,
and the gap between FOO-U and FOO-L vanishes as the number of objects grows large.

\section{Formal Definition of FOO}
\label{sec:overview}

This section shows how to construct FOO and that FOO yields upper and lower bounds on OPT.
Section~\ref{sec:overview:notation} introduces our notation.
Section~\ref{sec:overview:newilp} defines our new interval representation of OPT.
Section~\ref{sec:overview:mcf} relaxes the integer constraints and proves that our min-cost flow representation yields upper and lower bounds on OPT.

\subsection{Notation and definitions}\label{sec:overview:notation}

The trace $\sigma$ consists of $N$ requests to $M$ distinct objects.
The $i$-th request $\sigma_i$ contains the corresponding object id, for all $i \in \{1 \dots N\}$.
We use $s_i$ to reference the size of the object $\sigma_i$ referenced in the $i$-th request.
We denote the $i$-th interval (e.g., in \autoref{fig:foo:ilp})
by $[i,\ell_i)$, where $\ell_i$ is the time of the next request to object $\sigma_i$ after time $i$,
  or $\infty$ if $\sigma_i$ is not requested again.

OPT minimizes the number of cache misses, while having full knowledge of the request trace.
OPT is constrained to only use cache capacity $C$ (bytes), and is \emph{not allowed to prefetch} objects as this would lead to trivial solutions (no misses)~\cite{albers1999page}.
Formally,
\begin{assumption}\label{ass:prefetching}
An object $k \in \{1 \dots M\}$ can only enter the cache at times $i \in \{1 \dots N\}$ with $\sigma_i = k$.
\end{assumption}

\subsection{New ILP representation of OPT}\label{sec:overview:newilp}

We start by formally stating our ILP formulation of OPT,
based on intervals as illustrated in \autoref{fig:foo:ilp}.
First, we define the set $I$ of all requests $i$ where $\sigma_i$ is requested again,
i.e., $I = \{ i : \ell_i < \infty \}$.
$I$ is the times when OPT must decide whether to cache an object.
For all $i \in I$, we associate a decision variable $x_i$.
This decision variable denotes whether object $\sigma_i$ is cached during the interval $[i, \ell_i)$.
Our ILP formulation needs only $N-M$ variables, vs.\ $N \times M$ for prior approaches~\cite{albers1999page},
and leads directly to our flow-based approximation.

\begin{definition}[\textbf{Definition of OPT}]\label{def:ilp}
  The interval representation of OPT for a trace of length $N$ with $M$ objects is as follows.
  \begin{align}
    \label{eq:ilp} \text{OPT} = \min & \sum_{i \in I} (1-x_i)\\
    \nonumber \text{subject to:} \quad &\\
    \label{eq:ilpcapacity}\sum_{j:j < i < \ell_j} s_j x_j \le C & \quad\quad\forall i \in I\\
    \label{eq:ilpintegrality} x_i \in \{0,1\} & \quad\quad\forall i \in I
  \end{align}
\end{definition}

To represent the capacity constraint at every time step $i$, our representation needs to find all intervals $[j,\ell_j)$ that intersect with $i$, i.e., where $j<i<\ell_j$.
Eq.~\eqref{eq:ilpcapacity} enforces the capacity constraint
by bounding the size of cached intervals to be less than the cache size $C$.
Eq.~\eqref{eq:ilpintegrality} ensures that decisions are integral,
i.e., that each interval is cached either fully or not at all.
\autoref{sec:proof:lma:opt-equivalence} proves that our interval ILP is equivalent to classic ILP formulations of OPT from prior work~\cite{albers1999page}.

Having formulated OPT with fewer decision variables, we could try to solve the LP relaxation of this specific ILP.
However, the capacity constraint, Eq.~\eqref{eq:ilpcapacity}, still poses a practical problem since finding the intersecting intervals is computationally expensive.
Additionally, the LP formulation does not exploit the underlying problem structure, which we need to bound the number of integer solutions.
We instead reformulate the problem as min-cost flow.

\subsection{FOO's min-cost flow representation of OPT}\label{sec:overview:mcf}

This section presents the relaxed version of OPT as an instance of min-cost flow (MCF) in a graph $G$.
We denote a surplus of flow at a node $i$ with $\beta_i > 0$, and a demand for flow with $\beta_i<0$.
Each edge $(i,j)$ in $G$ has a cost per unit flow $\gamma_{(i,j)}$ and a capacity for flow $\mu_{(i,j)}$ (see right-hand side of \autoref{fig:foo:mcf}).

As discussed in \autoref{sec:foo},
the key idea in our construction of an MCF instance is that each interval introduces an amount of flow equal to the object's size.
The graph $G$ is constructed such that this flow competes for a single sequence of edges (the ``inner edges'') with zero cost.
These ``inner edges'' represent the cache's capacity: if an object is stored in the cache, we incur zero cost (no misses).
As not all objects will fit into the cache, we introduce ``outer edges'', which allow MCF to satisfy the flow constraints.
However, these outer edges come at a cost: when the full flow of an object uses an outer edge we incur cost 1 (i.e., a miss).
Non-integer decision variables arise if part of an object is in the cache (flow along inner edges) and part is out of the cache (flow along outer edges).


Formally, we construct our MCF instance of OPT as follows:

\begin{definition}[\textbf{FOO's representation of OPT}]\label{def:mcfinstance}
  Given a trace with $N$ requests and $M$ objects, the MCF graph $G$ consists of $N$ nodes.
  For each request $i \in \{1 \dots N\}$ there is a node with supply/demand
  \begin{align}
    \beta_i =
    \begin{cases}
      s_i & \text{if } i \text{ is the first request to } \sigma_i \\
      -s_i & \text{if } i \text{ is the last request to } \sigma_i \\
      0 & \text{otherwise.}
    \end{cases}
  \end{align}
  
  An \textbf{inner edge} connects nodes $i$ and $i+1$. Inner edges have capacity $\mu_{(i,i+1)}=C$ and cost $\gamma_{(i,i+1)}=0$, for $i \in \{1 \dots N-1\}$.

  For all $i \in I$, an \textbf{outer edge} connects nodes $i$ and $\ell_i$.
  Outer edges have capacity $\mu_{(i,\ell_i)}=s_i$ and cost $\gamma_{(i,\ell_i)}=1/s_i$.
  We denote the flow through outer edge $(i,\ell_i)$ as $f_i$.

  \vspace{1mm}
\noindent
  \textbf{FOO-L} denotes the cost of an optimal feasible solution to the MCF graph $G$.
  \textbf{FOO-U} denotes the cost if all non-zero flows through outer edges $f_i$ are rounded up to the edge's capacity $s_i$.
\end{definition}

This representation yields a min-cost flow instance with $2N-M-1$ edges, which is solvable in ${O(N^{3/2})}$~\cite{daitch2008faster,becker2014,becker2013combinatorial}.
Note that while this paper focuses on optimizing miss ratio (i.e., the fault model~\cite{albers1999page}, where all misses have the same cost),
\autoref{def:mcfinstance} easily supports non-uniform miss costs by setting outer edge costs to $\gamma_{(i,\ell_i)} = \text{cost}_i / s_i$.
We next show how to derive upper and lower bounds from this min-cost flow representation.

\begin{theorem}[\textbf{FOO bounds OPT}]\label{lma:mcfcorrectness}
  For FOO-L and FOO-U from Definition~\ref{def:mcfinstance}, 

  \vspace{-3mm}
  \begin{equation}
    \fbox{\large$\text{FOO-L} \leq \text{OPT} \leq \text{FOO-U}$}
  \end{equation}
\end{theorem}

\begin{proof} 
  We observe that $f_i$ as defined in Definition~\ref{def:mcfinstance}, defines the number of bytes ``not stored'' in the cache.
  $f_i$ corresponds to the $i$-th decision variable $x_i$ from Definition~\ref{def:ilp} as $x_i = (1 - f_i / s_i)$.
  
  \vspace{1mm}
\noindent
  ($\text{FOO-L} \leq \text{OPT}$):
  FOO-L is a feasible solution for the LP relaxation of Definition~\ref{def:ilp}, because a total amount of flow $s_i$ needs to flow from node $i$ to node $\ell_i$ (by definition of $\beta_i$).
  At most $\mu_{(i,i+1)} = C$ flows uses an inner edge which enforces constraint Eq.~\eqref{eq:ilpcapacity}.
  FOO-L is an optimal solution because it minimizes the total cost of flow along outer edges.
  Each outer edge's cost is $\gamma_{(i,\ell_i)} = 1/s_i$, so $\gamma_{(i,\ell_i)} f_i = (1 - x_i)$, and thus
  \begin{equation}
    \text{FOO-L} = \min \big\{ \sum_{i \in I} \gamma_{(i,\ell_i)} f_i \big\} = \min \big\{ \sum_{i\in I} (1 - x_i) \big\} \leq \text{OPT}
  \end{equation}

  \vspace{1mm}
\noindent
($\text{OPT} \leq \text{FOO-U}$):
After rounding, each outer edge $(i, \ell_i)$ has flow $f_i \in \{0, s_i\}$, so the corresponding decision variable $x_i \in \{0, 1\}$.
FOO-U thus yields a feasible integer solution, and OPT yields no more misses than any feasible solution.
\end{proof}

\section{FOO is Asymptotically Optimal}
\label{sec:proof}

This section proves that FOO is asymptotically optimal, namely that the gap between FOO-U and FOO-L vanishes as the number of objects grows large.
\autoref{sec:opt:main} formally states this result and our assumptions, and Sections~\ref{sec:opt:precgraph}--\ref{sec:opt:finalproof} present the proof.


\subsection{Main result and assumptions}\label{sec:opt:main} 

Our proof of FOO's optimality relies on two assumptions:
\emph{(i)}~that the trace is created by stochastically independent request processes
and \emph{(ii)}~that the popularity distribution is not concentrated on a finite set of objects as the number of objects grows.

\begin{assumption}[\textbf{Independence}]\label{ass:localirm}
  The request sequence is generated by independently sampling from a popularity distribution $\mathcal{P}^M$.
  Object sizes are sampled from an arbitrary continuous size distribution $\mathcal{S}$, which is independent of $M$ and has a finite $\max_i s_i$.
\end{assumption}
We assume that object sizes are unique to break ties when making caching decisions.
If the object sizes are not unique, one can simply add small amounts of noise to make them so.
We assume a maximum object size to show the existence of a scaling regime, i.e., that the number of cached objects grows large as the cache grows large.
For the same reason, we exclude trivial cases where a finite set of objects dominates the request sequence even as the total universe of objects grows large:

\begin{assumption}[\textbf{Diverging popularity distribution}]\label{ass:inf}
  For any number $M>0$ of objects, the popularity distribution $\mathcal{P}^M$ is defined via an infinite sequence $\psi_k$.
  At any time $1 \leq i \leq N$,
  \begin{align}
    \Pb{\text{object }k\text{ is requested} \;\vert\; M\text{ objects overall}} = \frac{\psi_k}{\sum_{k=1}^M \psi_k}
  \end{align}
  The sequence $\psi_k$ must be positive and diverging such that cache size $C \rightarrow \infty$ is required to achieve a constant miss ratio as $M \rightarrow \infty$.
\end{assumption}
Our assumptions on $\mathcal{P}^M$ allow for many common distributions, such as uniform popularities ($\psi_k$ = 1) or heavy-tailed Zipfian probabilities ($\psi_k = 1/k^\alpha$ for $\alpha \leq 1$, as is common in practice~\cite{breslau1999web,huang2013analysis,sitaraman2014overlay,maggs2015algorithmic,berger2017adaptsize}).
Moreover, with some change to notation, our proofs can be extended to require only that $\psi_k$ remains constant over short timeframes.
With these assumptions in place, we are now ready to state our main result on FOO's asymptotic optimality.

\begin{theorem}[\textbf{FOO is Asymptotically Optimal}]\label{thm:exactness}
  Under Assumptions~\ref{ass:localirm} and~\ref{ass:inf}, for any error $\varepsilon$ and violation probability $\kappa$, there exists an $M^*$ such that for any trace with $M>M^*$ objects
  \begin{align}
    \Pb{\text{FOO-U} - \text{FOO-L} \geq  \varepsilon \; N } \leq \kappa
  \end{align}
    where the trace length $N\geq M \log^2 M$ and the cache capacity $C$ is scaled with $M$ such that FOO-L's miss ratio remains constant.
\end{theorem}

\autoref{thm:exactness} states that, as $M \rightarrow \infty$, FOO's \emph{miss ratio error} is almost surely less than $\varepsilon$
for any $\varepsilon > 0$.
Since FOO-L and FOO-U bound OPT (\autoref{lma:mcfcorrectness}),
  $\text{FOO-L} = \text{OPT} = \text{FOO-U}$.

The rest of this section is dedicated to the proof \autoref{thm:exactness}.
The key idea in our proof is to bound the number of non-integer solutions in FOO-L via a precedence relation that forces FOO-L to strictly prefer some decision variables over others, which forces them to be integer.
\autoref{sec:opt:precgraph} introduces this precedence relation.
\autoref{sec:opt:coupons} maps this relation to a representation that can be stochastically analyzed (as a variant of the coupon collector problem).
\autoref{sec:opt:stochastic} then shows that almost all decision variables are part of a precedence relation and thus integer, and \autoref{sec:opt:finalproof} brings all these parts together in the proof of \autoref{thm:exactness}.

\subsection{Bounding the number of non-integer solutions using a precedence graph}\label{sec:opt:precgraph}

This section introduces the precedence relation $\prec$ between caching intervals.
The intuition behind $\prec$ is that if an interval $i$ is nested entirely within interval $j$,
then min-cost flow must prefer $i$ over $j$.
We first formally define $\prec$, and then state the property about optimal policies in \autoref{thm:precrelation}.

\begin{definition}[\textbf{Precedence relation}]\label{def:precrelation}
  For two caching intervals $[i,\ell_i)$ and $[j,\ell_j)$, let the relation $\prec$ be such that $i \prec j$ (``$i$ takes precedence over $j$'') if
\begin{enumerate}
\item $j<i$,
\item $\ell_j > \ell_i$, and
\item $s_i < s_j$.
\end{enumerate}
\end{definition}

The key property of $\prec$ is that it forces integer decision variables.

\begin{theorem}[\textbf{Precedence forces integer decisions}]\label{thm:precrelation}
  If $i \prec j$, then $x_j > 0$ in FOO-L's min-cost flow solution implies $x_i = 1$.
\end{theorem}


In other words, if interval $i$ is strictly preferable to interval $j$, then FOO-L will take all of $i$ before taking any of $j$.
The proof of this result relies on the notion of a residual MCF graph~\cite[p.304~ff]{ahuja1993network}, where for any edge $(i,j) \in G$ with positive flow, we add a backwards edge $(j,i)$ with cost $\gamma_{j,i} = - \gamma_{i,j}$.

\begin{proof}
  By contradiction.
  Let $G'$ be the residual MCF graph induced by a given MCF solution.
  \autoref{fig:precproof} sketches the MCF graph in the neighborhood of $j,\dots,i,\dots,\ell_i,\dots,\ell_j$.
  
\figprecproof

Assume that $x_j>0$ and that $x_i < 1$, as otherwise the statement is trivially true.
Because $x_j>0$ there exist backwards inner edges all the way between $\ell_j$ and $j$.
Because $x_i < 1$, the MCF solution must include some flow on the outer edge $(i,\ell_i)$,
  and there exists a backwards outer edge $(\ell_i,i) \in G'$ with cost $\gamma_{\ell_i,i} = -1/s_i$.
%

  
  We can use the backwards edges to create a cycle in $G'$, starting at $j$, then following edge $(j,\ell_j)$, backwards inner edges to $\ell_i$, the backwards outer edge $(\ell_i,i)$, and finally backwards inner edges to return to $j$.
  \autoref{fig:precproof} highlights this clockwise path in darker colored edges.

  This cycle has cost $= 1/s_j -1/s_i$, which is negative because $s_i < s_j$ (by definition of $\prec$ and since $i \prec j$).
  As negative-cost cycles cannot exist in a MCF solution~\cite[Theorem 9.1, p.307]{ahuja1993network}, this leads to a contradiction.
\end{proof}

To quantify how many integer solutions there are, we need
to know the general structure of the precedence graph.
Intuitively, for large production traces with many overlapping intervals,
the graph will be very dense.
We have empirically verified this for our production traces.

Unfortunately, the combinatorial nature of caching traces made it difficult
for us to characterize the general structure of the precedence graph under stochastic assumptions.
For example, we considered classical results on random graphs~\cite{karrer2009random} and the concentration of measure in random partial orders~\cite{bollobas1992height}.
None of these paths yielded sufficiently tight bounds on FOO.
Instead, we bound the number of non-integer solutions via the generalized coupon collector problem.

\subsection{Relating the precedence graph to the coupon collector problem}\label{sec:opt:coupons}


We now translate the problem of intervals without child in the precedence graph (i.e., intervals $i$ for which there exists no $i \prec j$) into a tractable stochastic representation.

We first describe the intuition for equal object sizes, and then consider variable object sizes.

\begin{definition}[\textbf{Cached objects}]
  Let $H_i$ denote the set of cached intervals that overlap time $i$, excluding $i$.
  \begin{align}
    H_i = \big\{j\neq i : \;\;x_j>0 \text{ and } i\in [j,\ell_j)   \big\} \quad \quad \text{and} \quad \quad h_i = \vert H_i \vert
  \end{align}
\end{definition}

\begin{figure}[b]
  \hspace{1mm}
  \begin{minipage}{0.24\linewidth}
    \centering
    \includegraphics[width=.9\linewidth]{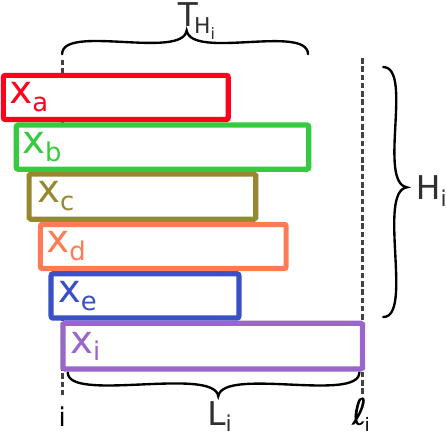}
  \end{minipage}
  \hfill
  \begin{minipage}{0.73\linewidth}
    \caption{Simplified notation for the coupon collector representation of offline caching with equal object sizes.
      We translate the precedence relation from \autoref{thm:precrelation} into the relation between two random variables.
      $L_i$ denotes the length of interval $i$.
      $T_{H_i}$ is the coupon collector time, where we wait until all objects that are cached at the beginning of $L_i$ ($H_i$ denotes these objects) are requested at least once.}
    \label{fig:coupon:nochild:equalsize}
  \end{minipage}
\end{figure}

We observe that interval $[i,\ell_i)$ is without child if and only if \emph{all other objects} $x_{j} \in H_i$ are \emph{requested at least once} in $[i,\ell_i)$.
\autoref{fig:coupon:nochild:equalsize} shows an example where $H_i$ consists of five objects (intervals $x_a,\dots,x_e$).
As all five objects are requested before $\ell_i$, all five intervals end before $x_i$ ends, and so $[i, \ell_i)$ cannot fit in any of them.
To formalize this observation, we introduce the following random variables, also illustrated in \autoref{fig:coupon:nochild:equalsize}.
$L_i$ is the length of the $i$-th interval, i.e., $L_i = \ell_i - i$.
$T_{H_i}$ is the time after $i$ when all intervals in $H_i$ end.
We observe that $T_{H_i}$ is the stopping time in a \emph{coupon-collector problem} (CCP) where we associate a coupon type with every object in $H_i$.
With equal object sizes, the event $\{i \text{ has a child} \}$ is equivalent to the event $\{T_{H_i} > L_i\}$.

We now extend our intuition to the case of variable object sizes.
We now need to consider that objects in $H_i$ can be smaller than $s_i$ and thus may not be $i$'s children for a new reason: the precedence relation (\autoref{def:precrelation}) requires $i$'s children to have size larger than or equal to $s_i$.
\autoref{fig:coupon:nochild:varsize} shows an example where $L_i$ is without child because \emph{(i)}~$x_b$, which ends after $\ell_i$, is smaller than $s_i$, and \emph{(ii)}~all larger objects ($x_a, x_c, x_d$) are requested before $\ell_i$.
The important conclusion is that, by ignoring the smaller objects, we can reduce the problem back to the CCP.

\begin{figure}[h]
  \hspace{1mm}
  \begin{minipage}{0.24\linewidth}
    \centering
    \includegraphics[width=.9\linewidth]{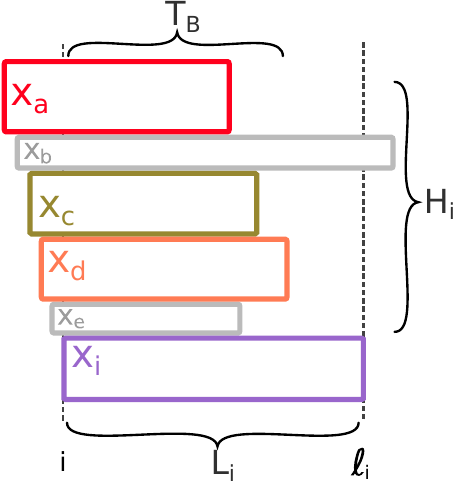}
  \end{minipage}
  \hfill
  \begin{minipage}{0.73\linewidth}
    \caption{Full notation for the coupon collector representation of offline caching.
      With variable object sizes, we need to ignore all objects with a smaller size than $s_i$ (greyed out intervals $x_b$ and $x_e$).
      We then define the coupon collector time $T_{B}$ among a subset $B \subset H_i$ of cached objects with a size larger than or equal to $s_i$.
      Using this notation, the event $T_{B} > L_i$ implies that $x_i$ has a child, which forces $x_i$ to be integer by \autoref{thm:precrelation}.}
    \label{fig:coupon:nochild:varsize}
  \end{minipage}
\end{figure}

To formalize our observation about the relation to the CCP, we introduce the following random variables, also illustrated in \autoref{fig:coupon:nochild:varsize}.
We define $B$, which is a subset of the cached objects $H_i$ with a size equal to or larger than $s_i$, and the coupon collector time $T_{B}$ for $B$-objects.
These definitions are useful as the event $\{T_{B} > L_i\}$ implies that $i$ has a child and thus $x_i$ is integer, as we now show.

\begin{theorem}[\textbf{Stochastic bound on non-integer variables}]\label{lma:stochastic:bound}
  For decision variable $x_i$, $i \in \{1 \dots N\}$, assume that $B \subseteq H_i$ is a subset of cached objects where $s_j \geq s_i$ for all $j \in B$.
  Further, let the random variable $T_{B}$ denote the time until all intervals in $B$ end, i.e.,
    $T_{B} = \max_{j \in B}\ell_j-i$.

  If $B$ is non-empty, then the probability that $x_i$ is non-integer is upper bounded by the probability interval $i$ ends after all intervals in $B$, i.e.,
  \begin{align}
    \Pb{0 < x_i < 1} \leq \Pb{L_i > T_{B}}~~~~.
  \end{align}
\end{theorem}

The proof works backwards by assuming that $T_{B} > L_i$.
We then show that this implies that there exists an interval $j$ with $i \prec j$ and then apply \autoref{thm:precrelation} to conclude that $x_i$ is integer.
Finally, we use this implication to bound the probability.

\begin{proof}
  Consider an arbitrary $i \in \{1 \dots N\}$ with $T_{B} > L_i$.
  Let $[j,\ell_j)$ denote the interval in $B$ that is last requested, i.e., $\ell_j = \max_{k \in B}\ell_k$ ($j$ exists because $B$ is non-empty).
  To show that $i \prec j$, we check the three conditions of \autoref{def:precrelation}.
  \begin{enumerate}
  \item $j<i$, because $j \in H_i$ (i.e., $j$ is cached at time $i$);
  \item $\ell_j > \ell_i$, because $\ell_j = \max_{k \in B}\ell_k = i + T_{B} > i + L_i = \ell_i$; and
  \item $s_i < s_j$, because $j \in B$ (i.e., $s_j$ is bigger than $s_i$ by assumption).
  \end{enumerate}
  Having shown that $i \prec j$, we can apply \autoref{thm:precrelation}, so that $x_j>0$ implies $x_i=1$.
  Because $j \in H_i$, $j$ is cached $x_j > 0$ and thus $x_i=1$.
  Finally, we observe that $L_i \neq T_{B}$ and conclude the theorem's statement by translating the above implications, $x_i = 1 \Leftarrow i \prec j \Leftarrow T_B > L_i$, into probability.
  \begin{alignat}{2}
    \Pb{0 < x_i < 1} = 1 - \Pb{x_i \in \{0,1\}} &\leq 1 - \Pb{i \prec j} &&\leq 1 - \Pb{T_{B} > L_i} = \Pb{L_i > T_{B}}~~~.
  \end{alignat}
\end{proof}

\autoref{lma:stochastic:bound} simplifies the analysis of non-integer $x_i$ to the relation of two random variables, $L_i$ and $T_{B}$.
While $L_i$ is geometrically distributed, $T_{B}$'s distribution is more involved.

We map $T_{B}$ to the stopping time $T_{b,p}$ of a generalized CCP with $b = \vert B \vert$ different coupon types.
The coupon probabilities $p$ follow from the object popularity distribution $\mathcal{P}^M$ by conditioning on objects in $B$.
As the object popularities $p$ are not equal in general, characterizing the stopping time $T_{b,p}$ is much more challenging than in the classical CCP, where the coupon probabilities are assumed to be equal.
We solve this problem by observing that collecting $b$ coupons under equal probabilities stops faster than under $p$.
This fact may appear obvious, but it was only recently shown by Anceaume et al.~\cite[Theorem 4, p. 415]{anceaume2015new} (the proof is non-trivial).
Thus, we can use a classical CCP to bound the generalized CCP's stopping time and $T_{B}$.

\begin{lma}[\textbf{Connection to classical coupon connector problem}]\label{lma:couponinterpretation}
  For any object popularity distribution $\mathcal{P}^M$, and for $q=(1/b,\dots,1/b)$, using the notation from \autoref{lma:stochastic:bound}
  \begin{align}
    \Pb{T_{B} < l} \leq \Pb{T_{b,q} < l}\quad\quad\quad \text{for any} \qquad\quad l\geq 0~~~~.
  \end{align}
\end{lma}

The proof of this Lemma simply extends the result by Anceaume et al.\ to account for requests to objects not in $B$.
We state the full proof in \autoref{sec:proof:lma:couponinterpretation}.

\subsection{Typical objects almost always lead to integer decision variables}\label{sec:opt:stochastic}

We now use the connection to the coupon collector problem to show that almost all of FOO's decision variables are integer.
Specifically, we exclude a small number of very large and unpopular objects, and show that the remaining objects are almost always part of a precedence relation, which forces the corresponding decision variables to become integer.
In \autoref{sec:opt:finalproof}, we show that the excluded fraction is diminishingly small.

We start with a definition of the sets of large objects and the set of popular objects.

\begin{definition}[\textbf{Large objects and popular objects }]\label{def:popular}
  Let $N^*$ be the time after which the cache needs to evict at least one object.
  For time $i \in \{N^* \dots N\}$, we define the sets of large objects, $B_i$, and the set of popular objects, $F_i$.
  
  The set $B_i \subseteq H_i$ consists of the requests to the fraction $\delta$ largest objects of $H_i$ ($0 < \delta < 1$). We also define $b_i = \vert B_i \vert$, and we write ``$s_i < B_i$'' if $s_i < s_j$ for all $j \in B_i$ and ``$s_i \not < B_i$'' otherwise.
  
  The set $F_i$ consists of those objects $k$ with a request probability $\rho_k$ which lies above the following threshold.
  \begin{align}
    F_i = \left\{k: \rho_k \geq \frac{1}{b_i \log \log b_i}\right\}
  \end{align}
\end{definition}

Using the above definitions, we prove that ``typical'' objects (i.e., popular objects that are not too large) rarely lead to non-integer decision variables as the number of objects $M$ grows large.

\begin{theorem}[\textbf{Typical objects are rarely non-integer}]\label{lma:fractionprobability}
 For $i \in \{N^* \dots N\}$, $B_i$, and $F_i$ from Definition~\ref{def:popular}, 
  \begin{align}
\Pb{0<x_i<1 \;\big\vert\; s_i < B_i, \sigma_i \in F_i} \rightarrow 0\text{ as }M \rightarrow \infty~~~.
  \end{align}
\end{theorem}

The intuition is that, as the number of cached objects grows large,
it is vanishingly unlikely that \emph{all} objects in $B_i$ will be requested before a \emph{single} object is requested again.
That is, though there are not many large objects in $B_i$, there are enough that,
following \autoref{lma:stochastic:bound}, $x_i$ is vanishingly unlikely to be non-integer.
The proof of \autoref{lma:fractionprobability} uses elementary probability but relies on several very technical proofs.
We sketch the proof here, and state the formal proofs in the Appendices~\ref{sec:proof:lma:fractionprobability} to~\ref{sec:proof:lemma:ccpbound}.


\begin{proof}[Proof sketch]
  Following \autoref{lma:stochastic:bound}, it suffices to consider the event $\{L_i > T_{B_i}\}$.

  \begin{itemize}
  \item $\left(\Pb{L_i > T_{B_i}} \rightarrow 0 \text{ as } h_i \rightarrow \infty\right)$:
  We first condition on $L_i=l$, so that $L_i$ and $T_{B_i}$ become stochastically independent
  and we can bound $\Pb{L_i > T_{B_i}}$ by bounding either $\Pb{L_i > l}$ or $\Pb{l > T_{B_i}}$.
  Specifically, we split $l$ carefully into ``small $l$'' and ``large $l$'',
  and then show that $L_i$ is concentrated at small $l$, and $T_{B_i}$ is concentrated at large $l$.
  Hence, $\Pb{L_i > T_{B_i}}$ is negligible.

  \begin{itemize}
  \item (Small $l$:)
  For small $l$, we show that it is unlikely for all objects in $B_i$ to have been requested after $l$ requests.
  We upper bound the distribution of $T_{B_i}$ with $T_{b_i,u}$ (\autoref{lma:couponinterpretation}).
  We then show that the distribution of $T_{b_i,u}$ decays exponentially at values below its expectation (\autoref{lma:ccpbound}).
  Hence, for $l$ far below the expectation of $T_{b_i,u}$, the probability vanishes $\Pb{T_{b_i,u} < l} \rightarrow 0$,
  so long as $b_i = \delta h_i$ grows large, which it does because $h_i$ grows large.

\item (Large $l$:)
  For large $l$, $\Pb{L_i > l} \rightarrow 0$ because we only consider popular objects $\sigma_i \in F_i$ by assumption,
  and it is highly unlikely that a popular object is not requested after many requests.
\end{itemize}

\item $\left(h_i \rightarrow \infty\right)$:
  What remains to be shown is that the number of cached objects $h_i$ actually grows large.
  Since the cache size $C \rightarrow \infty$ as $M \rightarrow \infty$ by Assumption~\ref{ass:inf},
  this may seem obvious. Nevertheless, it must be demonstrated (\autoref{lma:hi:infinite}).
  The basic intuition is that $h_i$ is almost never much less than $h^*=C/\max_k s_k$,
  the fewest number of objects that could fit in the cache,
  and $h^*\rightarrow \infty$ as $C \rightarrow \infty$.

  \quad To see why $h_i$ is almost never much less than $h^*$,
  consider
  the probability that $h_i < x$, where $x$ is constant with respect to $M$.
  For any $x$, take large enough $M$ such that $x < h^*$.
  
  \quad Now, in order for $h_i < x$, almost all requests must go to distinct objects.
  Any object that is requested twice between $u$ (the last time where $h_u\geq h^*$) and $v$ (the next time where $h_v \geq h^*$) produces an interval (see \autoref{fig:hi:fluid}).
  This interval is guaranteed to fit in the cache, since $h_i < x < h^*$ means there is space for an object of any size.
  As $h^*$ and $M$ grow further,
  the amount of cache resources that must lay unused for $h_i < x$ grows further and further,
  and the probability that no interval fits within these resources becomes negligible.
\begin{figure}[h]
  \centering
  \includegraphics[width=.65\linewidth]{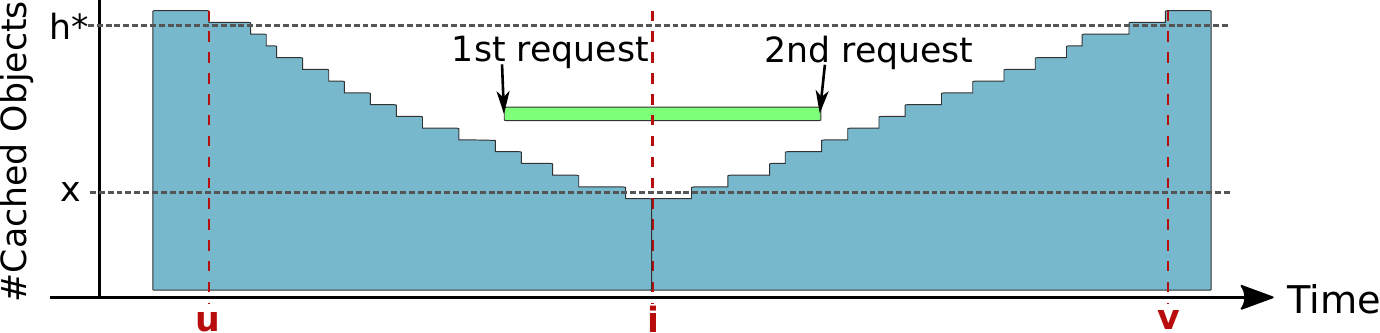}
\caption{The number of cached objects at time $i$, $h_i$, is unlikely to be far below $h^*=C/\max_k s_k$, the fewest number of objects that fit in the cache.
  In order for $h_i < h^*$ to happen, no other interval must fit in the white triangular space centered at $i$ (otherwise FOO-L would cache the interval).}
\label{fig:hi:fluid}
\end{figure}
\end{itemize}
\end{proof}


\subsection{Bringing it all together: Proof of \autoref{thm:exactness}}\label{sec:opt:finalproof}

This section combines our results so far and shows how to obtain \autoref{thm:exactness}:
\emph{There exists $M^*$ such that for any $M>M^*$ and for any error $\varepsilon$ and violation probability $\kappa$,} 
  \begin{align}
\label{eq:proof:thm:exactness}    \Pb{\text{FOO-U} - \text{FOO-L} \geq  \varepsilon \; N } \leq \kappa
  \end{align}

\begin{proof}[Proof of \autoref{thm:exactness}]
  We start by bounding the cost of non-integer solutions by the number of non-integer solutions, $\Omega$.
\begin{align}
\label{eq:f:definition}  \sum_{\{i:\;0<x_i<1\}} x_i \leq \sum_{i \in I} \mathbbm{1}_{\{i:\;0<x_i<1\}} = \Omega
\end{align}
It follows that $(\text{FOO-U}-\text{FOO-L})\leq\Omega$.
\begin{align}
  \label{eq:theorem2:foo-vs-b}
\Pb{\text{FOO-U} - \text{FOO-L} \geq \varepsilon \; N } &\leq \Pb{\Omega\geq \varepsilon \; N }
\intertext{We apply the Markov inequality.}
  \label{eq:vio:definition} &\leq \frac{\Ex{\Omega}}{N\;\varepsilon} = \frac{1}{N\;\varepsilon}\sum_{i \in I} \Pb{0<x_i<1}
\intertext{There are at most $N$ terms in the sum.}
& \leq \frac{\Pb{0<x_i<1}}{\varepsilon} 
\end{align}
To complete Eq.~\eqref{eq:proof:thm:exactness}, we upper bound $\Pb{0<x_i<1}$ to be less than $\varepsilon \, \kappa$.
We first condition on $s_i < B_i$ and $\sigma_i \in F_i$, double counting
those $i$ where $s_i \not < B_i$ and $\sigma_i \notin F_i$.
\begin{align}
  \Pb{0<x_i<1} \leq &
                      \Pb{0<x_i<1 \;\big\vert\; s_i < B_i, \sigma_i \in F_i }
                      \Pb{s_i < B_i, \sigma_i \in F_i}\\
                 & + \Pb{0<x_i<1 \;\big\vert\; s_i \not < B_i }\Pb{s_i \not < B_i}\\
                      & + \Pb{0<x_i<1 \;\big\vert\; \sigma_i \notin F_i }\Pb{\sigma_i \notin F_i}\\
\intertext{Drop $\leq 1$ terms.}
\label{eq:theorem2:proof:twoterms}  \leq & \Pb{0<x_i<1 \;\big\vert\; s_i < B_i, \sigma_i \in F_i} + \Pb{s_i \not < B_i} + \Pb{i \notin F_i}
\end{align}
To bound $\Pb{ 0 < x_i < 1 } \leq \varepsilon \, \kappa$, we choose parameters such that each term in Eq.~\eqref{eq:theorem2:proof:twoterms} is less than $\varepsilon \, \kappa/3$.
The first term vanishes by \autoref{lma:fractionprobability}.
The second term is satisfied by choosing $\delta = \varepsilon \, \kappa/3$ (\autoref{def:popular}).
For the third term, 
the probability that any cached object is unpopular vanishes as $h_i$ grows large.
\begin{align}
  \Pb{i \notin F_i}\leq & \frac{h_i}{b_i\log \log b_i}
                          = \frac{3}{\varepsilon \, \kappa \log \log \varepsilon \, \kappa \, h_i /3}
                          \rightarrow 0 \text{ as }h_i \rightarrow \infty
\end{align}
Finally, we choose $M^*$ large enough that the first and third terms in Eq.~\eqref{eq:theorem2:proof:twoterms} are each less than $\varepsilon \, \kappa / 3$.

\end{proof}

\vspace{-2mm}
This concludes our theoretical proof of FOO's optimality.

\section{Practical Flow-based Offline Optimal for Real-World Traces}
\label{sec:practical}

While FOO is asymptotically optimal and very accurate in practice,
as well as faster than prior approximation algorithms,
it is still not fast enough to process production traces with hundreds of millions of requests in a reasonable timeframe.
We now use the insights gained from FOO's graph-theoretic formulation
to design new upper and lower bounds on OPT,
which we call \emph{practical flow-based offline optimal (PFOO)}.
We provide the first practically useful lower bound, PFOO-L,
and an upper bound that is much tighter than prior practical offline upper bounds, PFOO-U:
\begin{align}
\label{eq:practical:ineq}  \fbox{$\large
    \text{PFOO-L} \leq \text{FOO-L} \leq \text{OPT} \leq \text{FOO-U} \leq \text{PFOO-U}
    $}
\end{align}



\subsection{Practical lower bound: PFOO-L}\label{sec:pfoo:l}

PFOO-L considers the \emph{total resources} consumed by OPT.
As \autoref{fig:pfoo-l:trace} illustrates, cache resources are limited in both space and time~\cite{beckmann:nsdi18:lhd}:
measured in resources, the cost to cache an object is the product of
\emph{(i)}~its size and \emph{(ii)}~its reuse distance (i.e., the number of accesses until it is next requested).
On a trace of length $N$, a cache of size $C$ has total resources $N \times C$.
The objects cached by OPT, or any other policy, cannot cost more total resources than this.

\begin{figure}[h]
\centering
\begin{subfigure}[b]{.46\linewidth}
  \centering
  \includegraphics[width=\linewidth]{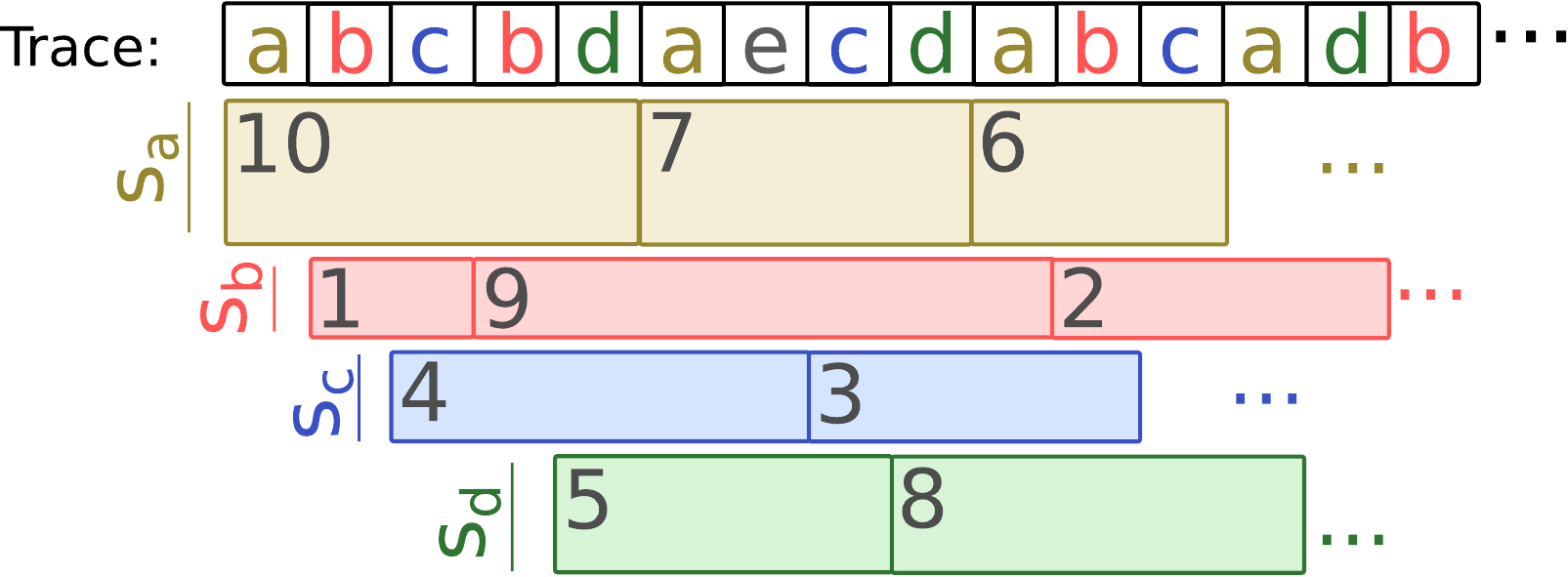}
  \caption{Intervals sorted by resource cost = reuse distance $\times$ object size.}
  \label{fig:pfoo-l:trace}
\end{subfigure}
\hspace{.025\linewidth}
\begin{subfigure}[b]{.48\linewidth}
  \centering
  \hspace{0.25in}\includegraphics[width=\linewidth]{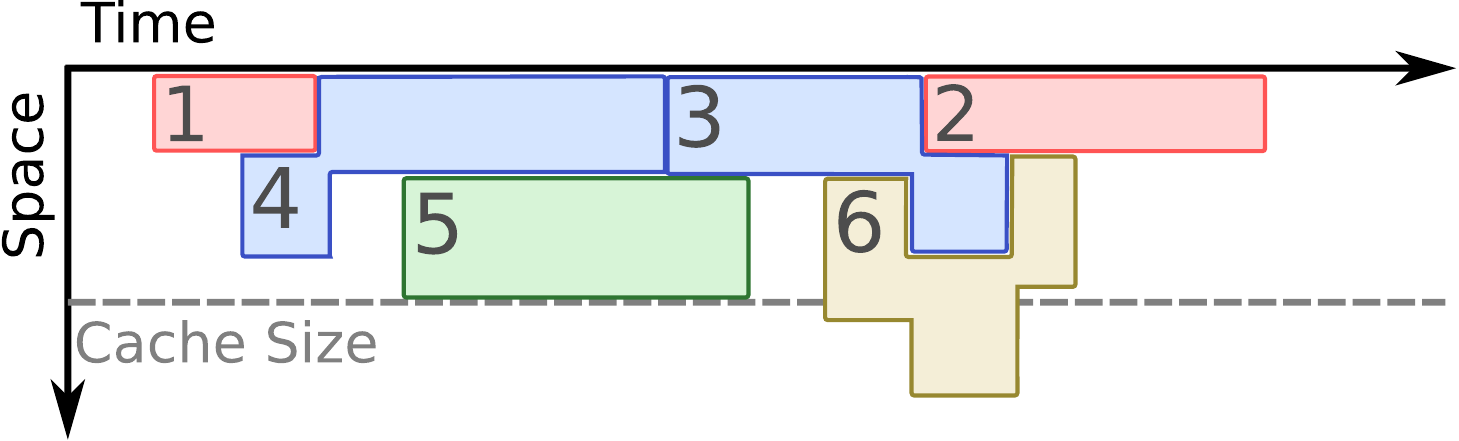}
  \vspace{.001mm}
  \caption{PFOO-L greedily claims the smallest intervals while not exceeding the average cache size.}
  \label{fig:pfoo-l:greedy}
\end{subfigure}
\caption{PFOO's lower bound, PFOO-L, constrains the total resources used over the full trace (i.e., size $\times$ time).
  PFOO-L claims the hits that require fewest resources, allowing cached objects to temporarily exceed the cache capacity.
}
\label{fig:pfoo-l}
\end{figure}

\paragraph{Definition of PFOO-L}
PFOO-L sorts all intervals by their resource cost, and caches the smallest-cost intervals up a total resource usage of $N \times C$.
\autoref{fig:pfoo-l:greedy} shows PFOO-L on a short request trace.
By considering only the total resource usage, PFOO-L ignores other constraints that are faced by caching policies.
In particular, PFOO-L does not guarantee that cached intervals take less than $C$ space at all times,
as shown by interval 6 for object \textcolor{darkyellow}{\bf a} in \autoref{fig:pfoo-l:greedy},
which exceeds the cache capacity during part of its interval.

\paragraph{Why PFOO-L works}
PFOO-L is a lower bound because no policy, including OPT or FOO-L, can get fewer misses using $N \times C$ total resources.
It gives a reasonably tight bound because, on large caches with many objects, the distribution of interval costs is similar throughout the request trace.
Hence, for a given cache capacity,
the ``marginal interval'' (i.e., the one barely does not fit in the cache under OPT) is also of similar cost throughout the trace.
Informally, PFOO-L caches intervals up to this marginal cost,
and so rarely exceeds cache capacity by very much.
The intuition behind PFOO-L is thus similar to the widely used Che approximation~\cite{che2002hierarchical},
which states that LRU caches keep objects up until some ``characteristic age''.
But unlike the Che approximation, which has unbounded error, PFOO-L uses this intuition to produce a robust lower bound.
This intuition holds particularly when requests are largely independent,
as in our proof assumptions and in traces from CDNs or other Internet services.
However, as we will see, PFOO-L introduces modest error even on other workloads where these assumptions do not hold.

PFOO-L uses a similar notion of ``cost'' as Belady-Size,
but provides two key advantages.
First, PFOO-L is closer to OPT.
Relaxing the capacity constraint lets PFOO-L avoid the pathologies discussed in \autoref{sec:bg:practice},
since PFOO-L can temporarily exceed the cache capacity to retain valuable objects that Belady-Size is forced to evict.
Second, relaxing the capacity constraint makes PFOO-L a \emph{lower} bound,
giving the first reasonably tight lower bound on long traces.

\subsection{Practical upper bound: PFOO-U}\label{sec:pfoo:u}

\pfooufigure

\paragraph{Definition of PFOO-U}
PFOO-U breaks FOO's min-cost flow graph into smaller segments of constant size,
and then solves each segment using min-cost flow incrementally.
By keeping track of the resource usage of already solved segments, PFOO-U yields a globally-feasible solution, which is an upper bound on OPT or FOO-U.

\paragraph{Example of PFOO-U}
\autoref{fig:pfoo-u} shows our approach on the trace from \autoref{fig:foo:trace} for a cache capacity of 3.
At the top is FOO's full min-cost flow problem; for large traces, this MCF problem is too expensive to solve directly.
Instead, PFOO-U breaks the trace into segments and constructs a min-cost flow problem for each segment.

PFOO-U begins by solving the min-cost flow for the first segment.
In this case, the solution is to cache both \textcolor{darkred}{\bf b}s, \textcolor{blue}{\bf c}, and one-third of \textcolor{darkyellow}{\bf a},
since these decisions incur the minimum cost of two-thirds, i.e., less than one cache miss.
As in FOO-U, PFOO-U rounds down the non-integer decision for \textcolor{darkyellow}{\bf a} and all following non-integer decisions.
Furthermore, PFOO-U only fixes decisions for objects in the first half of this segment.
This is done to capture interactions between intervals that cross segment boundaries.
Hence, PFOO-U ``forgets'' the decision to cache \textcolor{blue}{\bf c} and the second \textcolor{darkred}{\bf b},
and its final decision for this segment is only to cache the first \textcolor{darkred}{\bf b} interval.

PFOO-U then updates the second segment to account for its previous decisions.
That is, since \textcolor{darkred}{\bf b} is cached until the second request to \textcolor{darkred}{\bf b},
capacity must be removed from the min-cost flow to reflect this allocation.
Hence, the capacity along the inner edge \mbox{\textcolor{blue}{\bf c} $\rightarrow$ \textcolor{darkred}{\bf b}} is reduced from 3 to 2 in the second segment (\textcolor{darkred}{\bf b} is size 1).
Solving the second segment, PFOO-U decides to cache \textcolor{blue}{\bf c} and \textcolor{darkred}{\bf b}
(as well as half of \textcolor{darkgreen}{\bf d}, which is ignored).
Since these are in the first half of the segment,
we fix both decisions,
and move onto the third segment,
updating the capacity of edges to reflect these decisions as before.

PFOO-U continues to solve the following segments in this manner
until the full trace is processed.
On the trace from \autoref{sec:foo}, it decides to cache all requests to \textcolor{darkred}{\bf b} and \textcolor{blue}{\bf c},
yielding 5 misses on the requests to \textcolor{darkyellow}{\bf a} and \textcolor{darkgreen}{\bf d}.
These are the same decisions as taken by FOO-U and OPT.
We generally find that PFOO-U yields nearly identical miss ratios as FOO-U,
as we next demonstrate on real traces.



\subsection{Summary}\label{sec:pfoo:proof}

Putting it all together,
PFOO provides efficient lower and upper bounds on OPT with variable object sizes.
PFOO-L runs in $O(N \log N)$ time, as required to sort the intervals;
and PFOO-U runs in $O(N)$ because it divides the min-cost flow into segments of constant size.
In practice, PFOO-L is faster than PFOO-U at realistic trace lengths,
despite its worse asymptotic runtime,
due to the large constant factor in solving each segment in PFOO-U.

PFOO-L gives a lower bound ($\text{PFOO-L} \leq \text{FOO-L}$, Eq.~\eqref{eq:practical:ineq}) because
PFOO-L caches all of the smallest intervals,
so no policy can get more hits (i.e., cache more intervals) without exceeding $N \times C$ overall resources.
Resources are a hard constraint obeyed by all policies, including FOO-L,
whose resources are constrained by the capacity $C$ of all $N-1$ inner edges.


PFOO-U gives an upper bound ($\text{FOO-U} \leq \text{PFOO-U}$, Eq.~\eqref{eq:practical:ineq}) because FOO-U and PFOO-U both satisfy the capacity constraint and the integrality constraint of the ILP formulation of OPT (\autoref{def:ilp}).
By sequentially optimizing the MCF segment after segment, we induce additional constraints on PFOO-U, which can lead to suboptimal solutions.

\section{Experimental Methodology}
\label{sec:methodology}

We evaluate FOO and PFOO against prior offline bounds and online caching policies on eight different production traces.

\subsection{Trace characterization}\label{sec:app:traces}

We use production traces from three global content-distribution networks (CDNs), two web-applications from different anonymous large Internet companies, and storage workloads from Microsoft~\cite{snia}.
We summarize the trace characteristics in \autoref{tbl:methodology:traces}.
The table shows that our traces typically span several tens to hundreds of million of requests and tens of millions of objects.
\autoref{fig:methodology:traces} shows four key distributions of these workloads.

\begin{table}[h]
  \centering
  \begin{tabular}{cccccc}
    \toprule
    \bf Trace & \bf Year & \bf \# Requests & \bf \# Objects & \multicolumn{2}{c}{\bf Object sizes} \\
    \midrule
    CDN 1 & 2016 & 500\,M & 18\,M & {10\,B} -- {616\,MB} \\
    CDN 2 & 2015 & 440\,M & 19\,M & {1\,B} -- {1.5\,GB}\\
    CDN 3 & 2015 & 420\,M & 43\,M & {1\,B} -- {2.3\,GB}\\
    WebApp 1 & 2017 & 104\,M & 10\,M & {3\,B} -- {1.9\,MB}\\
    WebApp 2 & 2016 & 100\,M & 14\,M & {5\,B} -- {977\,KB}\\
    Storage 1 & 2008 & {29}\,M & 16\,M & {501\,B} -- {780\,KB}\\
    Storage 2  & 2008 & {37}\,M & 6\,M & {501\,B} -- {78\,KB}\\
    Storage 3  & 2008 & {45}\,M & 14\,M & {501\,B} -- {489\,KB}\\
    \bottomrule
  \end{tabular}
  \vspace{2mm}
  \caption{Overview of key properties of the production traces used in our evaluation.}
  \label{tbl:methodology:traces}
\end{table}

The object size distribution (\autoref{fig:methodology:traces:sizes}) shows that object sizes are variable in all traces.
However, while they span almost ten orders of magnitude in CDNs, object sizes vary only by six orders of magnitude in web applications, and only by three orders of magnitude in storage systems.
WebApp 1 also has noticeably smaller object sizes throughout, as is representative for application-cache workloads.

The popularity distribution (\autoref{fig:methodology:traces:popularity}) shows that CDN workloads and WebApp workloads all follow approximately a Zipf distribution with $\alpha$ between 0.85 and 1.
In contrast, the popularity distribution of storage traces is much more irregular with a set of disproportionally popular objects, an approximately log-linear middle part, and an exponential cutoff for the least popular objects.

The reuse distance distribution (\autoref{fig:methodology:traces:reuse-distance}) --- i.e., the distribution of the number of requests between requests to the same object --- further distinguishes CDN and WebApp traces from storage workloads.
CDNs and WebApps serve millions of different customers and so exhibit largely independent requests with smoothly diminishing object popularities, which matches our proof assumptions.
Thus, the CDN and WebApp traces lead to a smooth reuse distance, as shown in the figure.
In contrast, storage workloads serve requests from one or a few applications,
and so often exhibit highly correlated requests (producing spikes in the reuse distance distribution).
For example, scans are common in storage (e.g., traces like: \textsf{ABCD\,ABCD\,...}),
but never seen in CDNs.
This is evident from the figure, where the storage traces exhibit several steps in their cumulative request probability, as correlated objects (e.g., due to scans) have the same reuse distance.

Finally, we measure the correlation across different objects (\autoref{fig:methodology:traces:correlation}).
Ideally, we could directly test our independence assumption (Assumption~\ref{ass:localirm}).
Unfortunately, quantifying independence on real traces is challenging.
For example, classical methods such as Hoeffding's independence test \cite{hoeffding1948non} only apply to continuous distributions, whereas we consider cache requests in discrete time.

\figTraceChars

We therefore turn to correlation coefficients.
Specifically, we use the Pearson correlation coefficient as it is computable in linear time (as opposed to Spearman's rank and Kendall's tau~\cite{daniel1978applied}).
We define the coefficient based on the number of requests each object receives in a time bucket that spans 4000 requests (we verified that the results do not change significantly for time bucket sizes in the range 400 - 40k requests).
In order to capture pair-wise correlations, we chose the top 10k objects in each trace, calculated the request counts for all time buckets, and then calculated the Pearson coefficient for all possible combinations of object pairs.

\autoref{fig:methodology:traces:correlation} shows the distribution of coefficients for all object pair combinations.
We find that both CDN and WebApps do not show a significant correlation; over 95\% of the object pairs have a coefficient coefficient between -0.25 and 0.25.
In contrast, we find strong positive correlations in the storage traces.
For the first storage trace, we measure a correlation coefficient greater than $0.5$ for all 10k-most popular objects.
For the second storage trace, we measure a correlation coefficient greater than $0.5$ for more than 20\% of the object pairs.
And, for the third storage trace, we measure a correlation coefficient greater than $0.5$ for more than 14\% of the object pairs.
We conclude that the simple Pearson correlation coefficient is sufficiently powerful to quantify the correlation inherent to storage traces (such as loops and scans).


\subsection{Caching policies}
We evaluate three classes of policies: theoretical bounds on OPT, practical offline heuristics, and online caching policies.
Besides FOO, there exist three other theoretical bounds on OPT with approximation guarantees: OFMA, LocalRatio, and LP (\autoref{sec:bg:offlinevarsize}).
Besides PFOO-U, we consider three other practical upper bounds: Belady, Belady-Size, and Freq/Size (\autoref{sec:bg:practice}).
Besides PFOO-L, there is only one other practical lower bound: a cache with infinite capacity (Infinite-Cap).
Finally, for online policies,
we evaluated
GDSF~\cite{cherkasova1998improving}, GD-Wheel~\cite{li2015gd},
AdaptSize~\cite{berger2017adaptsize}, 
Hyperbolic~\cite{blankstein2017hyperbolic},
and several other older policies which perform much worse on our traces (including LRU-K~\cite{o1993lru},
TLFU~\cite{einziger2014tinylfu},
SLRU~\cite{huang2013analysis},
and LRU).

Our implementations are in C++ and use the COIN-OR::LEMON library~\cite{lemon}, GNU parallel~\cite{Tange2011a}, OpenMP~\cite{dagum1998openmp}, and CPLEX 12.6.1.0.
OFMA runs in $O(N^2)$, LocalRatio runs in $O(N^3)$, Belady in $O(N\log C)$.
We rely on sampling~\cite{psounis:infocom01:sampling} to run Belady-Size on large traces, which gives us an $O(N)$ implementation.
We will publicly release our policy implementations and evaluation infrastructure upon publication of this work.
To the best of our knowledge, these include the first implementations of the prior theoretical offline bounds.

\subsection{Evaluation metrics}
We compare each policy's \emph{miss ratio}, which ranges from 0 to 1.
We present absolute error as unitless scalars
and relative error as percentages.
For example, if FOO-U's miss ratio is $0.20$ and FOO-L's is $0.15$,
then absolute error is $0.05$ and relative error is $33$\%.

We present a \emph{miss ratio curve} for each trace,
which shows the miss ratio achieved by different policies at different cache capacities.
We present these curves in log-linear scale to study a wide range of cache capacities.
Miss ratio curves allow us to compare miss ratios achieved by different policies at fixed cache capacities,
as well as compute cache capacities that achieve equivalent miss ratios.

\section{Evaluation}
\label{sec:evaluation}

We evaluate FOO and PFOO to demonstrate the following:
\emph{(i)}~PFOO is fast enough to process real traces, whereas FOO and prior theoretical bounds are not;
\emph{(ii)}~FOO yields nearly tight bounds on OPT, even when our proof assumptions do not hold;
\emph{(iii)}~PFOO is highly accurate on full production traces;
and \emph{(iv)}~PFOO reveals that there is significantly more room for improving current caching systems than implied by prior offline bounds.

\subsection{PFOO is necessary to process real traces}


\autoref{fig:eval:runtime} shows the execution time of FOO, PFOO, and prior theoretical offline bounds
at different trace lengths.
Specifically, we run each policy on the first $N$ requests of the CDN 1 trace,
and vary $N$ from a few thousand to over 30 million.
Each policy ran alone on a 2016 SuperMicro server with 44 Intel Xeon E5-2699 cores and 500\,GB of memory.

\begin{figure}[h]
  \centering
 \includegraphics[width=0.5\linewidth]{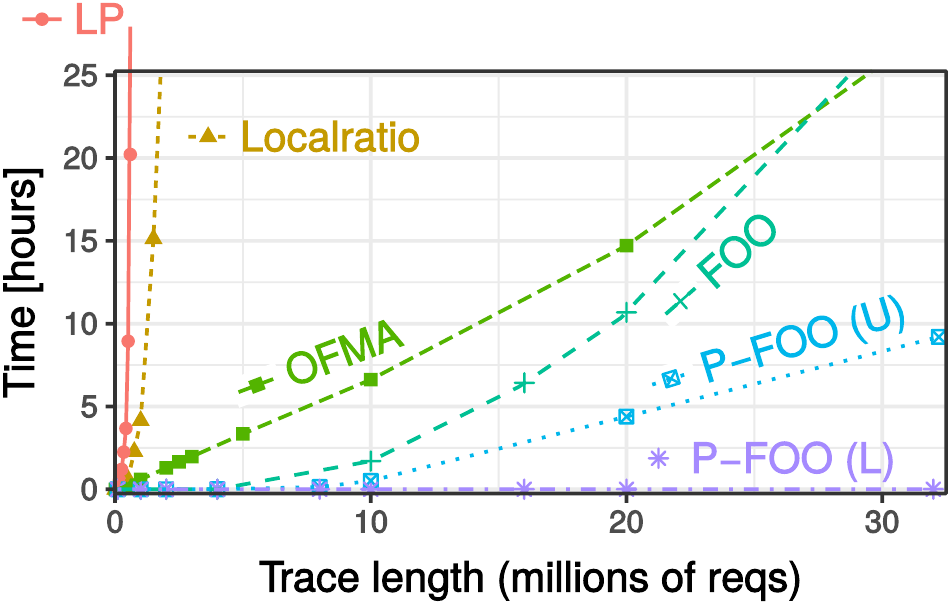}
\caption{Execution time of FOO, PFOO, and prior theoretical offline bounds at different trace lengths.
  Most prior bounds are unusable above 500\,K requests.
  Only PFOO can process real traces with many millions of requests.
}
\label{fig:eval:runtime}
\end{figure}

These results show that LP and LocalRatio are unusable:
they can process only a few hundred thousand requests in a 24-hour period,
and their execution time increases rapidly as traces lengthen.
While FOO and OFMA are faster, they both take more than 24 hours to process more than 30 million requests,
and their execution times increase super-linearly.

Finally, PFOO is much faster and scales well, allowing us to process traces with hundreds of millions of requests.
PFOO's lower bound completes in a few minutes,
and while PFOO's upper bound is slower,
it scales linearly with trace length.
\emph{PFOO is thus the only bound that completes in reasonable time on real traces.}


\subsection{FOO is nearly exact on short traces}

We compare FOO, PFOO, and prior theoretical upper bounds on the first 10 million requests of each trace.
Of the prior theoretical upper bounds, only OFMA runs in a reasonable time at this trace length,%
\footnote{While we have tried downsampling the traces to run LP and LocalRatio (as suggested in~\cite{kessler:tocs94:trace-sampling,beckmann:hpca15:talus,waldspurger:atc17:sampling} for objects with equal sizes), we were unable to achieve meaningful results. Under variable object sizes, scaling down the system (including the cache capacity), makes large objects disproportionately disruptive.}
so we compare FOO, PFOO, OFMA, the Belady variants, Freq/Size, and Infinite-Cap.


Our first finding is that \emph{FOO-U and FOO-L are nearly identical}, as predicted by our analysis.
The largest difference between FOO-U's and FOO-L's miss ratio on CDN and WebApp traces is 0.0005---a relative error of 0.15\%.
Even on the storage traces, where requests are highly correlated and hence our proof assumptions do not hold,
the largest difference is 0.0014---a relative error of 0.27\%.
Compared to the other offline bounds, FOO is at least an order of magnitude and often several orders of magnitude more accurate.

Given FOO's high accuracy, we use FOO to estimate the error of the other offline bounds.
Specifically, we assume that OPT lies in the middle between FOO-U and \mbox{FOO-L}.
Since the difference between FOO-U and FOO-L is so small, this adds negligible error (less than 0.14\%) to all other results.

\autoref{fig:eval:maxerror:cdn1} shows the maximum error from OPT across five cache sizes on our first CDN production trace.
All upper bounds are shown with a bar extending above OPT, and all lower bounds are shown with a bar extending below OPT.
Note that the practical offline upper bounds (e.g., Belady) do not have corresponding lower bounds.
Likewise, there is no upper bound corresponding to Infinite-Cap.
Also note that OFMA's bars are so large than they extend above and below the figure.
We have therefore annotated each bar with its absolute error from OPT.

\begin{figure}[h]
  \centering
  \begin{minipage}{0.54\linewidth}
    \vspace{-3mm}
    \caption{Comparison of the maximum approximation error of FOO, PFOO, and prior offline upper and lower bounds across five cache sizes on a CDN production trace.\\
      FOO's upper and lower bounds are nearly identical, with a gap of less than 0.0005.
      We therefore assume that OPT is halfway between FOO-U and FOO-L, which introduces negligible error due to FOO's high accuracy.\\
      PFOO likewise introduces small error, with PFOO-U having a similar accuracy to FOO-U.
      PFOO-L leads to higher errors than FOO-L but is still highly accurate, within 0.007.
      In contrast, all prior policies lead to errors several orders-of-magnitude larger.}
  \label{fig:eval:maxerror:cdn1}
  \end{minipage}
  \hfill
  \begin{minipage}{0.34\linewidth}
    \includegraphics[width=\linewidth]{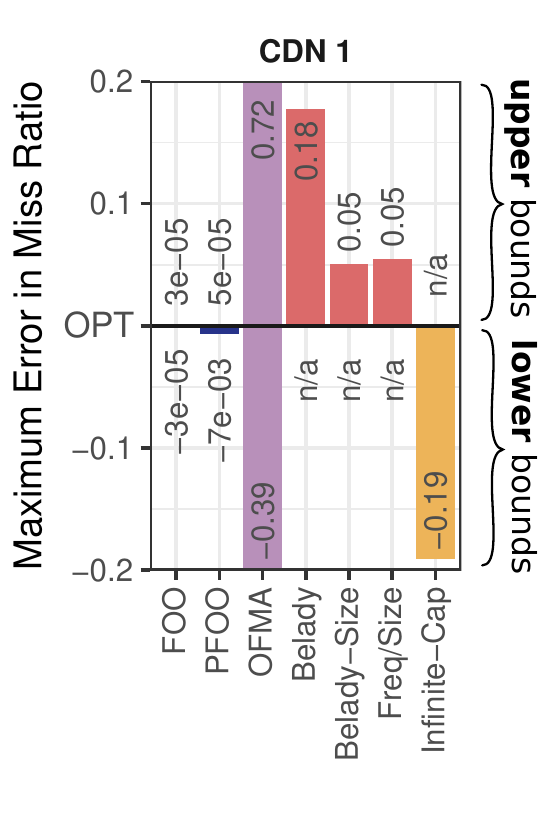}
  \end{minipage}
  \hspace{5mm}
  \vspace{-5mm}
\end{figure}

The figure shows that FOO-U and FOO-L nearly coincide, with error of $0.00003$ ($=$\num{3e-5}) on this trace.
PFOO-U is 0.00005 ($=$\num{5e-5}) above OPT, nearly matching FOO-U,
and PFOO-L is 0.007 below OPT, which is very accurate though worse than FOO-L.

All prior techniques yield error several orders of magnitude larger.
OFMA has very high error: its bounds are 0.72 above and 0.39 below OPT.
The practical upper bounds are more accurate than OFMA: Belady is 0.18 above OPT, Belady-Size 0.05, and Freq/Size 0.05.
Finally, Infinite-Cap is 0.19 below OPT.
Prior to FOO and PFOO, the best bounds for OPT give a broad range of up to 0.24.
\emph{PFOO and FOO reduce error by 34$\times$ and 4000$\times$, respectively.}

Figures~\ref{fig:eval:maxerror:all:split1} and~\ref{fig:eval:maxerror:all:split2} show the approximation error on all eight production traces.
The prior upper bounds are much worse than PFOO-U, except on one trace (Storage 1), where Belady-Size and Freq/Size are somewhat accurate.
Averaging across all traces, PFOO-U is 0.0014 above OPT.
PFOO-U reduces mean error by 37$\times$ over Belady-Size, the best prior upper bound.
Prior work gives even weaker lower bounds.
PFOO-L is on average 0.004 below OPT on the CDN traces, 0.02 below OPT on the WebApp traces, and 0.04 below OPT on the storage traces.
PFOO-L reduces mean error by 9.8$\times$ over Infinite-Cap and 27$\times$ over OFMA.
Hence, across a wide range of workloads, PFOO is by far the best practical bound on OPT.

\subsection{PFOO is accurate on real traces}

\begin{figure*}[t]
  \centering
  \includegraphics[width=0.98\textwidth]{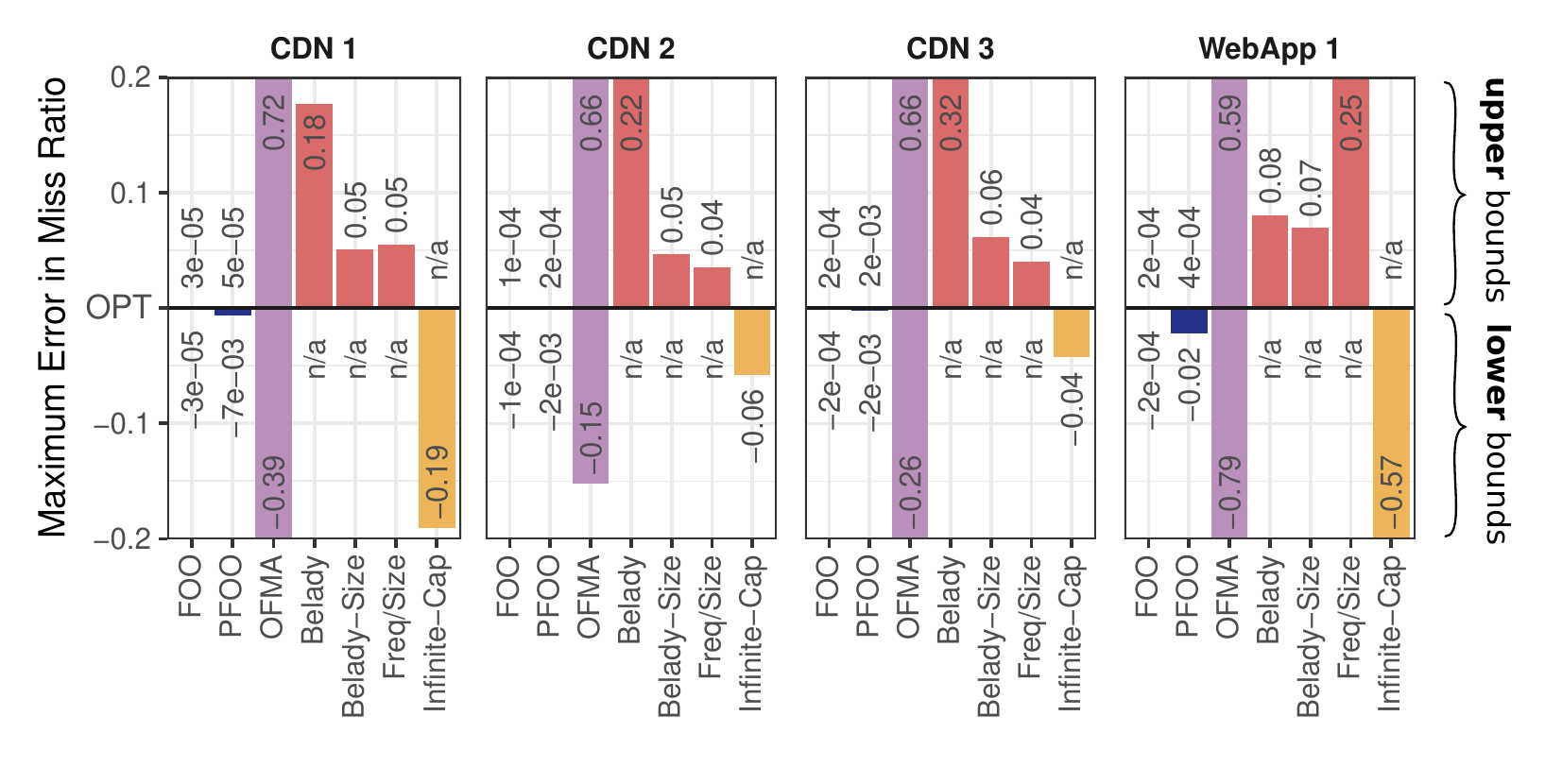}
  \caption{Approximation error of FOO, PFOO, and several prior offline bounds on four of our eight production traces (\autoref{fig:eval:maxerror:all:split2} shows the other four).
    FOO and PFOO's lower and upper bounds are orders of magnitude better than any other offline bound. (See \autoref{fig:eval:maxerror:cdn1}.)}
  \label{fig:eval:maxerror:all:split1}
\end{figure*}

Now that we have seen that FOO is accurate on short traces,
we next show that PFOO is accurate on long traces.
Figure~\ref{fig:eval:pfoo} 
shows the miss ratio over the full traces achieved by PFOO,
the best prior practical upper bounds (the best of Belady, Belady-Size, and Freq/Size),
the Infinite-Cap lower bound,
and the best online policy (see \autoref{sec:methodology}).

On average, PFOO-U and PFOO-L bound the optimal miss ratio within a narrow range
of 4.2\%.
PFOO's bounds are tighter on the CDN and WebApp traces than the storage traces:
PFOO gives an average bound of just 1.4\% on CDN 1-3 and 1.3\% on WebApp 1-2
but 5.7\% on Storage 1-3.
This is likely due to error in PFOO-L when requests are highly correlated, as they are in the storage traces.

Nevertheless, PFOO gives much tighter bounds than prior techniques on every trace.
The prior offline upper bounds are noticeably higher than PFOO-U.
On average, compared to PFOO-U, Belady-Size is 19\% higher,
Freq/Size is 22\% higher,
and Belady is fully 72\% higher.
These prior upper bounds are not, therefore, good proxies for the offline optimal.
Moreover, the best upper bound varies across traces: Freq/Size is lower on CDN 1 and CDN 3, but Belady-Size is lower on the others.
Unmodified Belady gives a very poor upper bound, showing that caching policies must account for object size.
The only lower bound in prior work is an infinitely large cache, whose miss ratio is much lower than PFOO-L.
\emph{PFOO thus gives the first reasonably tight bounds on the offline miss ratio for real traces.}

\begin{figure*}[t]
  \centering
  \includegraphics[width=0.98\textwidth]{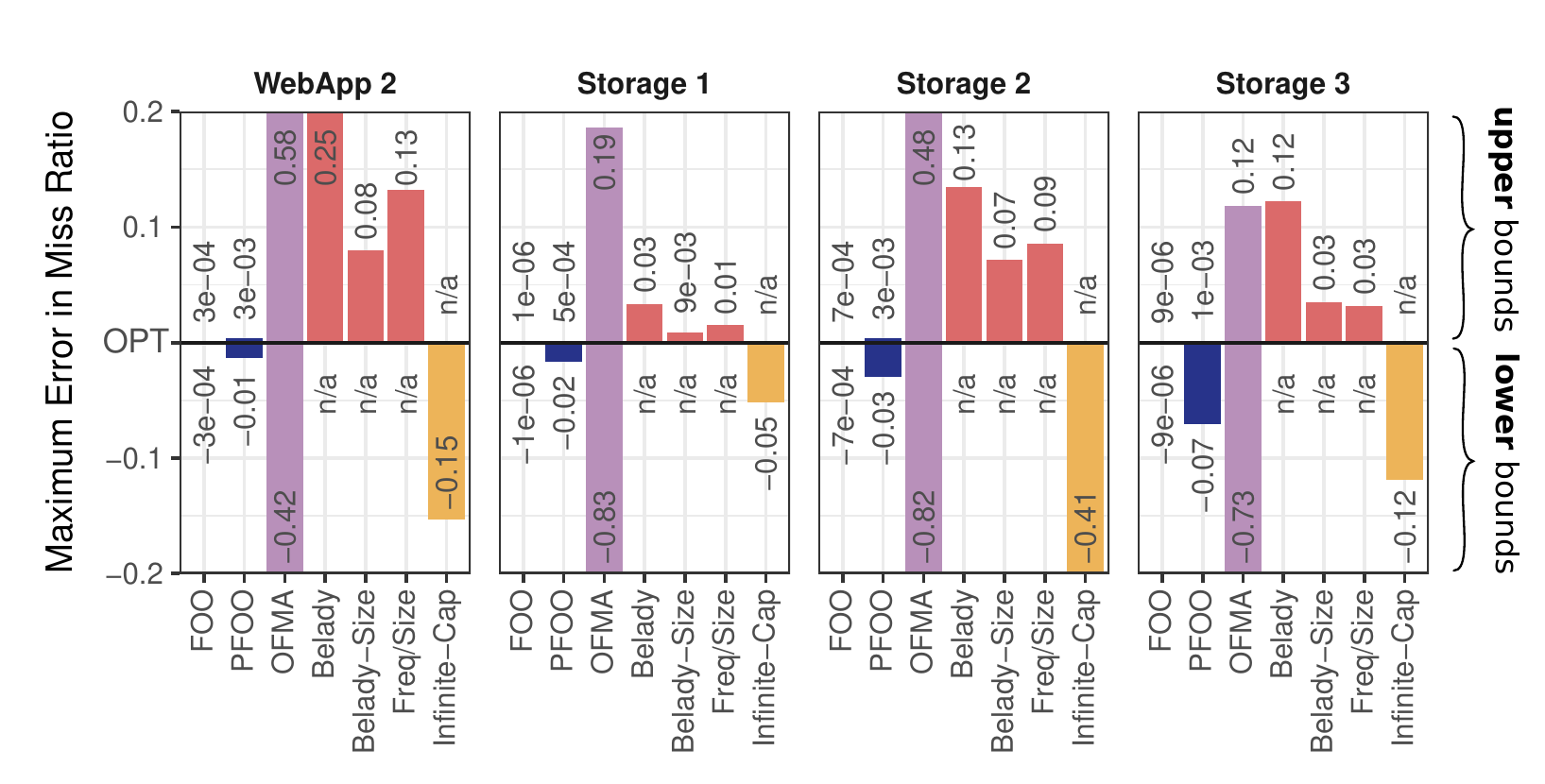}
  \caption{Approximation error of FOO, PFOO, and several prior offline bounds on four of our eight production traces (\autoref{fig:eval:maxerror:all:split1} shows the other four).
    FOO and PFOO's lower and upper bounds are orders of magnitude better than any other offline bound. (See \autoref{fig:eval:maxerror:cdn1}.)}
  \label{fig:eval:maxerror:all:split2}
\end{figure*}


\begin{figure*}[p]
  \centering
    \begin{subfigure}[b]{.35\textwidth}
      \centering
      \includegraphics[width=\linewidth]{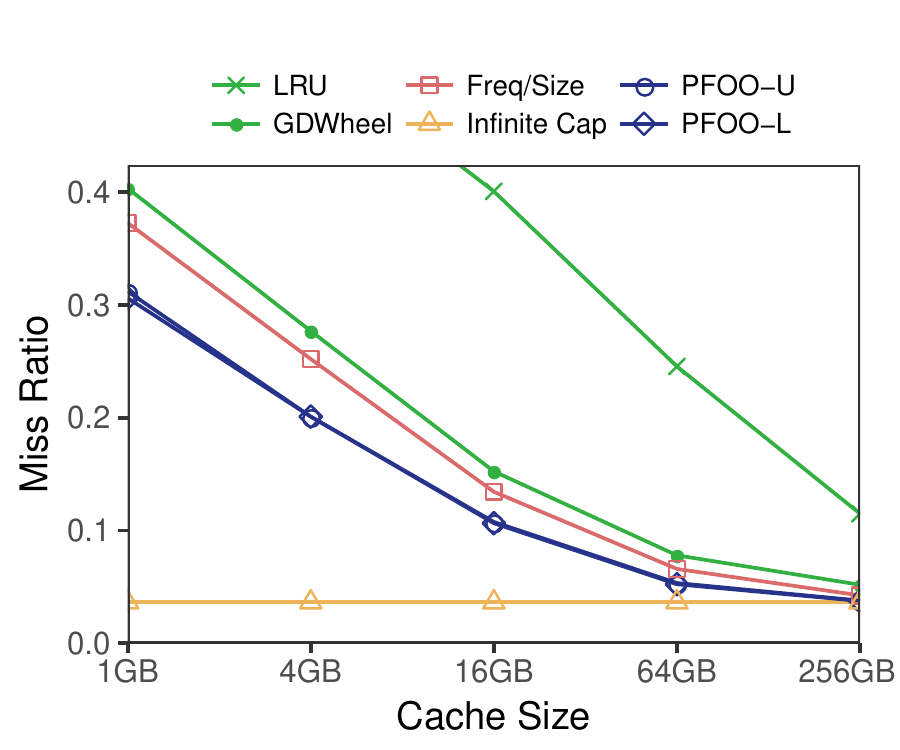} 
      \caption{CDN 1}
      \label{fig:eval:pfoo:cdn1}
    \end{subfigure}
    \hspace{4mm}
    \begin{subfigure}[b]{.35\textwidth}
      \centering
      \includegraphics[width=\linewidth]{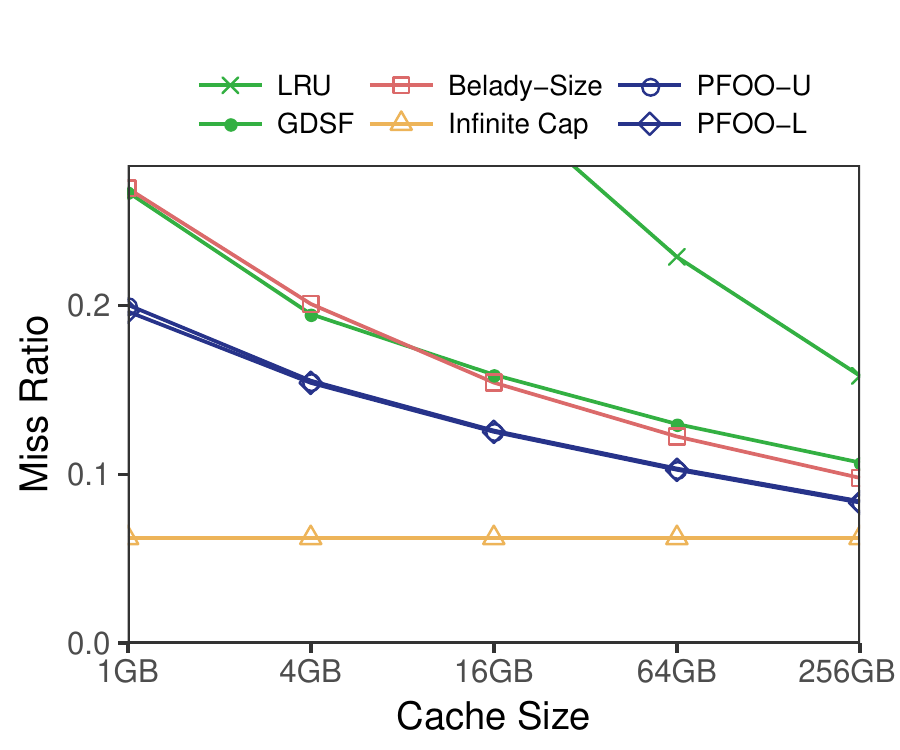} 
      \caption{CDN 2}
      \label{fig:eval:pfoo:cdn2}
    \end{subfigure}
    \begin{subfigure}[b]{.35\textwidth}
      \centering
      \includegraphics[width=\linewidth]{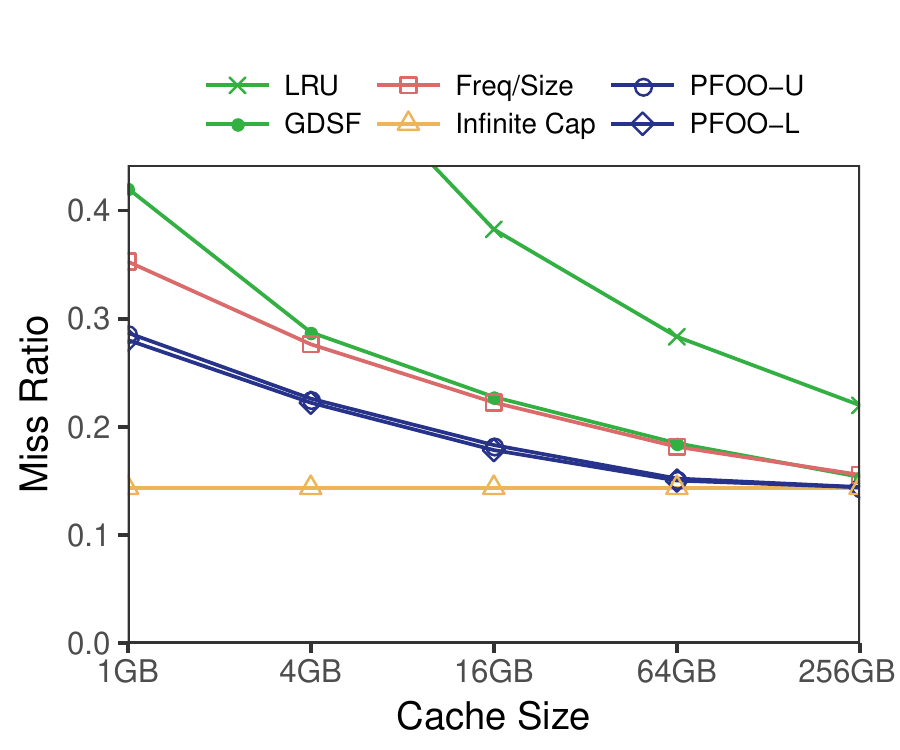}
      \caption{CDN 3}
      \label{fig:eval:pfoo:cdn3}
    \end{subfigure}
    \hspace{4mm}
    \begin{subfigure}[b]{.35\textwidth}
      \centering
      \includegraphics[width=\linewidth]{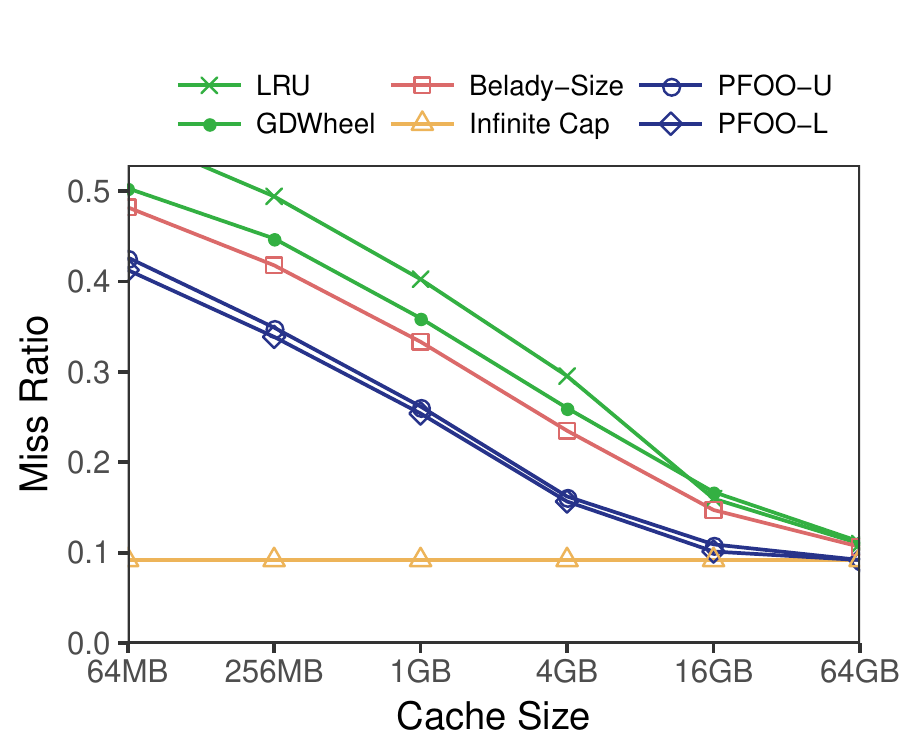}
      \caption{WebApp 1}
      \label{fig:eval:pfoo:webapp1}
    \end{subfigure}

  \centering
  \begin{subfigure}[b]{.35\textwidth}
    \centering
    \includegraphics[width=\linewidth]{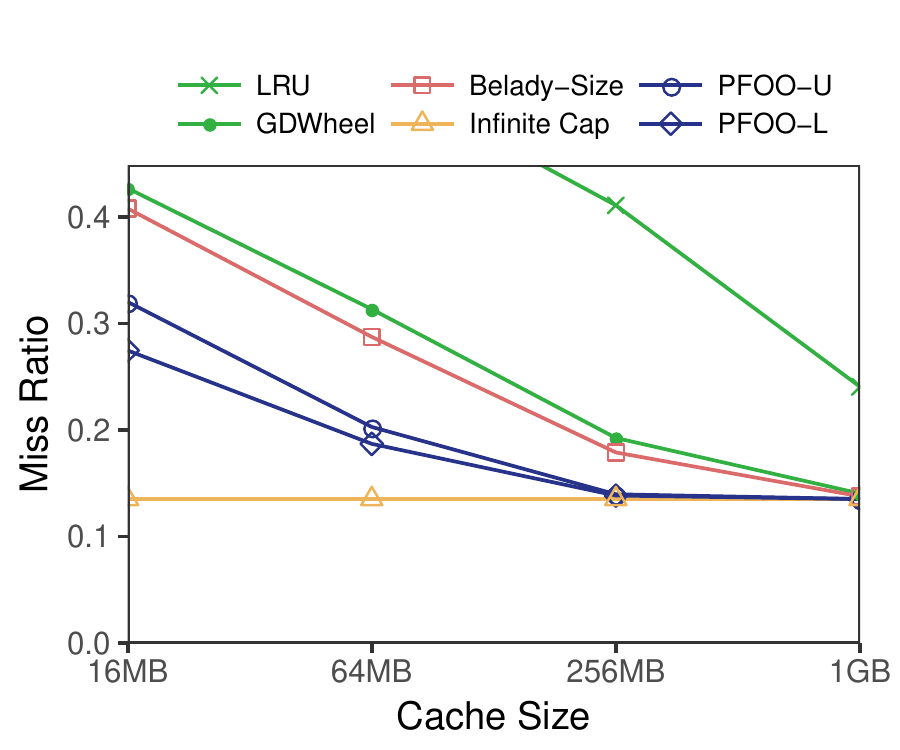}
    \caption{WebApp 2}
    \label{fig:eval:pfoo:webapp2}
  \end{subfigure}
  \hspace{4mm}
  \begin{subfigure}[b]{.35\textwidth}
    \centering
    \includegraphics[width=\linewidth]{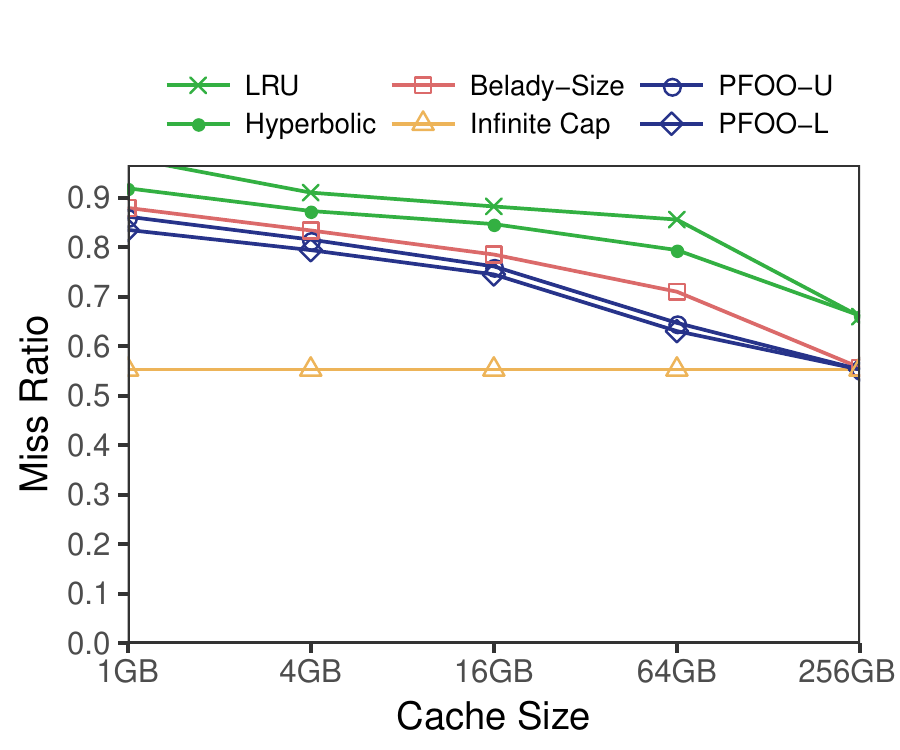}
    \caption{Storage 1}
    \label{fig:eval:pfoo:storage1}
  \end{subfigure}
  \begin{subfigure}[b]{.35\textwidth}
    \centering
    \includegraphics[width=\linewidth]{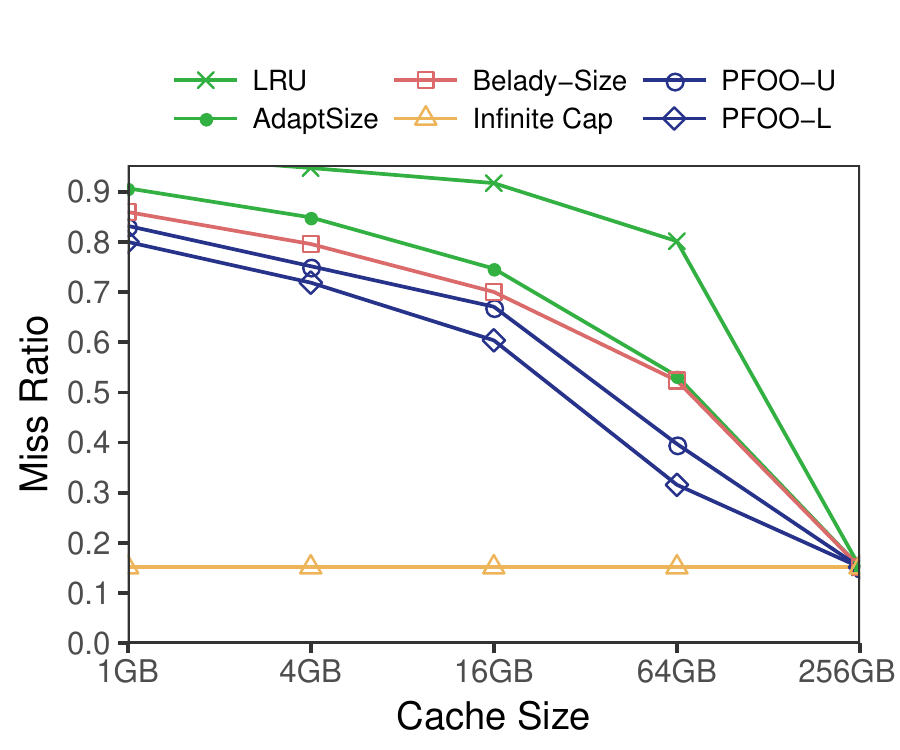}
    \caption{Storage 2}
    \label{fig:eval:pfoo:storage2}
  \end{subfigure}
  \hspace{4mm}
  \begin{subfigure}[b]{.35\textwidth}
    \centering
    \includegraphics[width=\linewidth]{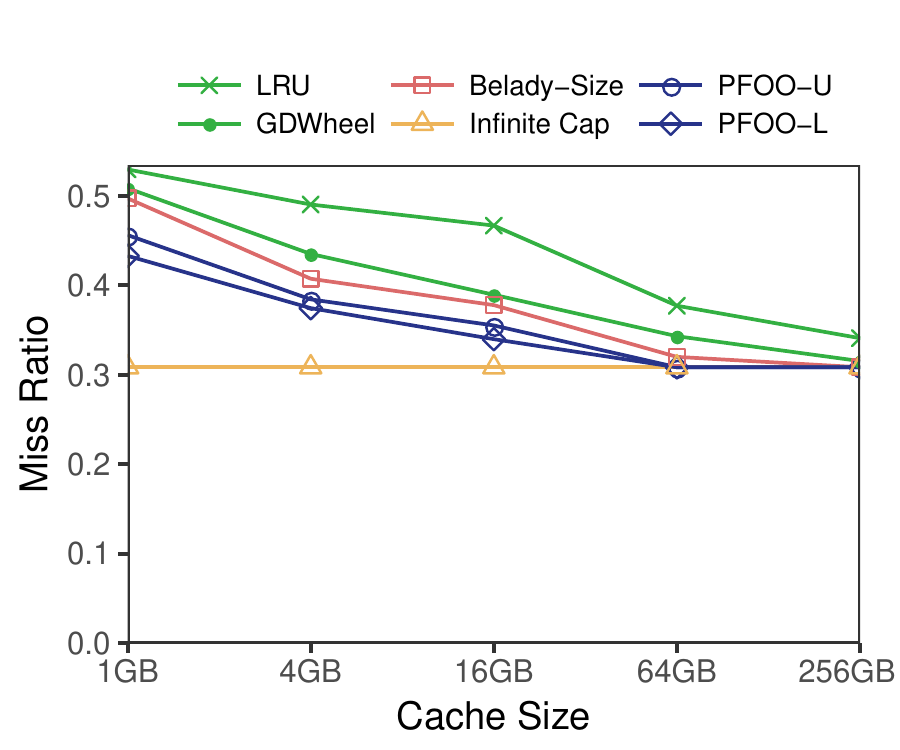}
    \caption{Storage 3}
    \label{fig:eval:pfoo:storage3}
  \end{subfigure}
  \caption{Miss ratio curves for PFOO vs.\ LRU, Infinite-Cap, the best prior offline upper bound, and the best online policy for our eight production traces.}
  \label{fig:eval:pfoo}
\end{figure*}

\subsection{PFOO shows that there is significant room for improvement in online policies}

Finally, we compare with online caching policies.
\autoref{fig:eval:pfoo}
show the best online policy (by average miss ratio) and the offline bounds for each trace.
We also show LRU for reference on all traces.

On all traces at most cache capacities, there is a large gap between the best online policy and PFOO-U,
showing that there remains significant room for improvement in online caching policies.
Moreover, \emph{this gap is much larger than prior offline bounds would suggest.}
On average, PFOO-U achieves 27\% fewer misses than the best online policy,
whereas the best prior offline policy achieves only 7.2\% fewer misses;
the miss ratio gap between online policies and offline optimal is thus 3.75$\times$ as large as implied by prior bounds.
The storage traces are the only ones where PFOO does not consistently increase this gap vs.\ prior offline bounds,
but even on these traces there is a large difference at some sizes (e.g., at 64\,GB in \autoref{fig:eval:pfoo:storage2}).
On CDN and WebApp traces, the gap is much larger.

For example, on CDN 2, GDSF (the best online policy) matches Belady-Size (the best prior offline upper bound) at most cache capacities.
One would therefore conclude that existing online policies are nearly optimal,
but PFOO-U reveals that there is in fact a large gap between GDSF and OPT on this trace,
as it is 21\% lower on average (refer back to \autoref{fig:intro}).

These miss ratio reductions make a large difference in real systems.
For example, on CDN 2, CDN 3, and WebApp 1, OPT requires just 16\,GB to match the miss ratio of the best prior offline bound at 64\,GB (the $x$-axis in these figures is shown in log-scale).
Prior bounds thus suggest that online policies require 4$\times$ as much cache space as is necessary.



\balance

\section{Conclusion and future work}
\label{sec:conclusion}

We began this paper by asking: \emph{Should the systems community continue trying to improve miss ratios, or have all achievable gains been exhausted?}
We have answered this question by developing new techniques, FOO and PFOO, to accurately and quickly estimate OPT with variable object sizes.
Our techniques reveal that prior bounds for OPT lead to qualitatively wrong conclusions about the potential for improving current caching systems.
Prior bounds indicate that current systems are nearly optimal,
whereas PFOO reveals that misses can be reduced by up to 43\%.

Moreover, since FOO and PFOO yield constructive solutions,
they offer a new path to improve caching systems.
While it is outside the scope of this work,
we plan to investigate adaptive caching systems that tune their parameters online by learning from PFOO, 
e.g., by recording a short window of past requests and running PFOO to estimate OPT's behavior in real time.
For example, AdaptSize, a caching system from 2017, assumes that admission probability should decay exponentially with object size.
This is unjustified, and we can learn an optimized admission function from PFOO's decisions.

This paper gives the first principled way to evaluate caching policies with variable object sizes:
FOO gives the first asymptotically exact, polynomial-time bounds on OPT,
and PFOO gives the first practical and accurate bounds for long traces.
Furthermore, our results are verified on eight production traces from several large Internet companies, including CDN, web application, and storage workloads,
where FOO reduces approximation error by 4000$\times$.
We anticipate that FOO and PFOO will prove important tools
in the design of future caching systems.


\section{Acknowledgments}
We thank Jens Schmitt, as well as our ACM POMACS shepherd Niklas Carlsson and the anonymous reviewers for their valuable feedback.

Daniel S. Berger was supported by the Mark Stehlik Postdoctoral Teaching Fellowship of the School of Computer Science at Carnegie Mellon University.
Nathan Beckmann was supported by a Google Faculty Research Award.
Mor Harchol-Balter was supported by NSF-CMMI-1538204 and NSF-XPS-1629444.

{\small
\bibliographystyle{abbrv}
\bibliography{../refs,../daniel}
}

\appendix
\vspace{3mm}
\section{Proofs}

\subsection{Proof of equivalence of interval and classic ILP representations of OPT}\label{sec:proof:lma:opt-equivalence}

\paragraph{Classic representation of OPT}\label{sec:foo:ilp}

The literature uses an integer linear program (ILP) representation of OPT~\cite{albers1999page}.
\autoref{fig:foo:classic:ilp} shows this classic ILP representation on the example trace from \autoref{sec:foo}.
The ILP uses decision variables $x_{i,k}$ to track at each time $i$ whether object $k$ is cached or not.
The constraint on the cache capacity is naturally represented: the sum of the sizes for all cached objects must be less than the cache capacity for every time $i$.
Additional constraints enforce that OPT is not allowed to prefetch objects (decision variables must not increase if the corresponding object is not requested) and that the cache starts empty.

\begin{figure}[h]
  \centering
  \begin{tabular}{p{8mm} c}
  \toprule
  Object
  &
  \large
  \bf
  \makebox[0.22in][l]{\textcolor{darkyellow}{a}}
  \makebox[0.22in][l]{\textcolor{darkred}{b}}
  \makebox[0.22in][l]{\textcolor{blue}{c}}
  \makebox[0.22in][l]{\textcolor{darkred}{b}}
  \makebox[0.22in][l]{\textcolor{darkgreen}{d}}
  \makebox[0.22in][l]{\textcolor{darkyellow}{a}}
  \makebox[0.22in][l]{\textcolor{blue}{c}}
  \makebox[0.22in][l]{\textcolor{darkgreen}{d}}
  \makebox[0.22in][l]{\textcolor{darkyellow}{a}}
  \makebox[0.24in][l]{\textcolor{darkred}{b}}
  \makebox[0.05in][l]{\textcolor{darkyellow}{a}\dots}
    \\
    \parbox[t]{8mm}{\multirow{4}{*}{\rotatebox[origin=c]{90}{
    \begin{minipage}[t]{8mm}
      Decision\\ Variables
    \end{minipage}
    \hspace{1mm}
    }}}
  &
    \makebox[0.22in][l]{\textcolor{darkyellow}{$x_{1,a}$}}
    \makebox[0.22in][l]{\textcolor{darkyellow}{$x_{2,a}$}}
    \makebox[0.22in][l]{\textcolor{darkyellow}{$x_{3,a}$}}
    \makebox[0.22in][l]{\textcolor{darkyellow}{$x_{4,a}$}}
    \makebox[0.22in][l]{\textcolor{darkyellow}{$x_{5,a}$}}
    \makebox[0.22in][l]{\textcolor{darkyellow}{$x_{6,a}$}}
    \makebox[0.22in][l]{\textcolor{darkyellow}{$x_{7,a}$}}
    \makebox[0.22in][l]{\textcolor{darkyellow}{$x_{8,a}$}}
    \makebox[0.22in][l]{\textcolor{darkyellow}{$x_{9,a}$}}
    \makebox[0.24in][l]{\textcolor{darkyellow}{$x_{10,a}$}}
    \makebox[0.22in][l]{\textcolor{darkyellow}{$x_{11,a}$}}
  \\
    &
    \makebox[0.22in][l]{\textcolor{darkred}{$x_{1,b}$}}
    \makebox[0.22in][l]{\textcolor{darkred}{$x_{2,b}$}}
    \makebox[0.22in][l]{\textcolor{darkred}{$x_{3,b}$}}
    \makebox[0.22in][l]{\textcolor{darkred}{$x_{4,b}$}}
    \makebox[0.22in][l]{\textcolor{darkred}{$x_{5,b}$}}
    \makebox[0.22in][l]{\textcolor{darkred}{$x_{6,b}$}}
    \makebox[0.22in][l]{\textcolor{darkred}{$x_{7,b}$}}
    \makebox[0.22in][l]{\textcolor{darkred}{$x_{8,b}$}}
    \makebox[0.22in][l]{\textcolor{darkred}{$x_{9,b}$}}
    \makebox[0.24in][l]{\textcolor{darkred}{$x_{10,b}$}}
    \makebox[0.22in][l]{\textcolor{darkred}{$x_{11,b}$}}
    \\
        &
    \makebox[0.22in][l]{\textcolor{blue}{$x_{1,c}$}}
    \makebox[0.22in][l]{\textcolor{blue}{$x_{2,c}$}}
    \makebox[0.22in][l]{\textcolor{blue}{$x_{3,c}$}}
    \makebox[0.22in][l]{\textcolor{blue}{$x_{4,c}$}}
    \makebox[0.22in][l]{\textcolor{blue}{$x_{5,c}$}}
    \makebox[0.22in][l]{\textcolor{blue}{$x_{6,c}$}}
    \makebox[0.22in][l]{\textcolor{blue}{$x_{7,c}$}}
    \makebox[0.22in][l]{\textcolor{blue}{$x_{8,c}$}}
    \makebox[0.22in][l]{\textcolor{blue}{$x_{9,c}$}}
    \makebox[0.24in][l]{\textcolor{blue}{$x_{10,c}$}}
    \makebox[0.22in][l]{\textcolor{blue}{$x_{11,c}$}}
    \\
        &
    \makebox[0.22in][l]{\textcolor{darkgreen}{$x_{1,d}$}}
    \makebox[0.22in][l]{\textcolor{darkgreen}{$x_{2,d}$}}
    \makebox[0.22in][l]{\textcolor{darkgreen}{$x_{3,d}$}}
    \makebox[0.22in][l]{\textcolor{darkgreen}{$x_{4,d}$}}
    \makebox[0.22in][l]{\textcolor{darkgreen}{$x_{5,d}$}}
    \makebox[0.22in][l]{\textcolor{darkgreen}{$x_{6,d}$}}
    \makebox[0.22in][l]{\textcolor{darkgreen}{$x_{7,d}$}}
    \makebox[0.22in][l]{\textcolor{darkgreen}{$x_{8,d}$}}
    \makebox[0.22in][l]{\textcolor{darkgreen}{$x_{9,d}$}}
    \makebox[0.24in][l]{\textcolor{darkgreen}{$x_{10,d}$}}
    \makebox[0.22in][l]{\textcolor{darkgreen}{$x_{11,d}$}}
    \\
  \bottomrule
  \end{tabular}
      \vspace{-1em}
  \caption{Classic ILP representation of OPT.}
  \label{fig:foo:classic:ilp}
\end{figure}

\begin{lma}
  Under Assumption~\ref{ass:prefetching}, our ILP in Definition~\ref{def:ilp} is equivalent to the classical ILP from ~\cite{albers1999page}.
\end{lma}

\begin{proof}[Proof sketch]
Under Assumption~\ref{ass:prefetching}, OPT changes the caching decision of object $k$ only at times $i$ when $\sigma_i=k$.
To see why this is true, let us consider the two cases of changing a decision variable\ $x_{k,j}$ for $i<j<\ell_i$.
If $x_{k,i}=0$, then OPT cannot set $x_{k,j}=1$ because this would violate Assumption~\ref{ass:prefetching}.
Similarly, if $x_{k,j} = 0$, then setting $x_{k,i}=1$ does not yield any fewer misses, so we can safely assume that $x_{k, i} = 0$.
Hence, decisions do not change within an interval in the classic ILP formulation.

To obtain the decision variables $x'_{p,i}$ of the classical ILP formulation of OPT from a given solution $x_i$ for the interval ILP, set $x'_{\sigma_i,j} = x_i$ for all $i\leq j < \ell_i$, and for all $i$.
This leads to an equivalent solution 
because the capacity constraint is enforced at every time step.
\end{proof}

\subsection{Proof of \autoref{lma:couponinterpretation}}\label{sec:proof:lma:couponinterpretation}

This section proves \autoref{lma:couponinterpretation} from \autoref{sec:opt:coupons}, which bounds $T_{B}$, i.e., the first time after $i$ that all intervals in $B$ are completed.

Recall the definitions of the generalized CCP, where coupon probabilities follow a distribution $p$, and the definition of the classical CCP, where coupons have equal request probability.

We will make use of the following result, which immediately follows from~\cite[Theorem 4, p.\ 415]{anceaume2015new}.
\begin{theorem}\label{thm:uniformcoupon}
 For $b\geq 0$ coupons, any probability vector $p$, and the equal-probability vector $q=(1/b,\dots,1/b)$, it holds that 
 $\Pb{T_{b,p} < l} \leq \Pb{T_{b,q} < l}$ for any $l\geq 0$.
\end{theorem}

We will use this result to prove a bound on $T_{B}$, which is defined as $T_{B} = \max_{j \in B}\ell_j-i$.
The proof bounds $T_{B}$ first using a generalized CCP and then a classical CCP.

\begin{proof}[Proof of \autoref{lma:couponinterpretation}]

We first bound $T_{B}$ via a generalized CCP with stopping time $T_{b,p}$ and $p=(p_1,\dots,p_{b})$ with
  \begin{align}
    p_i = \Pb{\text{object }k\text{ is requested }\vert \;k\in B}\quad\text{ under }\mathcal{P}^M~~.
  \end{align}
  As $T_{B}$ contains requests to other objects $j\notin B$, it always holds that $T_{B}\geq T_{b,p}$.
  \autoref{fig:coupon:mapping} shows such a case, where $T_{B}$ is extended because of requests to uncached objects and large cached objects.
  $T_{b,p}$, on the other hand, does not include these other requests, and is thus always shorter or equal to $T_{B}$.
  This inequality bounds the probabilities for all $l\geq 0$.
  \begin{align} 
    \Pb{T_{B} < l} & \leq \Pb{T_{b,p} < l} \\
    \intertext{We then apply \autoref{thm:uniformcoupon}.}
    & \leq \Pb{T_{b,q} < l}
  \end{align}
\end{proof}

  \begin{figure}[h]
    \centering
    \includegraphics[width=.6\linewidth]{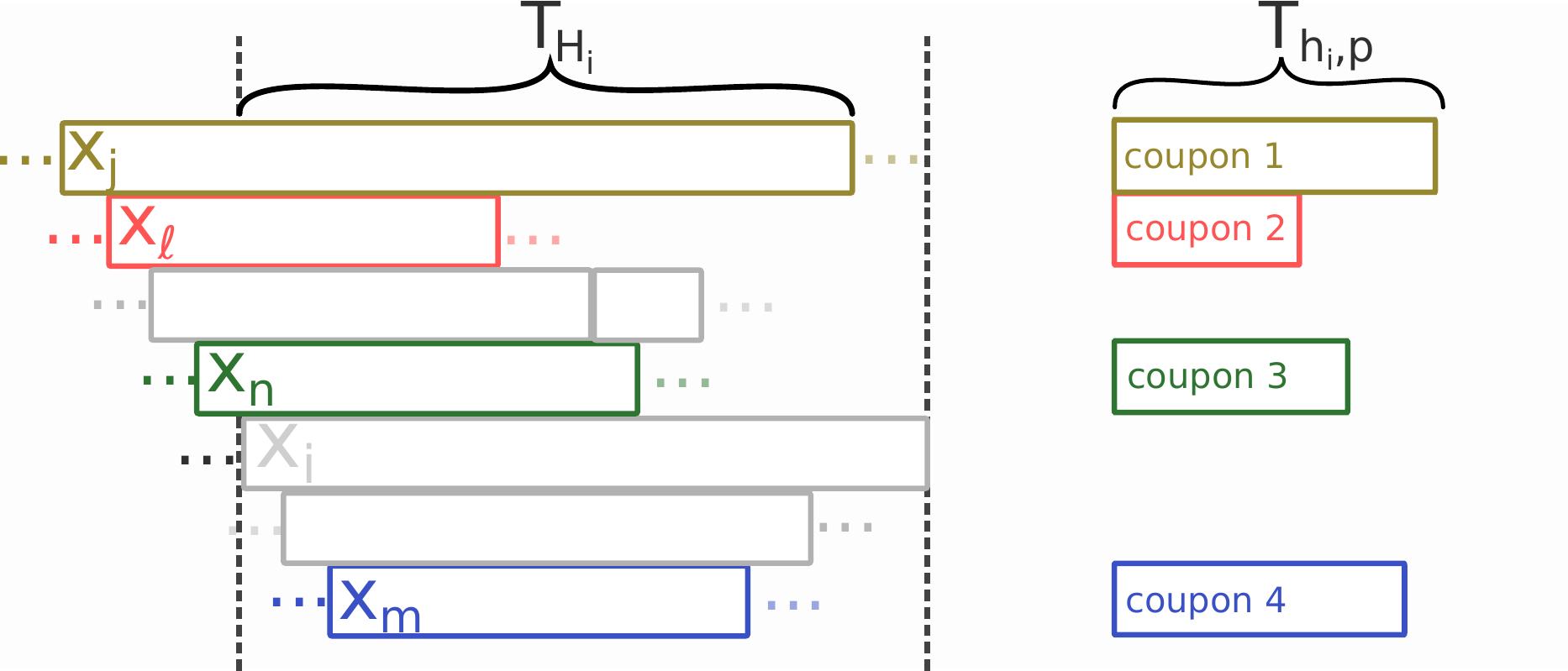}
    \caption{Translation of the time until all $B$ objects are requested once, $T_{B}$ into a coupon-collector problem (CCP), $T_{b,p}$. As the CCP is based on fewer coupons (only objects $\in B$), the CCP serves as a lower bound on $T_{B}$.}
    \label{fig:coupon:mapping}
  \end{figure}

\subsection{Proof of \autoref{lma:fractionprobability}}\label{sec:proof:lma:fractionprobability}

This section proves \autoref{lma:fractionprobability} from \autoref{sec:opt:stochastic}.
Recall the definition of large and unpopular objects (\autoref{def:popular}).
At a high level, the proof shows that the remaining, ``typical'' objects are almost always part of a precedence relation.
As a result, the probability of non-integer decision variables for typical objects vanishes as the number of objects $M$ grows large.

Before we state the proof of \autoref{lma:fractionprobability}, we introduce two auxiliary results, which are proven in Appendices~\ref{sec:ap:proof:lma:hi:infinite} and~\ref{sec:proof:lemma:ccpbound}).

The first auxiliary result shows that the number of cached objects $h_i$ goes to infinity as the cache capacity $C$ and the number of objects $M$ go to infinity.

\begin{lma}[Proof in \autoref{sec:ap:proof:lma:hi:infinite}]\label{lma:hi:infinite}
  For $i>N^*$ from Definition~\ref{def:popular}, $\Pb{h_i \rightarrow \infty}=1$ as $M\rightarrow \infty$.
\end{lma}

The second auxiliary result derives an exponential bound on the lower tail of the distribution of the coupon collector stopping time
as it gets further from its mean (roughly $b_i \log b_i$).

\begin{lma}[Proof in \autoref{sec:proof:lemma:ccpbound}]\label{lma:ccpbound}
  The time $T_{b,q}$ to collect $b>1$ coupons, which have equal probabilities $q=(1/b,\dots,1/b)$, is lower bounded by
  \begin{align}
    \Pb{T_{b,q} \leq b \log b - c\,b} < e^{-c}\quad\text{ for all }c>0~~.
  \end{align}
\end{lma}

With these results in place, we are ready to prove \autoref{lma:fractionprobability}.
This proof proceeds by using elementary probability theory and exploits our previous definitions of $L_i$, $T_{B_i}$, and $T_{b_i,q}$.

\begin{proof}[Proof of \autoref{lma:fractionprobability}]

  We know from \autoref{lma:stochastic:bound} that the probability of non-integer decision variables can be upper bounded using the random variables $L_i$ and $T_{B_i}$.
  \begin{align}
    \Pb{0 < x_i < 1} & \leq \Pb{L_i > T_{B_i}}\\
    \intertext{We expand this expression by conditioning on $L_i=l$.}
                     & = \sum_{l=1}^\infty \Pb{T_{B_i} < l | L_i = l} \Pb{L_i = l}\\
    \intertext{We observe that $\Pb{T_{B_i} < l} = 0$ for $l\leq b_i$ because requesting $b_i$ distinct objects takes at least $b_i$ time steps.}
                     & = \sum_{l=b_i+1}^\infty \Pb{T_{B_i} < l | L_i = l} \Pb{L_i = l}\\
    \intertext{We use the fact that conditioned on $L_i=l$, events $\{L_i=l\}$ and $\{T_{B_i}<l\}$ are stochastically independent.}
                     & = \sum_{l=b_i+1}^\infty \Pb{T_{B_i} < l} \Pb{L_i = l}\\
    \intertext{We split this sum into two parts, $l\leq \Lambda$ and $l>\Lambda$, where $\Lambda=\frac{1}{2}b_i \log b_i$ is chosen such that $\Lambda$ scales slower than the expectation of the underlying coupon collector problem with $b_i = \delta h_i$ coupons. (Recall that $\delta = \vert B_i \vert / \vert H_i \vert$ is the largest fraction of objects in $H_i$, defined in \autoref{def:popular}.)}
                     & \leq \sum_{l=b_i}^{\Lambda} \Pb{T_{B_i} < l} \Pb{L_i = l} \\
                     & \quad + \sum_{l=\Lambda+1}^{\infty}  \Pb{T_{B_i} < l} \Pb{L_i = l}\\
\label{eq:sumpart} & \leq \sum_{l=b_i}^{\Lambda} \Pb{T_{B_i} < l} + \sum_{l=\Lambda+1}^{\infty} \Pb{L_i = l}
\end{align}

\vspace{3mm}
We now bound the two terms in Eq.~\eqref{eq:sumpart}, separately.
For the first term, we start by applying \autoref{lma:couponinterpretation}.
\begin{align}
    \sum_{l=b_i}^{\Lambda} \Pb{T_{B_i} < l} 
  & \leq  \sum_{l=b_i}^{\Lambda} \Pb{T_{b_i,q} \leq l} \\
  \intertext{We rearrange the sum (replacing $l$ by $c$).}
  & = \sum_{c=\frac{1}{2}\log b_i}^{1+\log b_i} \Pb{T_{b_i,q}\leq b_i\log b_i - c\,b_i} \\
  \intertext{We apply \autoref{lma:ccpbound}.}
  & < \sum_{c=\frac{1}{2}\log b_i}^{1+\log b_i} e^{-c}
  \intertext{We solve the finite exponential sum.}
  \label{eq:sumpart1:final} &
                              = \frac{e^2}{e^2-e}\,\frac{1}{\sqrt{b_i}}
\end{align}

\vspace{3mm}
For the second term in Eq.~\eqref{eq:sumpart}, we use the fact that $L_i$'s distribution is Geometric($\rho_{\sigma_i}$) due to Assumption~\ref{ass:localirm}.
\begin{align}
 \sum_{l=\Lambda+1}^{\infty} \Pb{L_i = l} &
                                            = \sum_{l=\Lambda+1}^{\infty} \left(1 -\rho_{\sigma_i} \right)^{l-1}\;\rho_{\sigma_i} \\
  \intertext{We solve the finite sum.}
                                          & = \left( 1-\rho_{\sigma_i}\right)^{\Lambda}\\
  \intertext{We apply Definition~\ref{def:popular}, i.e., $\rho_{\sigma_i} \geq \frac{1}{b_i \log \log b_i}$.}
\label{eq:sumpart2:final}& \leq \left( 1-\frac{1}{b_i \log \log b_i}\right)^{\frac{1}{2}b_i \log b_i}
\end{align}

Finally, combining Eqs.~\eqref{eq:sumpart1:final} and~\eqref{eq:sumpart2:final} yields the following.
\begin{align}
  \Pb{L_i \geq T_{b_i,q}} & < \frac{e^2}{e^2-e}\,\frac{1}{\sqrt{b_i}} + \left( 1-\frac{1}{b_{i} \log \log b_{i}}\right)^{b_i \log b_i}
\end{align}
As $b_i \log\, b_i$ grows faster than $b_i \log \log b_i$, this proves the statement $\Pb{L_i \geq T_{b_i,q}}\rightarrow 0$ as $h_i \rightarrow\infty$ (implying that $b_i = \delta h_i \rightarrow \infty$) due to \autoref{lma:hi:infinite}.
\end{proof}


  \subsection{Proof of \autoref{lma:hi:infinite}}\label{sec:ap:proof:lma:hi:infinite}

  This section proves that, for any time $i>N^*$, the number of cached objects $h_i$ grows infinitely large as the number of objects $M$ goes to infinity.
  Recall that $N^*$ denotes the time after which the cache needs to evict at least one object.
  Throughout the proof, let $\overline{E}$ denote the complementary event of an event $E$.

  \emph{Intuition of the proof.}
  The proof exploits the fact that at least $h^* = C/\max_k s_k$ distinct objects fit into the cache at any time, and that FOO finds an optimal solution (\autoref{lma:mcfcorrectness}).
  Due to Assumption~\ref{ass:inf}, $M \rightarrow \infty$ implies that $C \rightarrow \infty$, and thus $h^* \rightarrow \infty$.
  So, the case where FOO caches only a finite number of objects requires that $h_i < h^*$.
  Whenever $h_i < h^*$ occurs, there cannot exist intervals that FOO could put into the cache.
  If any intervals could be put into the cache, FOO would cache them, due to its optimality.

  Our proof is by induction. We first show that the case of no intervals that could be put into the cache has zero probability, and then prove this successively for larger thresholds.
  
  \begin{proof}[Proof of \autoref{lma:hi:infinite}]
    We assume that $0 < h^* < M$ because $h^*\in \{0,M\}$ leads to trivial hit ratios~$\in \{0,1\}$.

    For arbitrary $i>N^*$ and any $x$ that is constant in $M$, we consider the event $X = \{h_i \leq x\}$.
    Furthermore, let $Z = \{h_i \leq z\}$ for any $0 \leq z \leq x$.
    We prove that $\Pb{X}$ vanishes as $M$ grows by induction over $z$ and corresponding $\Pb{Z}$.

    \autoref{fig:hi:fluid:apx} sketches the number of objects over a time interval including $i$.
    Note that, for any $z \leq x$, we can take a large enough $M$ such that $z < h^*$, because $h^* \rightarrow \infty$.
    So the figure shows $h^* > z$.
    The figure also defines the time interval $[u,v]$, where $u$ is the last time before $i$ when FOO cached $h^*$ objects, and $v$ is the next time after $i$ when FOO caches $h^*$ objects.
    So, for times $j\in (u,v)$, it holds that $h_j < h^*$.

\begin{figure}[h]
  \centering
  \vspace{1mm}
  \includegraphics[width=.7\linewidth]{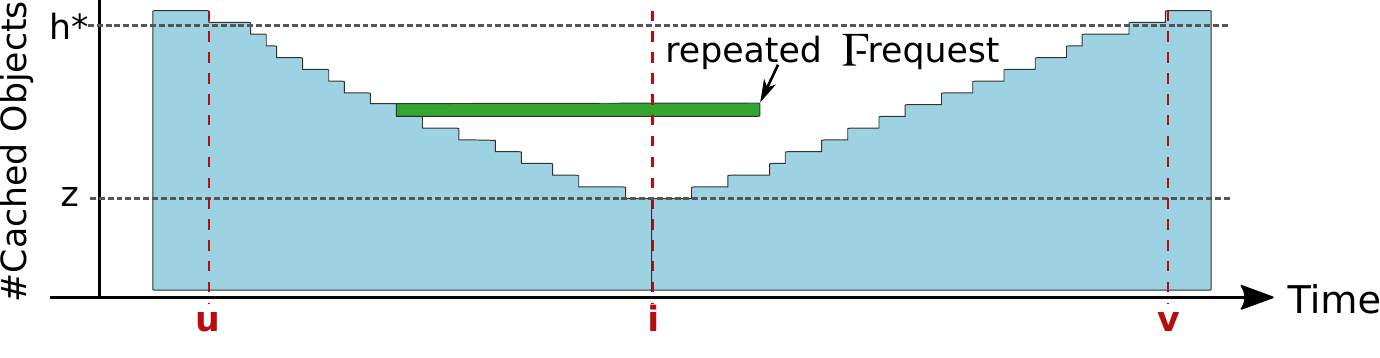}
  \caption{Sketch of the event $Z = \{h_i \leq z\}$, which happens with vanishing probability if $z$ is a constant with respect to $M$.
    The times $u$ and $v$ denote the beginning and the end of the current period where the number of cached objects is less than $h^*=C/\max_k s_k$, the fewest number of objects that fit in the cache.
    We define the set $\Gamma$ of objects that are requested in $(u,i]$.
    If any object in $\Gamma$ is requested in $[i,v)$, then FOO must cache this object (green interval).
    If such an interval exists, $h_i > z$ and thus $Z$ cannot happen.}
  \vspace{1mm}
\label{fig:hi:fluid:apx}
\end{figure}

\emph{Induction base: $z=0$ and event $Z = \{h_i \leq 0\}$.}
In other words, $Z$ means the cache is empty.
Let $\Gamma$ denote the set of distinct objects requested in the interval $(u,i]$.
Note that the event $Z$ requires that, during $(u,i]$, FOO stopped caching all $h^*$ objects.
Because FOO only changes caching decisions at interval boundaries, $\Gamma$ must at least contain $h^*$ objects.
Using the same argument, we observe that there happen at least $h^*$ requests to distinct objects in $[i,v)$.

A request to any $\Gamma$ object in $[i,v)$ makes $Z$ impossible.
Formally, let $A$ denote the event that any object in $\Gamma$ is requested again in $[i,v)$.
We observe that $A \Rightarrow \overline{Z}$ because any $\Gamma$-object that is requested in $[i,v)$ must be cached by FOO due to FOO-L's optimality (\autoref{lma:mcfcorrectness}).
By inverting the implication we obtain $Z \Rightarrow \overline{A}$ and thus $\Pb{Z} \leq \Pb{\overline{A}}$.

We next upper bound $\Pb{\overline{A}}$.
We start by observing that $\Pb{A}$ is minimized (and thus $\Pb{\overline{A}}$ is maximized) if all objects are requested with equal popularities.
This follows because, if popularities are not equal, popular objects are more likely to be in $\Gamma$ than unpopular objects due to the popular object's higher sampling probability (similar to the inspection paradox).
When $\Gamma$ contains more popular objects, it is more likely that we repeat a request in $[i,v)$, and thus $\Pb{A}$ increases ($\Pb{\overline{A}}$ decreases).

We upper bound $\Pb{\overline{A}}$ by assuming that objects are requested with equal probability $\rho_k = 1/M$ for $1\leq k \leq M$.
As the number of $\Gamma$-objects is at least $h^*$, the probability of requesting any $\Gamma$-object is at least $h^*/M$.
Further, we know that $v-i \geq h^*$ and so $$\Pb{\overline{A}} \leq \left(1-\frac{h^*}{M}\right)^{h^*}~~~.$$
We arrive at the following bound.
\begin{align}
  \Pb{\{h_i \leq z\}} = \Pb{Z} \leq  \Pb{\overline{A}} \leq \left(1-\frac{h^*}{M}\right)^{h^*} \longrightarrow 0 \quad \text{as } M \rightarrow \infty
\end{align}

\emph{Induction step: $z-1 \rightarrow z$ for $z \leq x$.}
We assume that the probability of caching only $z-1$ objects goes to zero as $M \rightarrow \infty$.
We prove the same statement for $z$ objects.

As for the induction base, let $\Gamma$ denote the set of distinct objects requested in the interval $(u,i]$, excluding objects in $H_i$.
We observe that $\vert \Gamma \vert \geq h^* -z$, following a similar argument.

We define $\Pb{A}$ as above and use the induction assumption.
As the probability of less than $z-1$ is vanishingly small, it must be that $h_i \geq z$.
Thus, a request to any $\Gamma$ object in $[i,v)$ makes $h_i = z$ impossible.
Consequently, $Z \Rightarrow \overline{A}$ and thus $\Pb{Z} \leq \Pb{\overline{A}}$.

To bound $\Pb{A}$, we focus on the requests in $[i,v)$ that do not go $H_i$-objects.
There are at least $h^*-x$ such requests.
$\Pb{A}$ is minimized if all objects, ignoring objects in $H_i$, are requested with equal popularities.
We thus upper bound $\Pb{\overline{A}}$ by assuming the condition requests happen to objects with equal probability $\rho_k = 1/(M-z)$ for $1\leq k \leq M-z$.
As before, we conclude that the probability of requesting any $\Gamma$-object is at least $(h^*-z)/(M-z)$ and we use the fact that there are at least $h^*-x$ to them in $[i,v)$.
\begin{align}
  \Pb{\{h_i \leq z\}} = \Pb{Z} & \leq  \Pb{\overline{A}}\\
                               & \leq \left(1-\frac{h^*-z}{M-z}\right)^{h^*-z}
\intertext{We then use that $z\leq x$ and that $x$ is constant in $M$.}
 & \leq \left(1-\frac{h^*-x}{M}\right)^{h^*-x} \longrightarrow 0 \quad \text{as } M \rightarrow \infty
\end{align}

In summary, the number of objects cached by FOO at an arbitrary time $i$ remains constant only with vanishingly small probability.
Consequently, this number grows to infinity with probability one.
\end{proof}

\subsection{Proof of Lemma~\ref{lma:ccpbound}}\label{sec:proof:lemma:ccpbound}

\begin{proof}[Proof of Lemma~\ref{lma:ccpbound}]
    
We consider the time $T_{b,q}$ to collect $b>1$ coupons, which have equal probabilities $q=(1/b,\dots,1/b)$.
To simplify notation, we set $T=T_{b,q}$ throughput this proof.

  We first transform our term using the exponential function, which is strictly monotonic.
\begin{align}
  \Pb{T \leq b \log b - c\,b} = \Pb{e^{-sT} \leq e^{-s(b \log b - c\,b)}}\quad \text{for all } s>0
\end{align}
We next apply the Chernoff bound.
\begin{align}
\Pb{e^{-sT} \leq e^{-s(b \log b - c\,b)}} \leq \Ex{e^{-sT}} \;e^{s(b \log b - c\,b)}
\end{align}

To derive $\Ex{e^{-sT}}$, we observe that $T=\sum_{i=1}^b T_i$, where $T_i$ is the time between collecting the $(i-1)$-th unique coupon and the $i$-th unique coupon.
As all $T_i$ are independent, we obtain a product of Laplace-Stieltjes transforms.
\begin{align}
\label{eq:prodLPtrans}   \Ex{e^{-sT}} = \prod_{i=1}^b \Ex{e^{-sT_i}}
\end{align}
We derive the individual transforms.
\begin{align}
  \Ex{e^{-sT_i}} & = \sum_{k=1}^\infty e^{-s\;k} p_i (1-p_i)^{k-1}\\
  & = \frac{p_i}{e^{s}+p_i-1}
\end{align}
We plug the coupon probabilities $p_i = 1-\frac{i-1}{b} = \frac{b-i+1}{b}$ into Eq.~\eqref{eq:prodLPtrans}, and simplify by reversing the product order.
\begin{align}
  \prod_{i=1}^b \Ex{e^{-sT_i}} = \prod_{i=1}^b \frac{(b-i+1)/b}{e^s+(b-i+1)/b-1} = \prod_{j=1}^b \frac{j/b}{e^s+j/b-1}
\end{align}

Finally, we choose $s=\frac{1}{b}$, which yields $e^s=e^{1/b}\geq 1+1/b$ and simplifies the product.
\begin{align}
  \prod_{j=1}^b \frac{j/b}{e^s+j/b-1} \leq  \prod_{j=1}^b \frac{j/b}{1/b+j/b} = \frac{1}{b+1}
\end{align}
This gives the statement of the lemma.
\begin{align}
  \Pb{T \leq b \log b - c\,b} \leq \frac{1}{b+1} \;e^{\frac{b \log b - c\,b}{b}} < e^{-c}
\end{align}

\end{proof}

\end{document}